\documentclass[final,5p,times,twocolumn]{elsarticle}

\usepackage{amssymb}
\usepackage{amsmath}
\usepackage{amsthm}
\usepackage{subfigure}
\usepackage{caption}
\usepackage{float}
\usepackage{braket}
\usepackage{dblfloatfix}
\usepackage[colorlinks=true, linkcolor=blue, citecolor=blue, urlcolor=blue]{hyperref}
\usepackage{dblfloatfix} 
\usepackage{amsthm}
\usepackage{mleftright}
\usepackage{tikz}
\usepackage{multicol}
\usetikzlibrary{arrows.meta, positioning}
\usetikzlibrary{mindmap}
\newtheorem{remark}{Remark}
\newtheorem{lem}{Lemma}
\newtheorem{result}{Result}
\newtheorem{conjecture}{Conjecture}

\newtheorem{openproblem}{Open problem}

\newtheorem{defn}{Definition} 
\newtheorem{fact}{Fact}
\newtheorem{theorem}{Theorem}

\allowdisplaybreaks

\usepackage{lmodern}

\journal{Theoretical Computer Science}

\begin{document}

\begin{frontmatter}



\title{A Hierarchy of Constant Communication Complexity}


\author[1]{Andris Ambainis }
\ead{andris.ambainis@lu.lv}
\author[2]{Hartmut Klauck}
\ead{hklauck@gmail.com}
\author[1]{Debbie Lim\corref{cor1}}
\ead{limhueychih@gmail.com}

\affiliation[1]{Center for Quantum Computing Science, Faculty of Sciences and Technology, University of Latvia, Latvia}
\affiliation[2]{Centre for Quantum Technologies, Singapore}
\cortext[cor1]{Corresponding author}

\begin{abstract}
Similarly to the Chomsky hierarchy, we offer a classification of communication complexity measures such that these measures are organized into equivalence classes. Different from previous attempts of this endeavor, we consider two communication complexity measures as equivalent, if, when one is constant, then the other is constant as well, and vice versa. Most previous considerations of similar topics have been using polylogarithmic input length as a defining characteristic of equivalence. In this paper, two measures ${\cal C}_1, {\cal C}_2$ are \emph{constant-equivalent}, if and only if for all total Boolean (families of) functions $f:\{0, 1\}^n\times\{0, 1\}^n\rightarrow \{0, 1\}$ we have ${\cal C}_1(f)=O(1)$ if and only if ${\cal C}_2(f)=O(1)$. We identify five equivalence classes according to the above equivalence relation. Interestingly, the classification is counter-intuitive in that powerful models of communication are grouped with weak ones, and seemingly weaker models end up on the top of the hierarchy.
\end{abstract}


\end{frontmatter}


\section{Introduction}
Communication complexity, as a subfield of theoretical computer science emerged from the seminal works of~\cite{yao1979some, babai1992multiparty, babai1986complexity, yao1983lower} more than four decades ago. Several other early works include Refs.~\cite{orlitsky1988communication, lovasz1989communication, lengauer1990vlsi}. In this field, we are able to show lower bound for explicit problems, with a plethora of applications to other computational models such as finite automata, Turing machines, decision trees, ordered binary decision diagrams, VLSI chips, streaming algorithms, networks of threshold gates and data structures~\cite{halldorsson2012streaming, babai1999superpolynomial, ablayev1996lower,babai1992multiparty, nisan1993communication, groger1991linear, chandra1983multi, alon1988meanders, babai1990lower, roychowdhury1994lower, goldmann1994communication, siu1995discrete, razborov1990applications, miltersen1994lower}. 

In the theory of formal languages, the Chomsky hierarchy~\cite{chomsky1956three} is a central concept. This hierarchy classifies formal languages according to the type of computational model that can recognize them. The levels of the Chomsky hierarchy can be labelled by a machine model that is able to recognize the languages from the corresponding level, i.e., Finite Automata, Context-Free Grammars, Context-Sensitive Grammars, and Turing Machines (acting as the most powerful formal computation model). The Chomsky hierarchy also classifies formal languages into one of its levels. For instance, the language $\{0^n1^n\}$ (understood as a family over all $n$) is in the class of context-free languages, but not the class of regular languages (as defined by DFA's). Nondeterministic and deterministic finite automata on the other hand define the same class of regular languages, even though their state-complexity can be exponentially far apart.

In this paper, we attempt a similar classification for communication complexity. Technically, communication complexity is usually established for a \emph{family} of functions $\{0, 1\}^n \times \{0, 1\}^n \rightarrow \{0, 1\}$  with all possible input lengths. The complexity of such a family is \emph{constant} if the needed communication does not increase with the input length. A trivial example of this is the decision if the first bits of Alice's and Bob's input are both 1, which a deterministic protocol can decide with 1 bit of communication. A less trivial problem is the Equality function, in which Alice and Bob each receive a string of $n$ bits and their task is to decide whether the two strings are the same. While a deterministic protocol cannot do better than communicating the whole length of a string, a public coin protocol can solve this problem with a single bit of communication and satisfactory error probability.

Prior to our work, there have been several studies on establishing hierarchies in communication complexity. In a seminal paper, Babai, Frankl, and Simon~\cite{babai1986complexity} proposed polylogarithmic complexity as the criteria for efficiency and used it to define communication classes analogous to the classical computational complexity classes. Their definition of communication classes provides a structured framework for comparing the powers of various communication models. While this work has undoubtedly led to the development in proving separations between different communication classes, it is also crucial to also consider communication problems that have uniformly bounded complexity. In fact, the idea of constant-cost communication classes began with the works of Hambardzumyan et al.~\cite{hambardzumyan2023dimension} and Harms et al.~\cite{harms2022randomized} studying constant-cost communication problems and the notion of ``classes" was used in prior works such as Refs.~\cite{harms2024randomized} and~\cite{fang2024no}. Other works on constant-cost communication include~\cite{linial2007complexity, linial2009learning, bhrushundi2014property, hatami2020sign, harms2019universal,  hatami2022lower, esperet2022sketching,  cheung2023separation, ahmed2023communication}. The recent survey paper of Hatami and Hatami~\cite{hatami2024structure} studies the constant-cost analogues of the communication classes by Babai, Frankl, and Simon~\cite{babai1986complexity}, where the criterion of effectiveness is a $O(1)$ complexity, independent of the input size $n$. Another related work is that of G{\"o}{\"o}s, Pitassi and Watson~\cite{landscape}. In their work, the authors outlined the relationships between a wide range of models for two-party communication, clarifying which are stronger, weaker, or incomparable. By establishing new separations and inclusions, they resolve several long-standing open questions and give a nearly complete structural picture of the complexity landscape in this area.

\paragraph{Main contribution} We establish a hierarchy of increasingly powerful classes of communication complexity, based on what these computational models can do with constant communication. While our investigation is based on the constant-equivalence of different measures of communication complexity (for instance nondeterministic and deterministic communication are constant-equivalent),
this also leads to a hierarchy of (families of) communication problems. Namely, the problems that can be solved by an equivalence class of communication complexity measure. For instance, the Equality problem is in the class of cheap public coin protocols (see below), while the Greater-Than problem is in the class of low sign-rank protocols~\cite{hatami2020sign}.

We identify the following levels of our hierarchy, named here by the simplest way to define them (more details in Figure~\ref{fig:summary_of-results}):
\begin{itemize}
    \item Class 1: Only a constant number of distinct rows in the communication matrix.
    \item Class 2: The $\gamma_2$ norm of the communication matrix is constant.
    \item Class 3: The approximate $\gamma_2$ norm, $\gamma_2^\alpha$, of the communication matrix is constant which implies that information complexity and also the commonly used randomized/quantum communication complexity lower bounds such as the rectangle bound (with error) and the discrepancy bound are also constant.
    \item Class 4: The sign-rank and the (truly) unbounded error communication complexity are constant. 
    \item Class 5: Measures defined in terms of product-distribution communication complexity are constant.
\end{itemize}

We provide equivalence results for all the measures in the respective classes, and separating functions for the different levels of the hierarchy. While most of this work is collecting existing results, we have the following new result. 
\begin{result}
Let $f$ be a total Boolean function and denote by $R^{\operatorname{pub}}_{\epsilon}(f)$ the public-coin randomized communication complexity with oblivious and exact error (see Section~\ref{sec:class_2}). Then we have $$\log \gamma_2 \leq R^{\operatorname{pub}}_{=\epsilon}(f) \leq O((\gamma_2(f))^2).$$
\end{result}

We summarize our results in Figure~\ref{fig:summary_of-results}. 
\begin{figure*}[t]
\centering
\begin{tikzpicture}[node distance=0pt]
  \node[draw, rectangle, minimum width=6cm, align=center] (r1) {\textbf{\underline{Class 5}} \\ $D^{[]}_\epsilon(f)$,  $D^{A \rightarrow B, []}_\epsilon(f)$, $R^{A \rightarrow B, []}_\epsilon(f)$, $Q^{A \rightarrow B, []}_\epsilon(f)$, $\operatorname{rec}^{A \rightarrow B, []}_\epsilon(f)$, $\operatorname{sub}^{A \rightarrow B, []}_{\mathcal Y}(f, \epsilon)$,  $VC(f)$, \\
  $sq(f)$, $\operatorname{disc}^{[]}(f)$, $IC^{[]}(\pi)$, $IC^{[]}_\epsilon(f)$, $IC^{\operatorname{ext},[]}(\pi)$, $IC^{\operatorname{ext},[]}_\epsilon(f)$, $pIC^{\infty,[]}_\epsilon(f)$, $AC^{[]}_\epsilon(f)$, \\
  $\operatorname{rdisc}^{[]}_\epsilon(f)$, $\operatorname{ardisc}^{[]}_\epsilon(f)$, $\operatorname{prt}^{[]}_\epsilon(f)$, $\operatorname{prt}^{+, []}_\epsilon(f)$, $\bar{\operatorname{prt}}^{[]}_\epsilon(f)$, $\operatorname{wprt}^{[]}_\epsilon(f)$, $\operatorname{wreg}_\epsilon(f)$};

  \node[draw, rectangle, below=of r1, minimum width=6cm, align=center] (r2) {\textbf{\underline{Class 4}} \\ $UPP(f)$, $\operatorname{signrank}(f)$};

  \node[draw, rectangle, below=of r2, minimum width=6cm, align=center] (r3) {\textbf{\underline{Class 3}} \\ $R^{A \rightarrow B, \operatorname{pub}}_\epsilon(f)$, $ \gamma^\infty_2(f)$, $\gamma^\alpha_2(f)$, $\operatorname{disc}^{\mu}(f)$, $\operatorname{disc}(f)$, $mc(f)$, $PP^{\operatorname{pub}}(f)$, 
  $\operatorname{rec}_\epsilon(f)$, $\widetilde{\operatorname{rec}}_\epsilon(f)$,\\ $\operatorname{srec}_\epsilon(f)$, $\widetilde{\operatorname{srec}}_\epsilon(f)$, $\bar{\operatorname{prt}}_\epsilon(f)$,
  $\operatorname{prt}_\epsilon(f)$, $\operatorname{sub^{A \rightarrow B}}_{\mathcal Y, \epsilon}(f)$, 
  $D^\mu_\epsilon(f)$, $MA(f)$, $SBP(f)$, $\operatorname{ment}_\epsilon(f)$, \\
  $\operatorname{rcment}_{\epsilon, \epsilon}(f)$, $Q^*_\epsilon(f)$, $QIC_{D, \epsilon}(f)$, $QIC_\epsilon(f)$, $AQCC_\epsilon(f)$, $ R^{\operatorname{pub}}_\epsilon(f)$, $IC^{\mu}_\epsilon(f)$, $IC_\epsilon(f)$, 
  $IC_{D., \epsilon}(f)$, \\
  $IC^{\operatorname{ext}}_\epsilon(f)$, $IC^{\operatorname{ext}}_{D, \epsilon}(f)$, $AC^{\mu}_\epsilon(f)$,  
  $\operatorname{wreg}^{\mu}(f)$, $\operatorname{wprt}^\mu_\epsilon(f)$, $\operatorname{prt}^{+, \mu}_\epsilon(f)$, $\bar{\operatorname{prt}}^{\mu}_\epsilon(f)$,\\
  $\operatorname{prt}^{\mu}_\epsilon(f)$, $\operatorname{rdisc}^\mu_\epsilon(f)$, 
  $\operatorname{ardisc}^{\mu}_\epsilon(f)$, $\operatorname{pprt}^\mu_\epsilon(f)$, $\operatorname{pprt}_\epsilon(f)$, $pIC^{\infty, \mu}_\epsilon(f)$, 
  $pIC^{\infty}_\epsilon(f)$, 
  $OMA^{A \rightarrow B}_\epsilon(f)$,\\ 
  $OIP^{A \rightarrow B}_\epsilon(f)$, $OIP^{A \rightarrow B}_{+, \epsilon}(f)$};

  \node[draw, rectangle, below=of r3, minimum width=6cm, align=center] (r4) {\textbf{\underline{Class 2}} \\ $\gamma_2(f)$, $R^{\operatorname{pub}}_{=\epsilon}(f)$, $Q^*_{=\epsilon}(f)$};

  \node[draw, rectangle, below=of r4, minimum width=6cm, align=center] (r5) {\textbf{\underline{Class 1}} \\ $D(f)$, $N(f)$, $N^0(f)$, $N^1(f)$, $C^0(f)$, $C^1(f)$, $C(f)$, $C^*_0(f)$, $C^D(f)$, $C^P(f)$, $R^{\operatorname{priv}}_\epsilon(f)$,\\ $R^{\operatorname{pub}}_0(f)$, $\mathring R_0(f)$, $Q_0(f)$, $Q^*_0(f)$, $Q^c_0(f)$, $\mathring Q_0(f)$, $\operatorname{rank}(f)$, $\operatorname{rank}_+(f)$, $\operatorname{signrank}_+(f)$, \\$\operatorname{rank}_\alpha(f)$, $\operatorname{rank}_{\alpha, +}(f)$, $ \oplus P(f)$, $ZAM_{c, s}(f)$, $AM^{\operatorname{priv}}_{c, s}(f)$, $UAM_{c, s}(f)$, $OIP^{[4]}_{c, s}(f)$, $D^{A \rightarrow B}(f)$, \\$Q^{A \rightarrow B}_0(f)$, $Q^{A \rightarrow B, *}_0(f), LV^{A \rightarrow B}(f)$};
\end{tikzpicture}
\caption{Summary of results. We categorize communication complexity measures in to five increasingly powerful classes.}
\label{fig:summary_of-results}
\end{figure*}


\paragraph{Outline} This paper is organized as follows: In Section~\ref{sec:preliminaries}, we define the communication model, the necessary notations and results that will be used to prove our main theorems. We present our main results in Section~\ref{sec:main_results}, which include graphs illustrating the constant equivalence relationship among all complexity measures in every class, followed by separation results between consecutive classes in Section~\ref{sec:separation}. In Section~\ref{sec:conclusion}, we conclude our work, discuss some open problems and make several conjectures. 

\section{Preliminaries and notations}\label{sec:preliminaries}
We consider Yao's two-party communication complexity model \cite{yao1979some}. Let $\mathcal X, \mathcal Y, \mathcal Z$ be arbitrary finite sets and let $f:\mathcal X \times \mathcal Y \rightarrow  \mathcal Z$ be an arbitrary function. Two players, Alice and Bob, who are assumed to have unlimited computational power, are each given inputs $x \in \mathcal X$ and $y \in \mathcal Y$ respectively. Each player knows only their input and the task is to evaluate $z = f(x, y)$ by alternately send messages to each other. When $f\subseteq \mathcal X \times \mathcal Y \times \mathcal Z$ is a relation, the task translates into determining an element $z \in \mathcal Z$ such that $(x,y,z) \in f$.

We let $X, Y$ denote the random variables from $\mathcal X , \mathcal Y$ respectively. In direct sum/product theorems, we define $f^k: \mathcal X ^k \times \mathcal Y^k \rightarrow \mathcal Z^k$ to be the concatenation of the evaluations $$f^k(x_1, \cdots, x_k, y_1, \cdots, y_k) \coloneq (f(x_1, y_1), \cdots, f(x_k, y_k)).$$ We write $\mu^k$ to denote the distribution on $k$ inputs, where each is sampled according to $\mu$ independently. 

All functions considered in this work are total Boolean functions $f:\{0, 1\}^n\times \{0, 1\}^n\rightarrow \{0,1\}$, unless stated otherwise. Sometimes, we use $f: \{-1, 1\}^n \times \{-1, 1\}^n \rightarrow \{-1, 1\}$ instead. For a function  $f$, let $M_f$ denote its corresponding communication matrix, i.e. $M_{x,y}=f(x,y)$ for all $x,y\in\{0, 1\}^n$. We use $f$ instead of $M_f$ as the argument of the complexity measure when there is no ambiguity. For any total Boolean function $f$, we use $D(f), N^1(f), N^0(f), R^{\operatorname{priv}}_\epsilon(f), R^{\operatorname{pub}}_\epsilon(f)$ to denote the deterministic communication complexity, nondeterministic communication complexity, co-nondeterministic communication complexity, private coin $\epsilon$-error randomized communication complexity and public coin $\epsilon$-error randomized communication complexity.  Readers are referred to Ref.~\cite{KNisan96} for background on classical communication complexity. For quantum communication complexity measures, we write $Q_0(f)$, $Q_\epsilon( f )$ to denote the exact (zero-error) quantum communication complexity and the $\epsilon$-error quantum communication complexity respectively, both without prior entanglement. We denote by $*$ in superscript when players share prior entanglement. For all communication complexity measures, we write $A\rightarrow B$ in superscript to denote the one-way communication model.

To distinguish between constant and non-constant complexity in a non-uniform model of computation like communication complexity,
we need to consider \emph {families} of functions $f_1, f_2, \ldots, f_n, \ldots$ with one function $f_n:\{0, 1\}^n\times \{0, 1\}^n\rightarrow \{0,1\}$ for every input size $n$. For example, the Equality problem is defined by taking a family of functions $EQ_1, EQ_2, \ldots, EQ_n, \ldots$ where $EQ_n(x, y)$ is a function from $\{0, 1\}^n\times \{0, 1\}^n$ to $\{0,1\}$ defined by $EQ_n(x, y)=1$ iff $x=y$. Then, the public-coin randomized communication complexity of the Equality problem is $O(1)$, which means that this problem can be solved with constant communication in the described model, \emph{no matter what the length of the input is}. On the other hand, the Disjointness problem cannot be solved in any model of communication complexity using constant communication, unless \emph{all} problems can be (under reasonable assumptions) solved in that model at constant cost.

We give a quick review of some communication complexity measures, starting with nondeterministic communication complexity. Let $\pi$ be a protocol and $v$ be a node of the protocol tree. The set $R_v$ of all inputs $(x, y)$ that lead $\pi$ to $v$ is a \emph{combinatorial rectangle}, where $R_v = A_v \times B_v$ for $A_v\subseteq \mathcal X$ and $B_v\subseteq \mathcal Y$. For every $z\in\{0,, 1\}$, we say that a combinatorial rectangle $R$ is a $z$-monochromatic rectangle if $f(x, y) = b$ for all $(x, y)\in R$. If $v$ is a leaf with the label $b$, then $R_v$ must be a  $z$-monochromatic rectangle. Hence, the leaves of a deterministic protocol partitions its communication matrix into monochromatic rectangles. The protocol partition number of $f$, denoted as $C^P(f)$, is the minimum number of leaves in a protocol for $f$.  The partition number of $f$, denoted as $C^D(f)$, is the minimum number of monochromatic rectangles in a partition of $\mathcal X\times \mathcal Y$. Moreover, the cover number of $f$, denoted as $C(f)$, is the minimum number of monochromatic rectangles required to cover $\mathcal X\times \mathcal Y$. For any $z\in\{0, 1\}$, $C^z(f)$ is the number of monochromatic rectangles required to cover the $z$-inputs of $f$~\cite{KNisan96}. 

\paragraph{Variants of rank and factorization norm.} Let $M\in\mathbb R^{n\times n}$ be a matrix. Denote by $\operatorname{rank}(M)$ the linear \emph{rank} of $M$  over the field of reals. The \emph{$\alpha$-approximate rank} of $M$, denoted as $\operatorname{rank}_\alpha(M)$, is defined as follows{~\cite[Definition~1]{Gal2021},~\cite{Krause1996}} 
\begin{align*}
    \operatorname{rank}_\alpha(f) \coloneqq \min_{M': 1\leq M_{i, j} \cdot M'_{i, j}\leq \alpha} \operatorname{rank}(M'), 
\end{align*}
where $\alpha \geq 1$. The \emph{nonnegative rank} of a matrix $M\in\mathbb{R}_+^{n\times n}$, denoted as $\operatorname{rank}_+(M)$, is the smallest $r$ such that there exist matrices $A\in\mathbb{R_+}^{n\times r}$ and $B\in\mathbb{R}_+^{r\times n}$ such that $M=AB$. 
The \emph{$\alpha$-approximate nonnegative rank} of $M$, denoted as $\operatorname{rank}_{\alpha, +}(M)$ is given by 
\begin{align*}
    & \operatorname{rank}_{\alpha, +}(M) \\
    & \coloneqq \min \left\{\operatorname{rank}_+(M^\prime): M^\prime \text{ is nonnegative, }\lVert M-M^\prime\rVert_\infty\leq\alpha\right\}.
\end{align*}
In certain cases, it makes sense to consider functions with $\pm1$ outputs, i.e. $f:\mathcal \{0, 1\}^n\times \{0, 1\}^n\rightarrow \{-1, 1\}$~\cite{linial2007complexity}. sign-rank is one of the important analytic notions in communication complexity. The \emph{sign-rank} of a matrix $M\in\{-1, 1\}^{n\times n}$, denoted as $\operatorname{signrank(M)}$, is  the minimal rank of a real matrix whose entries have the same sign pattern as $M$. More formally. 
\begin{align*}
    & \operatorname{signrank}(M) \\
    & \coloneqq \min\{\operatorname{rank}(M'): M'\in\mathbb{R}^{n\times n}, M'_{ij}M_{ij}> 0 \text{ for all } i, j\in[n]\}. 
\end{align*}
The \emph{nonnegative sign-rank} of a matrix $M\in\mathbb{R}_+^{m\times n}$, denoted as $signrank_+(M)$, is the smallest $r$ such that there exists $A\in\mathbb{R}_+^{m\times r}$ and $B\in\mathbb{R}_+^{r\times n}$ such that $AB$ has the same sign pattern as $M$.  
 
The $\gamma_2$ norm is an important matrix norm initially developed in Banach Space theory. Linial and Shraibman~\cite{linial2007lower} introduced this norm in the context of communication complexity and subsequently, this complexity measure and its approximate version found applications in communication complexity and other related fields such as discrepancy theory~\cite{matouvsek2020factorization} and differential privacy~\cite{muthukrishnan2012optimal, edmonds2020power, henzinger2022constant}. Formally, the \emph{$\gamma_2$ norm} of a matrix $M\in\mathbb{R}^{m\times n}$, denoted as $\gamma_2(M)$, is defined as  
\begin{align}\label{eqn:gamma_2}
    \gamma_2(M) = \min_{A, B: AB=M}\lVert A\rVert_r\lVert B\rVert_c,
\end{align}
where $\lVert A\rVert_r$ denotes the maximum $\ell_2$ norm of a row of $A$ and  $\lVert B\rVert_c$ denotes the maximum $\ell_2$ norm of a column of $B$. 
Alternatively, $\gamma_2(M)$ can be defined as 
$$\gamma_2(M) = \max_{u, v: \Vert u\Vert_2 = \Vert v\Vert_2 = 1} \Vert uv^T\circ M\Vert_{tr},$$
where $\circ$ and $\lVert \cdot \rVert_{tr}$ denotes the Hadamard product and trace norm respectively. 

The \emph{$\alpha$-approximate $\gamma_2$ norm} of $M\in\{-1, 1\}^{n\times n}$, where $\alpha\geq 1$, denoted as $\gamma_2^\alpha(M)$, is given by
\begin{align*}
    \gamma_2^\alpha(M) = \min_{\{M':1\leq M_{ij}M'_{ij}\leq \alpha\}} \gamma_2(M').
\end{align*}
In particular, 
\begin{align*}
    \gamma_2^\infty(M) = \min_{\{M':1\leq M_{ij}M'_{ij}\}} \gamma_2(M').
\end{align*}

\paragraph{Information complexity.} Before the notion of information complexity was formally defined, information-theoretic tools have already been used to prove communication complexity bounds. For example, Ablayev~\cite{ablayev1996lower} used information theory to prove a lower bound for the Index problem, Nayak~\cite{Nayakthesis} used it in the quantum setting, and Klauck et al.~\cite{klauck2001interaction}
used it for pointer jumping.  Interactive information complexity has been studied by information theorists starting with Orlitsky’s work in the 1990s~\cite{orlitsky2002worst}. A paper from this line of work by Ma and Ishwar~\cite{ma2011some} gives an alternative way of looking at the interactive information complexity of problems and amortized communication.

 From the perspective of information complexity, communication can be thought of as players wanting to reveal the minimum amount of information about their inputs in order to solve some pre-determined task. The internal information complexity measures the amount of information revealed to each player about the other player's input, whereas, the external information complexity measures the amount of information revealed about the inputs $x, y$ to an external observer who only observes the transcript of the protocol. Intuitively, both of these measures should be lower bounds on the communication complexity of the protocol. 

Given a protocol $\pi$, $T(X, Y)$ denotes the random variable of the transcript on inputs $X, Y$, which are also random variables. The \emph{(internal) information cost} of a protocol $\pi$ over all inputs from $\mathcal X\times \mathcal Y$ is given by
\begin{align*}
    IC^{\mu}(\pi)\coloneqq I(T(X, Y);X|Y) + I(T(X, Y);Y|X), 
\end{align*}
where $I(A;B|C)$ denotes the mutual information between random variables $A, B$, conditioned on $C$. A second measure of the information complexity of a communication protocol is its external information cost. The \emph{external information cost} of a protocol $\pi$ over all inputs from $\mathcal X\times \mathcal Y$ is given by
\begin{align*}
     IC^{\operatorname{ext}}_\mu(\pi) \coloneqq I(XY;T(X, Y)).
\end{align*}
For a function $f$, the external information cost of f with respect to $\mu$ is defined to be the infimum of $IC^{\operatorname{ext}}_\mu(\pi)$ over all protocols $\pi$ that compute $f$ on inputs from $\mathcal X \times \mathcal Y$ with probability larger than 2/3~\cite{BarakBCR10}. Clearly, this implies that the external information cost is always smaller than the communication complexity~. Moreover, the external information cost is always greater than or equal to the internal information cost of a protocol~\cite{braverman2012interactive}.

\section{Simulation results}\label{sec:main_results}
In this section, we state our main results. We divide this section into five subsections, one subsection for every class. Each subsection starts with a graph showing the network of simulation results for the corresponding class, followed by a main theorem and a list of facts and lemmas that are used in the main theorem.

The graphs in the subsequent subsections use the following notations. For any two complexity measures $\mathcal C_1, \mathcal C_2$, we use the arrow $\rightarrow$ to denote the ``$\leq$" relation, i.e.,$\mathcal C_1(f)\rightarrow \mathcal C_2(f)$ denotes $\mathcal C_1(f)\leq O(g(\mathcal C_2(f)))$ for some function $g:\mathbb R\rightarrow \mathbb R$.  Moreover, we use $\leftrightarrow$ between $\mathcal C_1, \mathcal C_2$ to denote that both complexity measures are equivalent, i.e. $\mathcal C_1\leftrightarrow \mathcal C_2$ denotes $\mathcal C_1(f) = \Theta(h(\mathcal C_2(f)))$ for some function $h:\mathbb R\rightarrow \mathbb R$. 

\subsection{Class 1}\label{sec:useful_1}
This class relates the constant equivalence between deterministic communication complexity, nondeterministic communication complexity, several variants of randomized communication complexity, quantum communication complexity, variants of Arthur-Merlin complexities, rank, parity complexity and  sampling complexity. The constant equivalence relationship between complexity measures in Class 1 is summarized in Figure~\ref{fig:C1}.
\begin{figure*}[!htbp]
\centering
\begin{tikzpicture}[->,>=Stealth,auto,node distance=3cm, thick]
\node (Rpriv) {$R^{\operatorname{priv}}_{\epsilon}(f)$};
\node (rankalphaPlus) [right of=Rpriv, xshift=1cm] {$\operatorname{rank}_{\alpha, +}(f)$};
\node (rankalpha) [right of=rankalphaPlus] {$\operatorname{rank}_{\alpha}(f)$};
\node (rank) [right of=rankalpha] {$\operatorname{rank}(f)$};
\node (QE) [right of=rank, xshift=1.5cm] {$Q_0(f)$};
\node (D1) [below of=rank] {$D^{A \rightarrow B}(f)$};
\node (QAB0) [below left of=D1] {$Q^{A \rightarrow B}_0(f)$};
\node (LV1) [below of=D1, yshift=-2cm] {$LV^{A \rightarrow B}(f)$};
\node (QAB0star) [below right of=D1] {$Q^{A \rightarrow B, *}_0(f)$};
\node (QEstar) [below of=QE, yshift=-13cm] {$Q^*_0(f)$};
\node (C) [below of=Rpriv, yshift=-13cm] {$C(f)$};
\node (CD) [right of=C] {$C^D(f)$};
\node (CP) [right of=CD] {$C^P(f)$};
\node (D) [right of=CP] {$D(f)$};
\node (N) [above right of=C, xshift=0.5cm, yshift=3cm] {$N(f)$};
\node (Pplus) [below left of=rank] {$\oplus P(f)$};
\node (ringQ) [below right of=rank] {$\mathring{Q}_0(f)$};
\node (Rpub0) [above right of=D] {$R^{\operatorname{pub}}_0(f)$};
\node (ringR) [above left of=D] {$\mathring{R}_0(f)$};
\node (Qc0) [below left of=QE, xshift=1cm] {$Q^c_0(f)$};
\node (C*0) [above left of=QEstar, xshift=1cm] {$C^*_0(f)$};
\node (signrank+) [below of=Rpriv] {$\operatorname{signrank}_+(f)$};
\node (signrank+') [below right of=Rpriv, xshift=3.3cm, yshift=0.5cm] 
{$\operatorname{signrank}_+(\neg f)$};
\node (N1) [below of=signrank+] {$N^1(f)$/$N^0(\neg f)$};
\node (N0) [below of=signrank+'] {$N^0(f)$/$N^1(\neg f)$};
\node (rank+) [below of=N1] {$\operatorname{rank}_+(f)$};
\node (rank+') [below of=N0, yshift=-1cm] {$\operatorname{rank}_+(\neg f)$};
\node (C1) [below of=rank+] {$C^1(f)$};
\node (C0) [below of=rank+'] {$C^0(f)$};
\node (AM) [below right of=Rpriv, xshift=1cm, yshift=-1cm] {$AM^{\operatorname{priv}}_{c, s}(f)$};
\node (UAM) [below of=AM, yshift=-1cm] {$UAM_{c, s}(f)$};
\node (ZAM) [below of=UAM, yshift=1cm] {$ZAM_{c, s}(f)$};
\node (OIP4) [below left of=AM, xshift=0.5cm, yshift=0.2cm] {$OIP^{[4]}_{c, s}(f)$};
\draw[->] (D) -- node[midway, above ] {Facts~~\ref{det_rank_UB},~\ref{fact:Q_E^*-rank(f)}} (QEstar);
\draw[->] (C) -- node[midway, above] {Fact~\ref{fact:prop_2.2}} (CD);
\draw[->] (CD) -- node[midway, above] {Fact~\ref{fact:prop_2.2}} (CP);
\draw[->] (CP) -- node[midway, above] {Fact~\ref{fact:prop_2.2}} (D);
\draw[->] (QEstar) -- node[midway, left] {Fact~\ref{fact:simple_rs}} (QE);
\draw[->] (QE) -- node[midway, above] {Facts~\ref{fact:simple_rs},~\ref{det_rank_UB}} (rank);
\draw[->] (rank) -- node[midway, above  ] {Fact~\ref{apprank_rank}} (rankalpha);
\draw[->] (rankalpha) -- node[midway, above] {Fact~\ref{easy}} (rankalphaPlus);
 \draw[->] (rankalphaPlus) -- node[midway, above] {Fact~\ref{fact:R^priv_epsilon(f)-rank_alpha_+}} (Rpriv);
 \draw [->] (D1) -- node[midway, right] {} (LV1);
 \draw [->] (LV1) -- node[near start, right] {Facts~\ref{fact:D1-LV1},~\ref{fact:LV1-D1}} (D1);
\draw[->] (Pplus) -- node[midway, left] {Fact~\ref{fact:oplusp-rank}} (rank);
\draw[->](rank) -- node[midway, left]{} (Pplus);
\draw [->] (D) -- node[midway, right] {Fact~\ref{fact:Rpub0-D}} (Rpub0);
\draw [->] (Rpub0) -- (D);
\draw [->] (ringQ) -- node[midway, right] {Fact~\ref{fact:mathringQ}}(rank);
\draw [->] (rank) -- node[midway, right]{} (ringQ);
\draw [->] (ringR) -- node[midway, left] {Fact~\ref{fact:ringR}} (D);
\draw [->] (D) -- node[midway, left] {} (ringR);
\draw [->] (QE) -- node[midway, left] {Fact~\ref{fact:simple_rs}}  (Qc0);
\draw [->] (Qc0)  --node[midway, left] {} (QE);
\draw [->] (C*0) -- node[midway, left] {Fact~\ref{fact:simple_rs}} (QEstar);
\draw [->] (QEstar) -- node[midway, left] {} (C*0);
\draw [->] (C) -- node[near end, right] {Fact~\ref{fact:N^Z(f)=logC^z(f)}} (N);
\draw [->] (N) -- node[midway, right] {} (C);
\draw [->] (D1) to node[near start, left] {Fact~\ref{fact:rank-D1}} (rank);
\draw [->] (rank) to node[midway, left] {} (D1);
\draw [->] (D1) to node [midway, left] {Fact~\ref{fact:QAB0-D1}} (QAB0);
\draw [->] (QAB0) to node[midway, right] {} (D1);
\draw [->] (D1) to node [midway, right] {Fact~\ref{fact:QAB*-D1}} (QAB0star);
\draw [->] (QAB0star) to node[midway, right] {} (D1);
\draw [->] (Rpriv) to node[midway, right] {Fact~\ref{fact:Rpriv-AM}} (AM);
\draw [->] (AM) to node[near end, left] {Fact~\ref{fact:AM-UAM-ZAM}} (UAM);
\draw [->] (UAM) to node[midway, left, yshift=0.5cm] {Fact~\ref{fact:AM-UAM-ZAM}} (ZAM);
\draw [->] (Rpriv) to node[midway, right] {Facts~\ref{fact:D(f)-N^0(f)-N^1(f)},~\ref{fact:signrank'}} (signrank+');
    \draw [->] (Rpriv) to node[align=center, near end, right] {Facts\\\ref{fact:D(f)-N^0(f)-N^1(f)},~\ref{fact:signrank+-N(f)}} (signrank+);
    \draw [->] (signrank+') to node[midway, right] {Fact~\ref{fact:signrank'}} (N0);
    \draw [->] (signrank+) to node[midway, left] {Fact~\ref{fact:signrank+-N(f)}} (N1);
    \draw [->] (C1) to node[near start, right] {Fact~\ref{fact:C(f)=C^0(f)+C^1(f)}} (C);
    \draw [->] (C0) to node[midway, right] {Fact~\ref{fact:C(f)=C^0(f)+C^1(f)}} (C);
    \draw [->] (ZAM) to node[near start, right] {Fact~\ref{fact:zam_LB}} (N0);
    \draw [->] (N1) to node[midway, right] {Fact~\ref{fact:N1-rank+}} (rank+);
    \draw [->] (N0) to node[right] {Fact~\ref{cor:N0-rank+'}} (rank+');
    \draw [->] (rank+') to node[align=center, near start, right] {Facts~\ref{cor:rank+'-D'},~\ref{fact:N^Z(f)=logC^z(f)},~\ref{fact:D(f)-N^0(f)-N^1(f)}} (C0);
    \draw [->] (rank+) to node[align=center, near start, right] {Facts~\ref{fact:rank+-D},~\ref{fact:D(f)-N^0(f)-N^1(f)},~\ref{fact:N^Z(f)=logC^z(f)}} (C1);
    \draw [->] (AM) to node[midway, left] {Fact~\ref{fact:AM=OIOP4}} (OIP4);
    \draw [->] (OIP4) to node[midway, left] {} (AM);
\end{tikzpicture}
\caption{Relationship between complexity measure in Class 1.} 
\label{fig:C1}
\end{figure*}

\begin{theorem}\label{thm:1}
Class 1 includes the following complexity measures: 
\begin{multicols}{3}
    \begin{itemize}
        \item $D(f)$
        \item $N(f)$
        \item $N^0(f)$
        \item $N^1(f)$
        \item $C^0(f)$
        \item $C^1(f)$
        \item $C(f)$
        \item $C^*_0(f)$
        \item $C^D(f)$
        \item $C^P(f)$
        \item $R^{\operatorname{priv}}_\epsilon(f)$
        \item $R^{\operatorname{pub}}_0(f)$
        \item $\mathring R_0(f)$
        \item $Q_0(f)$
        \item $Q^*_0(f)$
        \item $Q^c_0(f)$
        \item $\mathring Q_0(f)$
        \item $\operatorname{rank}(f)$
        \item $\operatorname{rank}_+(f)$
        \item $\operatorname{signrank}_+(f)$
        \item $\operatorname{rank}_\alpha(f)$
        \item $\operatorname{rank}_{\alpha, +}(f)$
        \item $ \oplus P(f)$
        \item $ZAM_{c, s}(f)$
        \item $AM^{\operatorname{priv}}_{c, s}(f)$
        \item $UAM_{c, s}(f)$
        \item $OIP^{[4]}_{c, s}(f)$
        \item $D^{A \rightarrow B}(f)$
        \item $LV^{A \rightarrow B}(f)$
        \item $Q^{A \rightarrow B}_0(f)$
        \item $Q^{A \rightarrow B, *}_0(f)$
    \end{itemize}
\end{multicols}
\end{theorem}
\begin{proof}  
We prove the cyclic relationship of the form $\mathcal C_1\rightarrow \mathcal C_2\rightarrow \mathcal C_3\rightarrow \cdots \rightarrow \mathcal C_1$ for all complexity measures in this class. 

We have $C(f)\rightarrow C^D(f)$ by Fact~\ref{fact:prop_2.2}; $C^D(f)\rightarrow C^P(f)$ by Fact~\ref{fact:prop_2.2}; $C^P(f)\rightarrow D(f)$ by Fact~\ref{fact:prop_2.2}; $D(f) \rightarrow Q^*_0(f)$ by Facts~~\ref{det_rank_UB},~\ref{fact:Q_E^*-rank(f)}; $Q^*_0(f)\rightarrow Q_0(f)$ by Fact~\ref{fact:simple_rs}; $Q_0(f)\rightarrow \operatorname{rank}(f)$ by Facts~\ref{fact:simple_rs},~\ref{det_rank_UB}; $\operatorname{rank}(f)\rightarrow \operatorname{rank}_\alpha(f)$ by Fact~\ref{apprank_rank}; $\operatorname{rank}_\alpha(f) \rightarrow \operatorname{rank}_{\alpha, +}(f)$ by Fact~\ref{easy}; $\operatorname{rank}_{\alpha, +}(f) \rightarrow R^{\operatorname{priv}}_\epsilon(f)$ by Fact~\ref{fact:R^priv_epsilon(f)-rank_alpha_+}. 
     
The relation $R^{\operatorname{priv}}_\epsilon(f) \rightarrow C(f)$ can be detailed separately in terms of 0- and 1-inputs. First, note that a 0-input for the function $f$ is equivalent to a 1-input for the function $\neg f$. Using the fact that $R^{\operatorname{priv}}_\epsilon(f) \leq D(f)$ and Facts~\ref{fact:D(f)-N^0(f)-N^1(f)},~\ref{fact:signrank+-N(f)},~\ref{fact:signrank'}, we jointly obtain $R^{\operatorname{priv}_\epsilon}(f)\rightarrow \operatorname{signrank}_+(\neg f)$ and $R^{\operatorname{priv}_\epsilon}(f)\rightarrow \operatorname{signrank}_+(f)$. From $\operatorname{signrank}_+(\neg f)$, we have $\operatorname{signrank}_+(\neg f)\rightarrow N^0(f)$ (in other words, $\operatorname{signrank}_+(\neg f)\rightarrow N^1(\neg f)$) by Fact~\ref{fact:signrank'}; $N^0(f)\rightarrow\operatorname{rank}_+(\neg f)$ by Fact~\ref{cor:N0-rank+'}; $\operatorname{rank}_+(f)\rightarrow C^0(f)$ by Facts~\ref{cor:rank+'-D'},~\ref{fact:N^Z(f)=logC^z(f)}~,\ref{fact:D(f)-N^0(f)-N^1(f)} (the direct relation $N^0(f)\rightarrow C^0(f)$ due to Fact~\ref{fact:N^Z(f)=logC^z(f)}). On the other hand, from $\operatorname{signrank}_+(f)$, we have $\operatorname{signrank}_+(f)\rightarrow N^1(f)$ (in other words, $\operatorname{signrank}_+(f)\rightarrow N^0(\neg f)$)by Fact~\ref{fact:signrank+-N(f)}; $N^1(f)\rightarrow \operatorname{rank}_+(f)$ by Fact~\ref{fact:N1-rank+}; $\operatorname{rank}_+(f)\rightarrow C^1(f)$ by Fact~\ref{fact:N^Z(f)=logC^z(f)},~\ref{fact:D(f)-N^0(f)-N^1(f)},~\ref{fact:rank+-D} (the direct relation $N^1(f)\rightarrow C^1(f)$ due to Fact~\ref{fact:N^Z(f)=logC^z(f)}). Finally, $C^0(f)\rightarrow C(f)$ and $C^1(f)\rightarrow C(f)$ due to Fact~\ref{fact:C(f)=C^0(f)+C^1(f)} and also $N(f) \leftrightarrow C(f)$ due to Fact~\ref{fact:N^Z(f)=logC^z(f)}. 

Furthermore, $R^{\operatorname{priv}}_\epsilon(f) \rightarrow AM^{\operatorname{priv}}_{c, s}(f)$ by Fact~\ref{fact:Rpriv-AM}; $AM^{\operatorname{priv}}_{c, s}(f) \rightarrow UAM_{c, s}(f)$ by Fact~\ref{fact:AM-UAM-ZAM}; $UAM_{c, s}(f) \rightarrow ZAM_{c, s}(f)$ by Fact~\ref{fact:AM-UAM-ZAM}; $ZAM_{c, s}(f) \rightarrow N^0(f)$ by Fact~\ref{fact:zam_LB}. 

Additionally , $AM_{c, s}(f) \leftrightarrow OIP^{[4]}_{c, s}(f)$ by Fact~\ref{fact:AM=OIOP4}; $\operatorname{rank}(f) \leftrightarrow \oplus P(f) $ by Fact~\ref{fact:oplusp-rank}; $\operatorname{rank}(f) \leftrightarrow \mathring{Q}_0(f)$ by Fact~\ref{fact:mathringQ}; $\operatorname{rank}(f) \leftrightarrow D^{A \rightarrow B}(f)$ by Fact~\ref{fact:rank-D1}; $D^{A \rightarrow B}(f) \leftrightarrow Q^{A \rightarrow B}_0(f)$ by Fact~\ref{fact:QAB0-D1}; $D^{A \rightarrow B}(f) \leftrightarrow Q^{A \rightarrow B, *}_0(f)$ by Fact~\ref{fact:QAB*-D1}, $Q_0(f) \leftrightarrow Q^c_0(f)$ by Fact~\ref{fact:simple_rs}; $Q^*_0(f) \leftrightarrow C^*_0(f)$ by Fact~\ref{fact:simple_rs}; $D(f)\leftrightarrow R^{\operatorname{pub}}_0(f)$ by Fact~\ref{fact:Rpub0-D}; $D(f)\leftrightarrow 
\mathring{R}_0(f)$ by Fact~\ref{fact:ringR}; $D^{A \rightarrow B}(f) \leftrightarrow LV^{A \rightarrow B}(f)$ by Facts~\ref{fact:D1-LV1},~\ref{fact:LV1-D1}. 
\end{proof}

We now describe the facts used in the proof of Theorem~\ref{thm:1}. We start with an upper bound result on $D(f)$ in terms of $\operatorname{rank}(f)$ by Lovett. 

\begin{fact}[{\cite[Theorem~1.1]{Shachar2016}}]\label{det_rank_UB}
Let $f : \mathcal X \times \mathcal Y \rightarrow \{-1, 1\}$ be a Boolean function. Then, $D(f)\leq O(\sqrt{\operatorname{rank}(f)}\log{\operatorname{rank}(f)})$. 
\end{fact}

For a function $f:\mathcal X \times \mathcal Y \rightarrow \{0, 1\}$, let $C^*_0(f)$ denote the (worst-case) cost of a protocol that computes $f$ exactly, which communicates classical bits but is allowed to make use of an unlimited (but finite) number of shared EPR-pairs. $Q^c_0(f)$ denotes the (worst-case) cost of a clean qubit protocol without prior entanglement that compute $f$ exactly, i.e. a
protocol that starts with $\ket{0}\ket{0}\ket{0}$ and ends with $\ket{0}\ket{f(x, y)}\ket{0}$~\cite{Buhrman2001}.  Some simple relations that hold between these measures are as follows: 
\begin{fact}[\cite{Buhrman2001},~\cite{bennett1993teleporting}]\label{fact:simple_rs}
For any function $f:\mathcal X \times \mathcal Y \rightarrow \{0, 1\}$, 
\begin{enumerate}
    \item $Q^*_0(f) \leq Q_0(f) \leq D(f) \leq D^{A \rightarrow B}(f)$,
    \item $Q_0(f) \leq Q^c_0(f) \leq 2Q_0(f)$ and
    \item $Q^*_0(f) \leq C^*_0(f) \leq 2 Q^*_0(f)$. 
\end{enumerate}
\end{fact}

As first mentioned by Buhrman, and Wigderson~\cite{buhrman1998quantum} and Ambainis \emph{et al.}~\cite{Ambainis2003}, the methods from~ Refs.\cite{yao1993quantum, Kre95} imply that $Q(f)\in \Omega(\log \operatorname{rank}(f))$. Buhman and de Wolf~\cite{Buhrman2001} then proved that the $\log\operatorname{rank}(f)$ lower bound holds for exact protocols. 


\begin{fact}[{\cite[Theorem~2]{Buhrman2001}}]\label{fact:Q_E^*-rank(f)}
For any function $f:\mathcal X \times \mathcal Y \rightarrow \{0, 1\}$, $Q^*_0(f) \geq \frac{\log \operatorname{rank}(f))}{2}$. 
\end{fact}



Other relations relating zero error quantum communication complexity with deterministic communication complexity are listed below. 
\begin{fact}[{\cite[Theorem~5.11]{klauck2007one}}]\label{fact:Q_0-D}
For all total functions $f$, we have $Q_0(f) = D(f)$. 
\end{fact}

\begin{fact}[{\cite[Theorem~4]{klauck2000quantum}}]\label{fact:QAB0-D1}
Let $f: \mathcal X  \times \mathcal Y \rightarrow \{0, 1\}$. Then, $Q^{A \rightarrow B}_0(f) = D^{A \rightarrow B}(f)$. 
\end{fact}

\begin{fact}[{\cite[Theorem~5]{klauck2000quantum}}]\label{fact:QAB*-D1}
Let $f: \mathcal X  \times \mathcal Y \rightarrow \{0, 1\}$. Then, $Q^{A\rightarrow B, *}_0(f) = \lceil D^{A \rightarrow B}(f)/2\rceil$. 
\end{fact}
    
\begin{fact}[{\cite[Equation~3]{mande2021one}}]\label{fact:rank-D1}
Let $f: \{0, 1\}^n \times \{0, 1\}^n \rightarrow \{0, 1\}$. Then, $\log\operatorname{rank}(f)\leq D^{A \rightarrow B}(f) \leq \operatorname{rank}(f)$. 
\end{fact}

A result of David \emph{et al.}~\cite{davis2022public} shows that the largest gap between public-coin zero-error randomized communication complexity and $D(f)$ is quadratic. This, together the trivial bound of $R^{\operatorname{pub}}_0(f)\leq D(f)$ gives the fact below. 
\begin{fact}[{\cite[Theorem~2.1]{davis2022public}}]\label{fact:Rpub0-D}
For every $f:\mathcal X \times \mathcal Y \rightarrow \{0, 1\}$, we have $\Omega\left((D(f))^{1/4}\right) \leq R^{\operatorname{pub}}_0(f)\leq D(f)$.  
\end{fact}

The Las Vegas communication complexity, denoted by $LV(f)$, is the communication complexity of a public-coin randomized protocol that does not err and either accept, rejects or says ``don't know". The one-way Las Vegas communication complexity $LV^{A\rightarrow B}(f)$ is trivially upper bounded by $D^{A \rightarrow B}(f)$ and lower bounded by $D^{A \rightarrow B}(f)/2$. 

\begin{fact}[{\cite[Theorem~2.1]{hromkovivc2001power}}]\label{fact:D1-LV1}
Let $f: \mathcal X \times \mathcal Y \rightarrow \{0, 1\}$ be a Boolean function. Then, $LV^{A \rightarrow B}(f)\geq D^{A \rightarrow B}(f)/2$.  
\end{fact}

\begin{fact}\label{fact:LV1-D1}
Let $f: \mathcal X \times \mathcal Y \rightarrow \{0, 1\}$ be a Boolean function. Then, $LV^{A \rightarrow B}(f)\leq D^{A \rightarrow B}(f)$.  
\end{fact}

Turning our attention to partition, cover numbers and nondeterministic communication complexity, we have a list of bounds among them. 
\begin{fact}[{\cite[Proposition~2.2] {KNisan96}}]\label{fact:prop_2.2}
Let $f: \mathcal X \times \mathcal Y \rightarrow \{0, 1\}$, then $C(f)\leq C^D(f)\leq C^P(f) \leq 2^{D(f)}$. 
\end{fact}

\begin{fact}[{\cite[Proposition~2.2] {KNisan96}}]\label{fact:C(f)=C^0(f)+C^1(f)}
Let $f: \mathcal X \times \mathcal Y \rightarrow \{0, 1\}$, then $C(f) = C^0(f) + C^1(f)$. 
\end{fact}



\begin{fact}[{\cite[Definition~2.3] {KNisan96}}]\label{fact:N^Z(f)=logC^z(f)}
Let $f: \mathcal X \times \mathcal Y \rightarrow \{0, 1\}$. For $z\in\{0,1\}$, $N^z(f) = \log C^z(f)$. Moreover, $N(f) = \log C(f)$. 
\end{fact}

\begin{fact}[{\cite[Theorem~2.11]{KNisan96}}]\label{fact:D(f)-N^0(f)-N^1(f)}
Let $f: \mathcal X \times \mathcal Y \rightarrow \{0, 1\}$, then $D(f) \leq O(N^0(f)N^1(f))$. 
\end{fact}

Deaett and Srinivasan~\cite{Deaett2012} showed that $N^1(f)$ can be tightly bounded by the nonnegative sign-rank, and this bound is tight up to constant factors. 
\begin{fact}[{\cite[Corollary~3.9] {Deaett2012}}]\label{fact:signrank+-N(f)}
Let $f:\mathcal X \times \mathcal Y \rightarrow \{0, 1\}$. Then, $\lceil\log \operatorname{signrank}_+(f)\rceil\leq N^1(f)\leq \lceil\log \operatorname{signrank}_+(f)\rceil+2$.
\end{fact}

An analogous result naturally holds between $N^0(f)$ and $\operatorname{signrank}_+(\neg f)$. We will be stating the $\neg f$ versions of certain results for cross-referencing purposes in our main results. 

\begin{fact}\label{fact:signrank'}
Let $f:\mathcal X \times \mathcal Y \rightarrow \{0, 1\}$. Then, $\lceil\log \operatorname{signrank}_+(\neg f)\rceil\leq N^0(f)\leq \lceil\log \operatorname{signrank}_+(\neg f)\rceil+2$. 
\end{fact}

Another variant of rank is the nonnegative rank, $\operatorname{rank}_+(f)$. It is clear that $\operatorname{rank}_+(f)$ is at least as large as $\operatorname{rank}(f)$. The nonnegative rank can be arbitrarily larger than the rank in the context of non-Boolean matrices. It was shown in that for every $k\in\mathbb N$, there exists a matrix $M$ such that $\operatorname{rank}(M) = 3$ and $\operatorname{rank}_+(M)\geq k$~\cite{beasley2009real}. However, no such separation between $\operatorname{rank}(M)$ and $\operatorname{rank}_+(f)$ is known when the matrix is restricted to being Boolean. For a Boolean function $f:\{0, 1\}^n\times \{0, 1\}^n\rightarrow \{0, 1\}$, we have $D(f)\leq O((\log \operatorname{rank}_+(f)+1)(\log \operatorname{rank}_+(\neg f)+1))$ and $D(f)\leq O(\log^2 r\operatorname{ank}_+(f))$~\cite{lovasz1989communication},{\cite[Theorem~2]{kol2014approximate} as upper bounds, while the lower bound is as below. 
\begin{fact}[\cite{lee2009lower}]\label{fact:rank+-D}
Let $f: \mathcal X \times \mathcal Y \rightarrow \{0, 1\}$. Then, $\log \operatorname{rank}_+(f)\leq D(f)$. 
\end{fact}
\begin{fact}\label{cor:rank+'-D'}
Let $f: \mathcal X \times \mathcal Y \rightarrow \{0, 1\}$. Then, $\log \operatorname{rank}_+(\neg f)\leq D(\neg f)$. 
\end{fact}

\begin{fact}[{\cite[Corollary] {yannakakis1988expressing}}, {\cite[Theorem~3.4] {lee2009lower}}]\label{fact:N1-rank+}
Let $f: \mathcal X \times \mathcal Y \rightarrow \{0, 1\}$. Then, $N^1(f)\leq \log \operatorname{rank}_+(f)$. 
\end{fact} 

\begin{fact}\label{cor:N0-rank+'}
Let $f: \mathcal X \times \mathcal Y \rightarrow \{0, 1\}$. Then, $N^0(f)\leq \log \operatorname{rank}_+(\neg f)$.
\end{fact}

In the following lemma, we show that if the nonnegative rank of a matrix is small, then it follows that its nonnegative sign-rank is small.

\begin{lem}
Let $M\in R^{m\times n}_+$ be a nonnegative matrix. Then $\operatorname{rank}_+(M)\geq \operatorname{signrank}_+(M)$. 
\end{lem}
\begin{proof}
Suppose that $M\in\mathbb R^{m\times n}$ is a nonnegative matrix. Let $k=\operatorname{rank}_+(M)$. Then, there exist nonnegative matrices $A\in\mathbb R^{m\times p}_+$ and $B\in\mathbb R^{p\times n}_+$ such that $M=AB$. It follows that $AB$ has the same sign pattern as M and hence, $\operatorname{signrank}_+(M)\leq k$. 
\end{proof}

\begin{fact}
Let $f: \mathcal X \times \mathcal Y \rightarrow \{0, 1\}$. Then, $\operatorname{rank}_+(f)\geq \operatorname{signrank}_+(f)$ and $\operatorname{rank}_+(\neg f)\geq \operatorname{signrank}_+(\neg f)$.
\end{fact}

Recall the definition of nonnegative approximate sign-rank from Section~\ref{sec:preliminaries}. Similarly to the separation between $\operatorname{rank}(f)$ and $\operatorname{rank}_+(f)$ in non-Boolean matrices, one can show that for every $k\in \mathbb N$, there exists a matrix $M$ and $\alpha > 0$ such that $\operatorname{rank}_\alpha(M) \leq 3$ and $\operatorname{rank}_{\alpha, +}(M) \geq k$~\cite{kol2014approximate}. The following shows how this complexity measure relates to different variants of $\operatorname{rank}(f)$ and private-coin randomized communication complexity. 
\begin{fact}[{\cite[Corollary~1]{Gal2021}}]\label{apprank_rank}
Let $f:\mathcal X \times \mathcal Y \rightarrow \{-1, 1\}$ be a Boolean function and let $\alpha > 1$ be a constant. Then, $\operatorname{rank}_\alpha(f)\geq \Omega(\log \operatorname{rank}_+(f))\geq \Omega(\log \operatorname{rank}(f))$
\end{fact}

\begin{fact}[\cite{kol2014approximate}]\label{easy}
For any nonegative matrix $M$, we have $\operatorname{rank}_\alpha(M)\leq \operatorname{rank}_{\alpha, +}(M)$.      
\end{fact}

\begin{fact}[\cite{Krause1996}]\label{fact:R^priv_epsilon(f)-rank_alpha_+}
Let $f:\mathcal X \times \mathcal Y \rightarrow \{-1, 1\}$. Then, $R^{\operatorname{priv}}_\epsilon(f)\geq \log \operatorname{rank}_{\alpha, +}(f)$. 
\end{fact}

Ambainis, Schulman, Ta-Shma, Vazirani and Wigderson introduced the classical and quantum sampling complexity~\cite{ambainis2003quantum}. Let $f:\mathcal X\times \mathcal Y\rightarrow \{0, 1\}$ and let $\mathcal D$ be any distribution on $\mathcal X\times \mathcal Y$. We say that a protocol samples $f$ according to $\mathcal D$ with error $\epsilon$ if the distribution the protocol induces on $\{(x, y, z)\}$ is $\epsilon$-close in the total variation distance to the distribution $(\mathcal D, f(\mathcal D))$ obtained by first choosing $(x, y)$ according to $\mathcal D$ and then evaluating $f(x, y)$. The number of bits (resp. qubits) needed for a randomized (resp. quantum) protocol to sample $f$ according to $\mathcal D$ with $\epsilon$ error is denoted as $\mathring{R}_\epsilon(f, \mathcal D)$ (resp. $\mathring{Q}_\epsilon(f, \mathcal D)$). When $\mathcal D$ is the uniform distribution, we write $\mathring{R}_\epsilon(f)$ (resp. $\mathring{Q}_\epsilon(f)$)~\cite{ambainis2003quantum}.

\begin{fact}[{\cite[Corollary~8.3] {ambainis2003quantum}}]\label{fact:mathringQ}
Let $f:\mathcal X\times \mathcal Y\rightarrow \{0, 1\}$, Then, $\mathring{Q}_0(f)=\Theta(log \operatorname{rank}(f))$. 
\end{fact}
  
\begin{fact}[{\cite[Theorem~8.4] {ambainis2003quantum}}]\label{fact:ringR}
Let $f:\mathcal X\times \mathcal Y\rightarrow \{0, 1\}$, Then, $\sqrt{D(f)}\leq \mathring{R}_0(f)\leq D(f)$. 
\end{fact}

The $\oplus P$ notion was first introduced by Papadimitriou and Zachos~\cite{papadimitriou1982two}. In the context of communication complexity, $\oplus P$ is defined as follows:  Let $f:\mathcal X\times \mathcal Y\rightarrow \{0, 1\}$ be a Boolean function. Given a cost-$k$ nondeterministic protocol for $f$ and its corresponding collection of rectangles $\{R_v:v\in\{0, 1\}^k\}$, where $R_v$ denotes the rectangle of inputs arriving at $v$ in the protocol. The protocol outputs 1 if the  number of $(x,y)$'s such that $(x, y)\in R_v$ is odd. The parity complexity of a function $f$, denoted as $\oplus P(f)$, is the total length $k$  of the accepting path in the protocol tree. A few results are known about the non-inclusion of $\oplus P(f)$ in communication complexity classes~\cite{forster2002linear, landscape}. 

\begin{fact}[{\cite[Observation~B.33] {landscape}}]\label{fact:oplusp-rank}
Let $f:\mathcal X\times \mathcal Y\rightarrow \{0, 1\}$, Then, $\oplus P(f)\in \log rank(f)\pm O(1)$. 
\end{fact}

Now we shall turn our attention to interactive proofs. In the interactive proof model of communication complexity, there are three computationally unbounded parties: Merlin, Alice and Bob. Both Alice and Bob can see only their input, while Merlin sees both Alice's and Bob's input. Merlin is the prover, who wants to convince the verifier, consisting of Alice and Bob together, that $f (x, y) = 1$. We formally define the Arthur-Merlin complexity. Consider a function $f:\mathcal X\times \mathcal Y \rightarrow Z$. In a private coin Arthur-Merlin (AM) protocol, Alice is given a uniform sample from some finite set $\mathcal R$, Bob is given a uniform sample from some other independent finite set $\mathcal Q$, and there is a collection of proofs $R_1, \cdots,R_m\subseteq (\mathcal X\times \mathcal R)\times (\mathcal Y\times \mathcal Q)$. The acceptance probability of the protocol on input $(x, y)$ is defined to be $\mathbb P_{r\in\mathcal R, q\in\mathcal Q}[\exists i:((x, r), y, q)\in R_i]$. The index $i$ of a rectangle $R_i$ represents a message sent from Merlin to Alice and Bob, who then separately decides whether to accept or not. The output of the protocol is 1 if and only if both Alice and Bob accept. The communication cost of the protocol is the length of Merlin’s proof. The protocol has completeness $c$ and soundness $s$ if it accepts 1-input with probability at least $c$ and accepts 0-inputs with probability at most $s$, We define $AM^{\operatorname{priv}}_{c, s}(f)$ to be the minimum cost over all $AM$ protocols for $f$ with completeness $c$ and soundness $s$. In short, an Arthur-Merlin ($AM$) protocol is a probability distribution over nondeterministic protocols, together with a bounded-error acceptance condition~\cite{goos2015zero}.  


An $AM$ protocol is \emph{unambiguous} if, for every 1-input and every outcome of the randomness, there is at most one proof of Merlin that causes the players to accept. More formally, if $(x, y\in f^{-1}(1))$, $r\in\mathcal R$, $q\in\mathcal Q$ and $i\neq j$, it holds that $((x, r), (y, q))\notin R_i\cap R_j$ and $i\neq j$. We write $UAM_{c, s}(f)$ to denote the minimum cost over all UAM protocols for $f$ with completeness $c$ and soundness $s$.

A $UAM$ protocol is a zero-information $AM$ ($ZAM$) protocol when the distribution of Merlin’s unique proof is identical across all 1-inputs. In other words, an unambiguous $AM$ protocol can be viewed as a mapping from an input $((x, r), (y, q))$ where $(x, y)\in f^{-1}(1)$ to the unique $i\in\{1, \cdots, m\}$ such that $((x, r), (y, q))\in R_i$, or to $\perp$ if no such $i$ exists~\cite{goos2015zero}.  We use $ZAM_{\epsilon}(f)$ to denote the minimum cost over all $ZAM$ protocols with completeness $c$ and soundness $s$ as $ZAM_{c, s}(f)$. 

\begin{fact}[\cite{goos2015zero}]\label{fact:AM-UAM-ZAM}
Let $f: \mathcal X \times \mathcal Y \rightarrow \{0, 1\}$  then $AM^{\operatorname{priv}}_{1, 1/2}(f) \leq UAM_{1, 1/2}(f) \leq ZAM_{1, 1/2}(f)$. 
\end{fact}
\begin{fact}[{\cite[Theorems~2 \& 6]{goos2015zero}}]\label{fact:zam_LB}
Let $f: \mathcal X \times \mathcal Y \rightarrow \{0, 1\}$  then  $\Omega \left(N^0(f)\right)\leq ZAM_{1, 1/2}(f)\leq O\left(2^{N^0(f)}\right)$. 
\end{fact}

Proving lower bounds on $AM(f)$ has been notoriously known to be hard~\cite{Klauck2003, klauck2011arthur, lokam2001spectral}. The only known $AM(f)$ lower bounds follow from the observation that $AM(f) \geq \Omega(\log R^{\operatorname{priv}}_\epsilon(f))$ for all $f$~\cite{landscape}. 

\begin{fact}[\cite{landscape}]\label{fact:Rpriv-AM}
For all $f$, we have $AM^{\operatorname{priv}}_{1, 1/2}(f) \geq \Omega(\log R^{\operatorname{priv}}_\epsilon(f))$. 
\end{fact}

Chakrabarti et al.~\cite{chakrabarti2015verifiable} gave a hierarchy of  communication models, called Online Interactive Proofs (OIPs). These models consist of $OMA^{[k]}$, $OIP^{[k]}$ and $OIP^{[k]}_+$ for $k \geq 1$, which are described as follows: In each model, Alice and Bob first toss some hidden coins (these coins are shared between Alice and Bob, but are not known to Merlin). Upon receiving the input, 
\begin{enumerate}[(i)]
    \item Merlin and Bob communicate for $k$ rounds, with Merlin being the last to communicate at the end of the $k$ rounds. 
    \item Alice sends Bob a message, randomized using public coins that are hidden from Merlin.
\end{enumerate}
Then, Bob produces an output in $\{0, 1\}$. The difference between the
three models are as follows.
\begin{itemize}
    \item $OMA^{[k]}$: (i) happens before (ii) and Bob can only look at his input after communicating with Merlin. 
    \item $OIP^{[k]}$: (i) happens before (ii) and Bob may choose to look at his input before communicating with Merlin.
    \item $OIP^{[k]}_+$: Similar to $OIP^{[k]}$, with the difference being that (ii) happens before (i). In this case, Bob's message may depend on Alice’s actual message to Bob, not just on his own input and the public coins. 
\end{itemize}

The soundness and completeness for all three models are defined similarly as that of the AM model. The communication cost of a protocol that computes a function $f$ is the sum of the maximum number of bits sent by Alice, the private coin and the number of bits communicated between Bob and Merlin. The corresponding communication complexity is the cost of the cheapest protocol that computes $f$ with soundness $s$ and completeness $c$. We shall denote them as  $OMA_{c, s}^{[k]}(f)$, $OIP_{c, s}^{[k]}(f)$ and $OIP^{[k]}_{+, c, s}(f)$. 

\begin{fact}[{\cite[Corollary~5.14]{chakrabarti2015verifiable}}]\label{fact:AM=OIOP4}
For all $f: \mathcal X \times \mathcal Y \rightarrow \{0, 1\}$, we have $\Omega(AM_{2/3, 1/3}(f)^c) \leq OIP_{2/3, 1/3}^{[4]}(f) \leq O(AM_{2/3, 1/3}(f) \log AM_{2/3, 1/3}(f))$. 
\end{fact}


\subsection{Class 2}\label{sec:class_2}
This class contains three complexity measures. Namely, the $\gamma_2$ norm, the randomized public coin communication complexity with exact error and the quantum communication complexity with prior entanglement and exact error. The equivalence relationship between  these complexity measures is depicted in Figure~\ref{C2}. 

\begin{figure}[H]
\centering
\begin{tikzpicture}[->,>=Stealth,auto,node distance=2.5cm, thick]

    \node (gamma_2) {$\gamma_2(f)$};
    \node (Q^*_epsilon)[below right of=gamma_2] {$Q^*_{=\epsilon}(f)$} ;
    \node (R^pub_epsilon) [above right of=Q^*_epsilon]{$R^{\operatorname{pub}}_{=\epsilon}(f)$};

    \draw [->] (gamma_2) to node[midway, left] {Lemma~\ref{facat:gamma2-Q^*_=epsilon}} (Q^*_epsilon);
    \draw [->] (Q^*_epsilon) to node[midway, right] {Fact~\ref{fact:Q*epsilon-Rprivepsilon}} (R^pub_epsilon);
    \draw [->] (R^pub_epsilon) to node[midway, above] {Lemma~\ref{lem:class2_UB}} (gamma_2);
\end{tikzpicture}
\caption{Relationship between complexity measures in Class 2.}
    \label{C2}
\end{figure}
\begin{theorem}\label{thm:2}
Class 2 includes the following complexity measures: 
\begin{multicols}{3}
    \begin{itemize}
        \item $\gamma_2(f)$
        \item $R^{\operatorname{pub}}_{=\epsilon}(f)$
        \item $Q^*_{=\epsilon}(f)$
    \end{itemize}
\end{multicols}
\end{theorem}
\begin{proof}
    We have $Q^*_{=\epsilon}(f)\rightarrow R^{\operatorname{pub}}_{=\epsilon}(f)$ by Fact~\ref{fact:Q*epsilon-Rprivepsilon}; $R^{\operatorname{pub}}_{=\epsilon}(f)\rightarrow \gamma_2(f)$ by Lemma~\ref{lem:class2_UB}; $\gamma_2(f)\rightarrow Q^*_{=\epsilon}(f)$ by Lemma~\ref{facat:gamma2-Q^*_=epsilon}. 
\end{proof}

We now describe the facts used in the proof of Theorem~\ref{thm:2}. Define oblivious-error public-coin randomized communication complexity as the usual public-coin randomized  communication complexity with the additional requirement that the error is exactly the same for all inputs $x,y$. In other words, let $\epsilon>0$ be the global error such that all 1-inputs are accepted with probability exactly $1-\epsilon$ and all 0-inputs are accepted with probability exactly $\epsilon$. For more flexibility, one can also define different errors for 1- and 0-inputs\footnote{More specifically, there exists $\epsilon_0, \epsilon_1>0$ such that 1-inputs are accepted with probability exactly $1 - \epsilon_1$ and 0-inputs are accepted with probability exactly $\epsilon_0$. Nevertheless, both measures are equivalent.}. For instance, the standard public coin protocol for the Equality problem has error 0 on 1-inputs and error $1/2$ on 0-inputs. We denote the oblivious-error randomized public coin communication complexity of a function $f$ as $R^{\operatorname{pub}}_{=\epsilon}(f)$. 

We first state an obvious fact.  This fact holds since we can replace classical transformations with their corresponding quantum counterpart. The public coin feature in randomized communication complexity corresponds to both players measuring the quantum state $\sum_i \sqrt{p(i)}\ket{i}$ in the same basis. 
\begin{fact}\label{fact:Q*epsilon-Rprivepsilon}
Let $f:\mathcal X \times \mathcal Y \rightarrow \{0, 1\}$ . Then,    $Q^*_{=\epsilon}(f)\leq R^{\operatorname{pub}}_{=\epsilon}(f)$. 
\end{fact}

In the following, we prove that $R_{=\epsilon}^{\operatorname{pub}} (f)$ is at least as large as the log of $\gamma_2(f)$. 
\begin{lem}\label{lem:class2_LB}
    $R_{=\epsilon}^{\operatorname{pub}} (f) \geq \log \gamma_2(f)$. 
\end{lem}
\begin{proof}
   Let $\epsilon>0$ be a small constant. Consider a $c$-round public-coin randomized communication protocol that accepts 0-inputs with probability $\epsilon$ and accepts 1-inputs with probability $1 - \epsilon$. Let $M\in\{\epsilon, 1 - \epsilon\}^{\vert \mathcal X\vert \times \vert \mathcal Y\vert}$ be the matrix whose entries $M_{x,y}$ correspond to the acceptance probabilities of the protocol on input $(x, y)$ for all $(x, y)\in \mathcal X\times \mathcal Y$. On the other hand, consider an optimal quantum protocol with shared entanglement that has cost $T$, such that the protocol accepts 0-inputs with probability at most $\epsilon$ and accepts 1-inputs with probability at least $1 - \epsilon$. Let $M'\in[0, 1]^{\vert \mathcal X\vert \times \vert \mathcal Y\vert}$ be the corresponding matrix of acceptance probabilities. Notice that for two sets of vectors $U\in\{\epsilon, 1 - \epsilon\}^n$ and $V = \{v\in[0, 1]^n\vert v(i)\in [0, \epsilon]\cup [1 - \epsilon, 1], \forall i\in[n]\}$, the following is true: 
   \begin{align*}
       \max_{u\in U}\Vert u\Vert_2 \leq \max_{v\in V} \Vert v\Vert_2. 
   \end{align*}
   Therefore, 
   \begin{align*}
       \gamma_2(M) 
       & = \min_{A, B: AB = M} \Vert A\Vert_r\Vert B\Vert_c \\
       & \leq \min_{A', B': A'B' = M'} \Vert A'\Vert_r\Vert B'\Vert_c \\
       & \leq 2^T \\
       & \leq 2^c,
   \end{align*}
   where the first inequality is due to {\cite[Lemma~12]{linial2007lower}} and the second inequality is due to Fact~\ref{fact:Q*epsilon-Rprivepsilon}.   
\end{proof}

We also prove that $R^{\operatorname{pub}}_{=\epsilon}(f)$ is polynomially  bounded from above by $\gamma_2(f)$. 

\begin{lem}\label{lem:class2_UB}
Let $f:\mathcal X \times \mathcal Y \rightarrow \{0, 1\}$. Then, $R_{=\epsilon}^{\operatorname{pub}} (f) \leq O((\gamma_2(f))^2)$. 
\end{lem}
\begin{proof}
We consider the definition of $\gamma_2(f)$ as in Equation (\ref{eqn:gamma_2}). Let $u_x, v_y$ be vectors of the optimal $\gamma_2$ decomposition of the communication matrix $M_f$ such that $\lVert u_x\rVert_2 , \lVert v_y\rVert_2 \leq \gamma_2(M_f)$. Given $x\in\mathcal X, y\in\mathcal Y$, Alice and Bob normalize their inputs to obtain $u_x,v_y$ and use a randomized public coin protocol~{\cite[Theorem~4.4]{kremer1999randomized}} with error at most $1/(10\gamma_2(f))$ to compute the inner product $\braket{u_x, v_y}$. This inner product is either -1 or 1 before normalization\footnote{This is due to the definition of $\gamma_2(f)$ and the fact that $M_f$ has entries $-1$ and $1$.}, so Alice and Bob can decide which is the case. Note that in this protocol, the error only depends on the absolute value of the inner product, since 
\begin{align*}
    \lvert \langle u_x, v_y\rangle  - \widetilde{\langle u_x, v_y\rangle}\vert\leq \frac{1}{10\gamma_2(f)} \leq \frac{1}{\lVert u_x\Vert_2\lVert v_y\Vert_2}\leq \frac{1}{\vert \langle u_x,  v_y\rangle \vert^2}
\end{align*}
by Cauchy-Schwartz inequality. Hence, the error is the same for all inputs.
\end{proof}

Lastly, we draw the connection between $\gamma_2(f)$ and $Q^*_{\epsilon}(f)$. 

\begin{lem}\label{facat:gamma2-Q^*_=epsilon}
Fix an $\epsilon > 0$. Let $f:\mathcal X \times \mathcal Y \rightarrow \{\epsilon, 1 - \epsilon\}$. Then, $\gamma_2(f) \leq 2^{Q^*_{=\epsilon}(f)}$. 
\end{lem}  
\begin{proof}
Consider a quantum protocol that solves $f$ with shared entanglement, that accepts $\epsilon$-inputs with probability $\epsilon$ and accepts $(1 - \epsilon)$-inputs with probability $1 - \epsilon$. Let $M_f$ be the communication matrix of $f$ and $P$ be its corresponding matrix of acceptance probabilities. Clearly, $P = M_f$. Then, 
\begin{align*}
    \gamma_2(M_f) = \gamma_2(P) \leq 2^{Q^*_{=\epsilon}(f)},
\end{align*}
where the last inequality is due to~{\cite[Lemmas~12]{linial2007lower}}. 
\end{proof}

\subsection{Class 3}\label{sec:useful_lemma_class_3}
This class consists of variants of the approximate $\gamma_2$ norm, various lower bounds measures, distributional complexity, weakly unbounded error communication complexity, quantum communication complexity with prior entanglement, public-coin randomized communication complexity, Merlin-Arthur communication complexity, variants of internal and external information complexity, pseudotranscript complexity and amortized communication complexity. The constant equivalence relationship between complexity measures in Class 3 is summarized in Figures~\ref{fig:C3a} and~\ref{fig:C3b}.
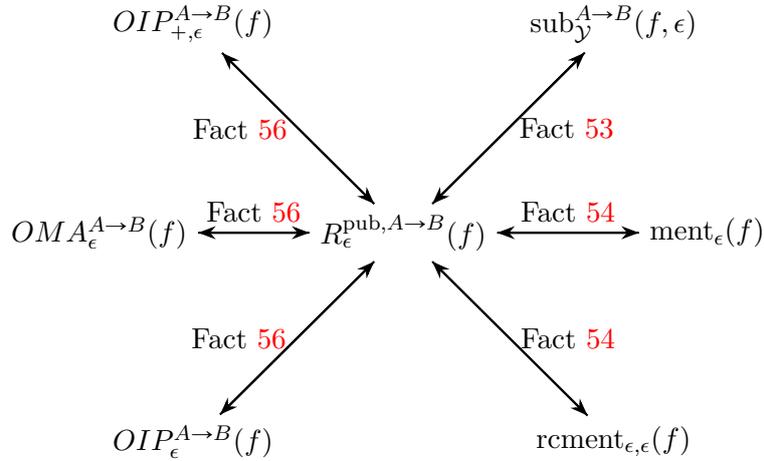
\begin{figure*}[h]
\centering 
\begin{subfigure}
\centering
\begin{tikzpicture}[->,>=Stealth,auto,node distance=2.5cm, thick]
    \node (RpubAB) {$R^{\operatorname{pub}, 
    A\rightarrow B}_\epsilon(f)$};
     \node (sub1) [above right of=RpubAB, xshift=1cm, yshift=1cm] {$\operatorname{sub}^{A\rightarrow B}_{\mathcal Y}(f, \epsilon)$};
     \node (ment) [right of=RpubAB, xshift=1.5cm] {$\operatorname{ment}_\epsilon(f)$};
      \node (rcment) [below right of=RpubAB, xshift=1cm, yshift=-1cm] {$\operatorname{rcment}_{\epsilon, \epsilon}(f)$};
      \node (OMAAB) [left of=RpubAB, xshift=-1.5cm] {$OMA^{A \rightarrow B}_\epsilon(f)$};
      \node (OIPAB) [below left of=RpubAB, xshift=-1cm, yshift=-1cm] {$OIP^{A \rightarrow B}_\epsilon(f)$};
      \node (OIPAB+) [above left of=RpubAB, xshift=-1cm, yshift=1cm] {$OIP^{A \rightarrow B}_{+, \epsilon}(f)$};
    \draw [->] (RpubAB) to node[midway, right] {Fact~\ref{fact:ment}} (rcment);
    \draw [->] (rcment) to node[near end, below] {}  (RpubAB);
    \draw [->] (RpubAB) to node[midway, above] {Fact~\ref{fact:ment}} (ment);
    \draw [->] (ment) to node[near end, below] {}  (RpubAB);
    \draw [->] (RpubAB) to node[midway, right] {Fact~\ref{fact:sub1}} (sub1);
    \draw [->] (sub1) to node[midway, left] {} (RpubAB);
    \draw [->] (RpubAB) to node[midway, right]  {} (OIPAB);
    \draw [->] (OIPAB) to node[midway, left]  {Fact~\ref{OIP}} (RpubAB);
    \draw [->] (RpubAB) to node[midway, right]  {} (OMAAB);
    \draw [->] (OMAAB) to node[midway, above]  {Fact~\ref{OIP}} (RpubAB);
    \draw [->] (RpubAB) to node[midway, right]  {} (OIPAB+);
    \draw [->] (OIPAB+) to node[midway, left]  {Fact~\ref{OIP}} (RpubAB);
\end{tikzpicture}
\caption{Constant equivalence relationship between complexity measures with $R^{A\rightarrow B, \operatorname{pub}}_\epsilon(f)$.}
\label{fig:C3a}
\end{subfigure}
\end{figure*}
\begin{figure*}[!h]
\centering 
\begin{subfigure}
\centering
\begin{tikzpicture}[->,>=Stealth,auto,node distance=2.5cm, thick]
    \node (gammaInf) {$\gamma^\infty_2(f)$};    
    \node (disc) [right of=gammaInf, xshift=2cm] {$\operatorname{disc}(f)$};
    \node (mc) [below left of=disc, yshift=0.5cm] {$mc(f)$};
    \node (PP) [below right of=disc, yshift=0.5cm] {$PP^{\operatorname{pub}}(f)$};
    \node (Dmu) [below right of=disc, xshift=1.4cm, yshift=-9.5cm]  {$D^\mu_\epsilon(f)$}; 
    \node (ardisc) [right of=Dmu, xshift=1cm] {$\operatorname{ardisc}^{\mu}_\epsilon(f)$};
    \node (rec) [right of=disc, xshift=4.5cm] {$\text{rec}_\epsilon(f)$};
    \node (MA) [right of=rec, xshift=0cm] {$MA_\epsilon(f)$};
    \node (RpubAB) [below of=gammaInf, yshift=-19cm] {$R^{A \rightarrow B, \operatorname{pub}}_\epsilon(f)$};
    \node (gamma2) [below right  of=gammaInf, xshift=0cm, yshift=-4.5cm] {$\gamma^\alpha_2(f)$};
    \node (Qstar) [below of=gamma2, yshift=-6cm] {$Q^*_\epsilon(f)$};
    \node (QIC) [below of=Qstar] {$QIC_\epsilon(f)$};
    \node (AQCC) [right of=QIC, xshift=2cm] {$AQCC_\epsilon(f)$}; 
    \node (Rpub) [right of=RpubAB, xshift=11.7cm]  {$R^{\operatorname{pub}}_\epsilon(f)$};
    \node (pprt) [above left of=Rpub, xshift=-2.8cm, yshift=2cm] {$\operatorname{pprt}_\epsilon(f)$}; 
    \node (srec) [below of=rec, yshift=0cm] {$\operatorname{srec}_\epsilon(f)$};
    \node (tildesrec) [below right of=srec] {$\widetilde{\operatorname{srec}}_\epsilon(f)$};
    \node (rprt) [below of=srec, xshift=0.05cm, yshift=0.8cm]{$\bar{\operatorname{prt}}_\epsilon(f)$};
    \node (prt) [below of=rprt, xshift=0.1cm, yshift=-1.7cm] {$\operatorname{prt}_\epsilon(f)$};
    \node (SBP) [below right of=rec, yshift=0.5cm] {$SBP(f)$};
    \node (tilderec) [below left of=rec, yshift=0.5cm] {$\widetilde{\operatorname{rec}}_\epsilon(f)$};
   \node (discmu) [below of=disc, xshift=0.8cm, yshift=0cm] {$\operatorname{disc}^{\mu}(f)$};
   \node (sdisc) [left of=discmu, xshift=0.5cm, yshift=-0.3cm] {$\operatorname{sdisc}_\epsilon(f)$};
   \node (wregmu) [right of=discmu, xshift=0.5cm] {$\operatorname{wreg}^{\mu}(f)$};
   \node (ICmu) [below of=discmu, yshift=-2cm] {$IC^{\mu}_\epsilon(f)$};
    \node (IC) [below of=ICmu, yshift=-5cm] {$IC_\epsilon(f)$};
    \node (ICD) [right of=IC, xshift=0.8cm] {$IC_{D, \epsilon}(f)$};
    \node (ICext) [below of=IC] {$IC^{\operatorname{ext}}_\epsilon(f)$};
    \node (ICDext) [right of=ICext, xshift=0.8cm] {$IC^{\operatorname{ext}}_{D, \epsilon}(f)$};
    \node(wprtmu) [below of=wregmu, yshift=1.2cm] {$\operatorname{wprt}^{\mu}_\epsilon(f)$};
   \node (prt+mu) [below left of=wprtmu, xshift=0.4cm, yshift=-0.5cm] {$\operatorname{prt}^{+, \mu}_\epsilon(f)$};
   \node (rdisc) [above of=prt+mu, yshift=0.3cm] {$\operatorname{rdisc}^\mu_\epsilon(f)$};
   \node(barprtmu) [below right of=wprtmu, xshift=-0.4cm, yshift=-0.5cm] {$\bar{\operatorname{prt}}^\mu_\epsilon(f)$};
   \node(prtmu) [below of=wprtmu, yshift=-2.1cm] {$\operatorname{prt}^{\mu}_\epsilon(f)$};
   \node (pprtmu) [below of=ardisc, xshift=0.7cm] {$\operatorname{pprt}^\mu_\epsilon(f)$};
   \node (pICmu) [below right of=prtmu, yshift=0.7cm] {$pIC^{\infty, \mu}_\epsilon(f)$};
   \node (pIC) [above right of=prt, xshift=-0.5cm, yshift=0.2cm] {$pIC_\epsilon(f)$};
   \node (QICD) [below of =sdisc, yshift=-10cm] {$QIC_{D, \epsilon}(f)$};
   \node (ACmu) [below left of=ICmu, xshift=1cm, yshift=0.5cm] {$AC^{\mu}_\epsilon(f)$};
    \draw [->] (RpubAB) to node[near start, left] {Fact~\ref{fact:RpubAB-gamma2inf}} (gammaInf);
    \draw [-> ](gammaInf) to node[midway, right] {Fact~\ref{fact:gamma2inf-gamma2alpha}} (gamma2);
    \draw [->] (gamma2) to node[midway, left] {Fact~\ref{fact:gamma2alpha-Qstar}} (Qstar);
    \draw [->] (mc) to node[midway, left] {Fact~\ref{fact:mc}} (disc);
    \draw [->] (disc) to node[midway, above] {} (mc);
    \draw [->] (disc) to node[midway, right] {Fact~\ref{fact:PP-disc}} (PP);
    \draw [->] (PP) to node[midway, right] {} (disc);
    \draw [->] (gammaInf) to node[midway, above] {Fact~\ref{fact:gammaInf-disc}} (disc);
    \draw [->] (rec)to node [midway, above] {Fact~\ref{fact:MApub-Rpub}} (MA);
    \draw [->] (MA)to node [near start, left] {Fact~\ref{fact:MApub-Rpub}} (Rpub);
    \draw [->] (Qstar) to node[midway, left] {Fact~\ref{fact:QCC-QIC}} (QIC); 
    \draw [->] (QIC) to node[midway, above] {Fact~\ref{fact:AQCC-QCC}} (AQCC);
    \draw [->] (AQCC) to node[midway, left] {} (QIC);
    \draw [->] (QIC) to node[midway, above] {Facts~\ref{fact:q-r},~\ref{fact:QIC-QCC}} (Rpub);
    \draw [->] (disc) to node[midway, above] {Fact~\ref{fact:recz-disc}, Defs.~\ref{def:disc},~\ref{def:rec_bound_conventional}} (rec);
    \draw [->] (disc) to node[near end, left] {Fact~\ref{fact:disc-sdisc}} (sdisc);
    \draw [->] (Rpub) to node[midway, above] {Fact~\ref{fact:Rpub-RpubAB}} (RpubAB);
    \draw [->] (sdisc) to node[midway, left] {Fact~\ref{fact:QICD-sdisc}} (QICD);
    \draw [->] (rec) to node[near end, right] {Fact~\ref{fact:partition}} (srec);
    \draw [->] (srec) to node[midway, left]{Fact~\ref{fact:rprt}} (rprt);
    \draw [->] (rprt) to node[near start, left]{Fact~\ref{fact:rprt}}(prt);
    \draw [->] (prt) to node[midway, left] {Fact~\ref{fact:partition}} (Rpub);
    \draw [->] (rec) to node[midway, right]  {Fact~\ref{fact:SBP}} (SBP);
    \draw [->] (SBP) to node[midway, right]  {} (rec);
    \draw [->] (ICmu) to node[align=center, near end, right] {Fact~\ref{fact:IC-ICD}, \\Eq.(\ref{def:3.3})} (IC);
    \draw[->] (IC) to node[midway, above] {Facts~\ref{fact:IC-ICD},~\ref{fact:3.5}} (ICD);
    \draw [->] (ICD) to node[,midway, above]{} (IC);
    \draw [->] (IC) to node[midway, right] {Fact~\ref{fact:ext-int}} (ICext);
    \draw [->] (ICD) to node[midway, left] {Fact~\ref{fact:ext-int}} (ICDext);
    \draw [->] (ICext) to node[midway, above] {Facts~\ref{fact:ICext-ICDext},~\ref{fact:3.16}} (ICDext);
    \draw [->] (ICDext) to node[midway, below] {} (ICext);
    \draw [->] (ICext) to node[near start, left] {Fact~\ref{fact:ICext-CC}} (Rpub);
    \draw [->] (rec) to node[midway, left] {Fact~\ref{fact:rec_bound_defs}} (tilderec);
    \draw [->] (tilderec) to node[midway, right] {} (rec);
    \draw [->] (srec) to node[midway, right] {Fact~\ref{fact:srec_bound_defs}} (tildesrec);
    \draw [->] (tildesrec) to node[midway, above] {} (srec);
    \draw [->] (Dmu) to node[midway, right] {Fact~\ref{fact:Dmu-Rpub}}  (Rpub);
    \draw [->] (discmu) to node[midway, left] {Fact~\ref{fact:discmu-ICmu}} (ICmu);
    \draw [->] (disc) to node[near end, right] {Eq.(\ref{eqn:disc})} (discmu);
    \draw [->] (ICmu) to node[align=center, near end, right] {Eq.~\ref{def:3.2}, \\ Fact~\ref{fact:ICpi_Dmu}} (Dmu);
    \draw [->] (QICD) to node[near end, right] {Fact~\ref{fact:QIC-QICD}} (QIC);
    \draw [->] (ICmu) to node[midway, left] {} (ACmu);
    \draw [->] (ACmu) to node[midway, left] {Fact~\ref{fact:amortized}} (ICmu);
    \draw [->] (pprt) to node[align=center, near start, right] {Fact~\ref{fact:pprt-Rpub},\\~\ref{fact:Rpub-pprt}}(Rpub);
    \draw [->] (Rpub) to node[midway, left] {} (pprt);
    \draw [->] (discmu) to node[midway, above] {Fact~\ref{fact:discmu-wregmu}} (wregmu);
    \draw [->] (wregmu) to node[midway, right] {Fact~\ref{fact:wreg_wprt}} (wprtmu);
    \draw [->] (wprtmu) to node[near end, right] {Fact~\ref{fact:prt_rs}} (prt+mu);
    \draw [->] (wprtmu) to node[midway, right] {Fact~\ref{fact:prt_rs}} (barprtmu);
     \draw [->] (barprtmu) to node[midway, right] {Fact~\ref{fact:prt_rs}} (prtmu);
     \draw [->] (prt+mu) to node[near start, right] {Fact~\ref{fact:prt_rs}} (prtmu);
     \draw [->] (prtmu) to node[midway, above] {Def.~\ref{def:partition_bound}} (prt);
     \draw [->] (Dmu) to node[midway, above] {Fact~\ref{facgt:Dmu_ardisc}} (ardisc);
     \draw [->] (ardisc) to node[midway, above] {} (Dmu);
     \draw [->]  (rdisc) to node[midway, right ] {} (prt+mu);
     \draw [->] (prt+mu) to node[midway, left] {Fact~\ref{fact:prt+-rdisc}}   (rdisc); 
     \draw [->] (ardisc) to node[near start, left] {Fact~\ref{def:pprt-ardisc}} (pprtmu);
     \draw [->] (pprtmu) to node[near start, left] {} (ardisc);
     \draw [->] (prtmu) to node[midway, left] {Fact~\ref{fact:pICmu-prtmu}} (pICmu);
     \draw [->] (pICmu) to node[midway, left] {} (prtmu);
     \draw [->] (prt) to node[midway, right] {Fact~\ref{fact:pIC-prt}} (pIC);
     \draw [->] (pIC) to node[midway, left] {} (prt); 
\end{tikzpicture}
\caption{Relationship between complexity measure in Class 3.}
\label{fig:C3b}
\end{subfigure}
\end{figure*}

\begin{theorem}\label{thm:3}
Class 3 includes the following complexity measures: 
\begin{multicols}{3}
    \begin{itemize}
        \item $R^{A \rightarrow B, \operatorname{pub}}_\epsilon(f)$
        \item $ \gamma^\infty_2(f)$
        \item $\gamma^\alpha_2(f)$
        \item $\operatorname{disc}^{\mu}(f)$
        \item $\operatorname{disc}(f)$
        \item $\operatorname{sdisc}_\epsilon(f)$
        \item $mc(f)$
        \item $PP^{\operatorname{pub}}(f)$
        \item $\operatorname{rec}_\epsilon(f)$
        \item $\widetilde{\operatorname{rec}}_\epsilon(f)$
        \item $\operatorname{srec}_\epsilon(f)$
        \item $\widetilde{\operatorname{srec}}_\epsilon(f)$
        \item $\bar{\operatorname{prt}}_\epsilon(f)$
        \item $\operatorname{prt}_\epsilon(f)$
        \item $\operatorname{sub^{A \rightarrow B}}_{\mathcal Y, \epsilon}(f)$
        \item $D^\mu_\epsilon(f)$
        \item $MA_\epsilon(f)$
        \item $SBP(f)$
        \item $\operatorname{ment}_\epsilon(f)$
        \item $\operatorname{rcment}_{\epsilon, \epsilon}(f)$
        \item $Q^*_\epsilon(f)$
        \item $QIC_{D, \epsilon}(f)$
        \item $QIC_\epsilon(f)$
        \item $AQCC_\epsilon(f)$
        \item $ R^{\operatorname{pub}}_\epsilon(f)$
        \item $IC^{\mu}_\epsilon(f)$
        \item $IC_\epsilon(f)$
        \item $IC_{D, \epsilon}(f)$
        \item $IC^{\operatorname{ext}}_\epsilon(f)$
        \item $IC^{\operatorname{ext}}_{D, \epsilon}(f)$
        \item $AC^{\mu}_\epsilon(f)$ 
        \item $\operatorname{wreg}^{\mu}(f)$
        \item $\operatorname{wprt}^\mu_\epsilon(f)$
        \item $\operatorname{prt}^{+, \mu}_\epsilon(f)$
        \item $\bar{\operatorname{prt}}^{\mu}_\epsilon(f)$
        \item $\operatorname{prt}^{\mu}_\epsilon(f)$
        \item $\operatorname{rdisc}^\mu_\epsilon(f)$
        \item $\operatorname{ardisc}^{\mu}_\epsilon(f)$
        \item $\operatorname{pprt}^\mu_\epsilon(f)$
        \item $\operatorname{pprt}_\epsilon(f)$
        \item $pIC^{\infty, \mu}_\epsilon(f)$
        \item  $pIC^{\infty}_\epsilon(f)$
        \item $OMA^{A \rightarrow B}_\epsilon(f)$
        \item $OIP^{A \rightarrow B}_\epsilon(f)$
        \item $OIP^{A \rightarrow B}_{+, \epsilon}(f)$
    \end{itemize}
\end{multicols}
\end{theorem}
\begin{proof}
First, note that $R^{A \rightarrow B, \operatorname{pub}}_\epsilon(f)\rightarrow \gamma^\infty_2(f)$ by Fact~\ref{fact:RpubAB-gamma2inf}; $\gamma_2^\infty(f) \rightarrow \operatorname{disc}(f)$ by Fact~\ref{fact:gammaInf-disc}; $\operatorname{disc}(f) \rightarrow \operatorname{rec}_\epsilon(f)$ by Fact~\ref{fact:recz-disc}, Definitions~\ref{def:disc} and~\ref{def:rec_bound_conventional}; $\operatorname{rec}_\epsilon(f) \rightarrow MA_\epsilon(f)$ by Fact~\ref{fact:MApub-Rpub}; $MA_\epsilon(f) \rightarrow R^{\operatorname{pub}}_\epsilon(f)$ by Fact~\ref{fact:MApub-Rpub}; $R^{\operatorname{pub}}_\epsilon(f) \rightarrow R^{A \rightarrow B, \operatorname{pub}}_\epsilon(f)$ by Fact~\ref{fact:Rpub-RpubAB}. 

From $\gamma_2^\infty(f)$, we have $\gamma_2^\infty(f) \rightarrow \gamma_2^\alpha(f)$ by Fact~\ref{fact:gamma2inf-gamma2alpha}; $\gamma_2^\alpha(f) \rightarrow Q^*_\epsilon(f)$ by Fact~\ref{fact:gamma2alpha-Qstar}; $Q^*_\epsilon(f) \rightarrow QIC_\epsilon(f)$ by Fact~\ref{fact:QCC-QIC}; $QIC_\epsilon(f) \rightarrow R^{\operatorname{pub}_\epsilon(f)}$ by Facts~\ref{fact:q-r} and~\ref{fact:QIC-QCC}. 

From $\operatorname{disc}(f)$, we have $\operatorname{disc}(f) \rightarrow \operatorname{sdisc}_\epsilon(f)$ by Fact~\ref{fact:disc-sdisc}; $\operatorname{sdisc}_\epsilon(f) \rightarrow QIC_{D, \epsilon}(f)$ by Fact~\ref{fact:QICD-sdisc}; $QIC_{D, \epsilon}(f) \rightarrow QIC_\epsilon(f)$ by Facts~\ref{fact:QIC-QICD} and~\ref{fact:QIC-QICD-alpha}. Furthermore, $\operatorname{disc}(f) \rightarrow \operatorname{disc}^{\mu}(f)$ by Equation~\ref{eqn:disc}; $\operatorname{disc}^{\mu}(f) \rightarrow IC^{\mu}_\epsilon(f)$ by Fact~\ref{fact:discmu-ICmu}; $IC^{\mu}_\epsilon(f) \rightarrow IC_\epsilon(f)$ by Fact~\ref{fact:IC-ICD} and Equation~\ref{def:3.3}; $IC_\epsilon(f) \rightarrow IC^{\operatorname{ext}}_\epsilon(f)$ by Fact~\ref{fact:ext-int}; $IC^{\operatorname{ext}}_\epsilon(f) \rightarrow R^{\operatorname{pub}}_\epsilon(f)$ by Fact~\ref{fact:ICext-CC}. At the same time, $IC_\epsilon(f) \leftrightarrow IC_{D, \epsilon}(f)$ by Facts~\ref{fact:IC-ICD},~\ref{fact:3.5}; $IC_{D, \epsilon}(f) \rightarrow IC^{\operatorname{ext}}_{D, \epsilon}(f)$ by Fact~\ref{fact:ext-int}; $ IC^{\operatorname{ext}}_\epsilon(f) \rightarrow IC^{\operatorname{ext}}_{D, \epsilon}(f)$ by Facts~\ref{fact:ICext-ICDext},~\ref{fact:3.16}.
From $IC^{\mu}_\epsilon(f)$, we have  $IC^{\mu}_\epsilon(f) \rightarrow D^{\mu}_\epsilon(f)$ by Equation~\ref{def:3.2} and Fact~\ref{fact:ICpi_Dmu}; $D^{\mu}_\epsilon(f) \rightarrow R^{\operatorname{pub}}_\epsilon(f)$ by Fact~\ref{fact:Dmu-Rpub}. 

From $\operatorname{disc}^{\mu}(f)$, we have $\operatorname{disc}^{\mu}(f) \rightarrow \operatorname{wreg}^{\mu}(f)$ by Fact~\ref{fact:discmu-wregmu}; $\operatorname{wreg}^{\mu}(f) \rightarrow \operatorname{wprt}^{\mu}(f)$ by Fact~\ref{fact:wreg_wprt}; $\operatorname{wprt}^{\mu}(f) \rightarrow \operatorname{prt}^{+, \mu}_\epsilon(f), \bar{\operatorname{prt}}^{\mu}_\epsilon(f)$ by Fact~\ref{fact:prt_rs}; $\operatorname{prt}^{+, \mu}_\epsilon(f), \bar{\operatorname{prt}}^{\mu}_\epsilon(f) \rightarrow \operatorname{prt}^{\mu}_\epsilon(f)$ by Fact~\ref{fact:prt_rs}. 

Moreover, $\operatorname{rec}_\epsilon(f) \rightarrow \operatorname{srec}_\epsilon(f)$ by Fact~\ref{fact:partition}; $\operatorname{srec}_\epsilon(f) \rightarrow \bar{\operatorname{prt}}_\epsilon(f)$ by Fact~\ref{fact:rprt}; $\bar{\operatorname{prt}}_\epsilon(f) \rightarrow \operatorname{prt}_\epsilon(f)$ by Fact~\ref{fact:rprt}; $\operatorname{prt}_\epsilon(f) \rightarrow R^{\operatorname{pub}}_\epsilon(f)$ by Fact~\ref{fact:partition}. 

Lastly, $\operatorname{prt}^{\mu}_\epsilon (f) \rightarrow \operatorname{prt}_\epsilon(f)$ by Definition~\ref{def:partition_bound}. For $\leftrightarrow$ relations, we have $QIC_\epsilon(f) \leftrightarrow AQCC_\epsilon(f)$ by Fact~\ref{fact:AQCC-QCC}; $\operatorname{mc}(f)\leftrightarrow \operatorname{disc}(f)$ by Fact~\ref{fact:mc}; $PP^{\operatorname{pub}}(f) \leftrightarrow \operatorname{disc}(f)$ by Fact~\ref{fact:PP-disc}; $AC^{\mu}_\epsilon(f) \leftrightarrow IC^{\mu}_\epsilon(f)$ by Fact~\ref{fact:amortized}; $\operatorname{prt}^{+, \mu}_\epsilon(f) \leftrightarrow \operatorname{rdisc}^{\mu}_\epsilon(f)$ by Fact~\ref{fact:prt+-rdisc}; $\operatorname{prt}^{\mu}_\epsilon(f) \leftrightarrow  pIC^{\infty, \mu}_\epsilon(f)$ by Fact~\ref{fact:pICmu-prtmu}; $\operatorname{prt}_\epsilon(f) \leftrightarrow  pIC^{\infty}_\epsilon(f)$ by Fact~\ref{fact:pIC-prt}; 
$\operatorname{rec}_\epsilon(f) \leftrightarrow \widetilde{\operatorname{rec}}_\epsilon(f)$ by Fact~\ref{fact:rec_bound_defs}; 
$\operatorname{rec}_\epsilon(f) \leftrightarrow SBP(f)$ by Fact~\ref{fact:SBP};
$\operatorname{srec}_\epsilon(f) \leftrightarrow \widetilde{\operatorname{srec}}_\epsilon(f)$ by Fact~\ref{fact:srec_bound_defs};
$R^{\operatorname{pub}}_\epsilon(f) \leftrightarrow \operatorname{pprt}_\epsilon(f)$ by Facts~\ref{fact:pprt-Rpub} and~\ref{fact:Rpub-pprt}.

For $\leftrightarrow$ relations with $R^{A \rightarrow B, \operatorname{pub}}_\epsilon(f)$, we have $R^{A \rightarrow B, \operatorname{pub}}_\epsilon(f)\leftrightarrow \operatorname{ment}_\epsilon(f)$ by Fact~\ref{fact:ment};
$R^{A \rightarrow B, \operatorname{pub}}_\epsilon(f)\leftrightarrow \operatorname{rcment}_{\epsilon, \epsilon}(f)$ by Fact~\ref{fact:ment};
$R^{A \rightarrow B, \operatorname{pub}}_\epsilon(f)\leftrightarrow OIP^{A \rightarrow B}_\epsilon(f)$ by  Fact~\ref{OIP};
$R^{A \rightarrow B, \operatorname{pub}}_\epsilon(f)\leftrightarrow OMA^{A \rightarrow B}_\epsilon(f)$ by Fact~\ref{OIP};
$R^{A \rightarrow B, \operatorname{pub}}_\epsilon(f)\leftrightarrow OIP^{A \rightarrow B}_{+, \epsilon}(f)$ by Fact~\ref{OIP};
$R^{A \rightarrow B, \operatorname{pub}}_\epsilon(f)\leftrightarrow \operatorname{sub}^{A\rightarrow B}_{\mathcal Y, \epsilon} (f)$ by Fact~\ref{fact:sub1}.  
\end{proof}

Below, we state the results used in the proof of Theorem~\ref{thm:3}. We first state the fact that one-way public-coin randomized communication complexity is at least exponentially larger than its two-way counterpart and quantum communication complexity with prior entanglement is smaller than public-coin randomized communication complexity. 
\begin{fact}[{\cite[Exercise~4.2.1]{KNisan96}}]\label{fact:Rpub-RpubAB}
For $f: \mathcal X \times \mathcal Y \rightarrow \{0, 1\}$, $R^{\operatorname{pub}}_\epsilon(f)\geq \log R^{A \rightarrow B, \operatorname{pub}}_\epsilon(f)$. 
\end{fact}

\begin{fact}\label{fact:q-r}
For $f: \mathcal X \times \mathcal Y \rightarrow \{0, 1\}$, $Q^*_\epsilon(f)\leq R^{\operatorname{pub}}_\epsilon(f)$.     
\end{fact}
\begin{proof}
    In order to simulate a public-coin randomized protocol using a quantum protocol with shared entanglement, Alice and Bob share the state $\sum_i \sqrt{p(i)}\ket{i}$, where $p(i)\in [0, 1]$ and $\sum_i \sqrt{p(i)} = 1$. Each of them then measures their sides of the entangled state to get the randomness. The quantum protocol then proceeds as per the randomized protocol. 
\end{proof}

The next three results from Linial and Shraibman~\cite{linial2007lower} show that  $\gamma_2^\infty(f) \leq \gamma_2^\alpha(f)$ and both $R^{\operatorname{pub}}_\epsilon(M_f)$ and $Q^*_\epsilon(M_f)$ can be lower bounded by $\gamma_2^\alpha(f)$ when $\alpha = \frac{1}{1 - 2\epsilon}$. Furthermore, Lee and Shraibman showed that $R^{A\to B, \operatorname{pub}}_\epsilon(f)$ is upper bounded by $\gamma_2^\infty(f)$~\cite{Lee2007a}. 

\begin{fact}[{\cite[Proposition~3]{linial2007lower}}]\label{fact:gamma2inf-gamma2alpha}
For every $m\times n$ sign matrix  $M_f$ and for every $\alpha\geq 1$, $\gamma_2^\infty(M_f)\leq \gamma_2^\alpha(M_f)\leq \gamma_2(M_f)\leq \sqrt{\operatorname{rank}(M_f)}$ for all $1\leq\alpha<\infty$. 
\end{fact}

\begin{fact}[{\cite[Theorem~1]{linial2007lower}}]
For every sign matrix $M_f$ and every $\epsilon >0$, $R^{\operatorname{pub}}_\epsilon(M_f)\geq 2\log \gamma_2^{\frac{1}{1-2\epsilon}}(M_f) - 2\log\left( \frac{1}{1-2\epsilon}\right)$. 
\end{fact}

\begin{fact}[{\cite[Theorem~1]{linial2007lower}}]\label{fact:gamma2alpha-Qstar}
For every sign matrix $M_f$ and every $\epsilon >0$,, $Q^*_\epsilon(M_f) \geq \log \gamma_2^{\frac{1}{1-2\epsilon}}(M_f) - \log\left( \frac{1}{1-2\epsilon}\right) - 2$.
\end{fact}

 \begin{fact}[{\cite[Claim~2]{linial2007lower}}]\label{fact:RpubAB-gamma2inf}
For any function $f$, we have $R^{A\to B, \operatorname{pub}}_\epsilon(f)\leq O((\gamma_2^\infty(f))^2)$. 
\end{fact}

Consider a probability distribution $\mu$ over the set of inputs $\mathcal X\times\mathcal Y $. The \emph{$(\mu, \epsilon)$-distributional complexity} of $f:\mathcal X\times\mathcal Y \rightarrow \mathcal Z$, denoted as $D^\mu_\epsilon(f)$, is the cost of the best deterministic protocol that errs on at most a $\epsilon$ fraction of the inputs, weighted according to the distribution $\mu$~\cite{KNisan96}. Distributional complexity is a useful complexity measure when obtaining lower bounds on randomized communication complexity, in both the one-way and two-way models. 

\begin{fact}[\cite{yao1983lower}, {\cite[Theorem~3.20] {KNisan96}}]\label{fact:Dmu-Rpub}
Let $f:\mathcal X \times \mathcal Y \rightarrow\{0, 1\}$,. Then, $R^{\operatorname{pub}}_\epsilon(f) = \max_\mu D_{\epsilon}^\mu(f)$ over distributions $\mu$. 
\end{fact}


Introduced by Chor and Goldreich~\cite{chor1988unbiased}, discrepancy has become one of the most commonly used measures in communication complexity to prove lower bounds for randomized protocols. Moreover, the inverse of this measure is rightly characterized by $\gamma_2^\infty(f)$ up to constant factors. For a Boolean matrix, $M_f$, the discrepancy of $M_f$ is the minimum over all input distributions of the maximum correlation that $M_f$ has with a rectangle. 
\begin{defn}[{\cite[Definition~3.27]{KNisan96}}]\label{def:disc}
    Let $f:\mathcal X\times \mathcal Y \rightarrow \{0,l\}$ be a function, $R$ be any rectangle, and $\mu$ be a probability distribution on $\mathcal X\times \mathcal Y$. Denote
    \begin{align*}
        & \operatorname{disc}_{\mu}(R,f) \\
        & = \Bigg\vert \Pr_\mu [f(x, y)  = 0 \text{ and } (x, y)\in R] \\
        & - \Pr_\mu [f(x, y) = 1 \text{ and } (x, y)\in R]\Bigg\vert. 
    \end{align*}
    The discrepancy of $f$ with respect to $\mu$ is given by 
    \begin{align*}
        \operatorname{disc}_\mu(f) = \max_R \operatorname{disc}_\mu(R, f),
    \end{align*}
    where the maximum is taken over all rectangles. Finally, the discrepancy of $f$ is 
    \begin{align}\label{eqn:disc}
        \operatorname{disc}(f) = \min_\mu \operatorname{disc}_\mu(f). 
    \end{align}
\end{defn}


\begin{fact}[{\cite[Theorem~7]{Anshu2017}}]\label{fact:Qstar-disc}
Let $f:\mathcal X \times \mathcal Y \rightarrow \{0, 1\}$ be a two-party function and let $\epsilon \in [0, 1/2)$. Then, $Q^*_\epsilon(f) = \Omega\left(\log\frac{1-2\epsilon}{\operatorname{disc}(f)}\right)$. 
\end{fact}

\begin{fact}[{\cite[Proposition~3.28]{KNisan96}}]\label{fact:disc-Dmu}
For every function $f:\mathcal X \times \mathcal Y \rightarrow \{0, 1\}$, every probability distribution $\mu$ on $\mathcal X \times \mathcal Y$ and every $\epsilon \geq 0$, $D^\mu_{1/2 - \epsilon}(f) \geq \log \left(\frac{2\epsilon}{\operatorname{disc}_\mu(f)}\right)$. 
\end{fact}

\begin{fact}[{\cite[Theorem~18]{linial2007lower}}]\label{fact:gammaInf-disc}
For every sign matrix $M_f$, $\frac{1}{8}\gamma_2^\infty(M_f)\leq \frac{1}{\operatorname{disc}(M_f)}\leq 8\gamma_2^\infty(M_f)$, 
\end{fact}



Define the weakly unbounded error (public-coin) communication complexity as follows {\cite[Definition~2.3.5] {mande2018communication}}:
$PP^{\operatorname{pub}}(f) = \displaystyle\inf_{0\leq\epsilon < 1/2}\left\{R^{\operatorname{pub}}_\epsilon(f) - \log\left(\frac{1}{1/2 - \epsilon}\right)\right\}$. 
The result below shows the equivalence relation between $PP^{\operatorname{pub}}(f)$ and $\operatorname{disc}(f)$. 
\begin{fact}[{\cite[Theorem~2.3.8]{mande2018communication}, \cite[Corollary~1.4]{Klauck2007}}]\label{fact:PP-disc}
For all $f:\{0, 1\}^n \times \{0, 1\}^n \rightarrow \{0, 1\}$, we have $PP^{\operatorname{pub}}(f) = \Theta\left(\log\frac{1}{\operatorname{disc}(f)}\right)$. 
\end{fact}

Given a sign matrix $M\in\{-1, 1\}^{n\times n}$, a (Euclidean) embedding of $M_f$ is a collection of vectors $u_1, \cdots, u_n, v_1, \cdots, v_n\in\mathbb R^k$ such that $\langle u_i, \cdot v_j\rangle \cdot M_{i, j} > 0$ for all $i, j\in[n]$. The margin of the embedding is defined as 
\begin{align*}
    \zeta = \min_{i, j} \frac{\vert \langle u_i, v_j\rangle \vert}{\Vert u_i\Vert_2 \cdot \Vert v_j\Vert_2}.
\end{align*}
The margin complexity of $M$, denoted as $\operatorname{mc}(M)$ is the minimum $1/\zeta$ over all embeddings of $M$~\cite{Sherstov2010}. The fact below shows a tight characterization of $\frac{1}{\operatorname{disc(f)}}$ in terms of $\operatorname{mc}(f)$ up to some constant factor. 
\begin{fact}[{\cite[Theorem~2.7]{Sherstov2010}, \cite[Theorem~3.1] {linial2009learning}}]\label{fact:mc}
Let $M_f$ be a sign matrix. Then, $\frac{1}{8}\operatorname{mc}(M_f)\leq \frac{1}{\operatorname{disc}(M_f)}\leq 8\operatorname{nc}(M_f)$.
\end{fact}

We now review a few other lower bound methods in communication complexity. In the beginning of this section, we briefly introduced the concept of rectangles.  We shall proceed with a more detailed description here.

\begin{defn}[Rectangle bound,~{\cite[Definition~5]{JK10}}]\label{def:rec}
Let $f:\mathcal X\times \mathcal Y\rightarrow \mathcal Z$. The $\epsilon$-rectangle bound of $f$, denoted as $\widetilde{\operatorname{rec}}_\epsilon(f)$ is defined as $\widetilde{\operatorname{rec}}_\epsilon(f)\coloneqq \max\{\widetilde{\operatorname{rec}}^z_\epsilon(f): z\in \mathcal Z\}$, where $\widetilde{\operatorname{rec}}^z_\epsilon(f)$ is the optimal value of the following linear program: 

$\displaystyle
\begin{aligned}
    & \underline{\text{Primal}}\\
    & \text{Min} \quad \sum_R w_R  \\
    & \text{st} \quad \forall (x, y) \in f^{-1}(z): \sum_{R:(x, y)\in R} w_R \geq 1 - \epsilon, \\ 
    & \quad \quad \forall (x, y) \in f^{-1} - f^{-1}(z): \sum_{R:(x, y)\in R} w_R \leq \epsilon, \\
    & \quad \forall R: w_R \geq 0.\\
    \\[1em]
    & \underline{\text{Dual}} \\
    & \text{Max} \quad \sum_{(x, y)\in f^{-1}(z)} (1 - \epsilon)\cdot \mu_{x, y} 
    - \sum_{(x, y)\in f^{-1} - f^{-1}(z)} \epsilon \cdot \mu_{x, y}\\
    & \text{st} \quad \forall R: \sum_{(x, y)\in f^{-1}(z) \cap R} \mu_{x, y} 
    - \sum_{(x, y)\in (R \cap f^{-1}) - f^{-1}(z)} \mu_{x, y} \leq 1, \\
    & \quad \quad \forall (x, y): \mu_{x,y}\geq 0.
\end{aligned}
$%
\end{defn}

\begin{defn}[Rectangle bound: conventional  definition~{\cite[Definition~3]{Klauck2003}},~{\cite[Definition~6]{JK10}}]\label{def:rec_bound_conventional}
Let $f:\mathcal X\times \mathcal Y\rightarrow \mathcal Z$ be a function and let $\mathcal R$ be the set of all rectangles in $\mathcal X\times \mathcal Y$. The $\epsilon$-rectangle bound of $f$ i, denoted as $\operatorname{rec}_\epsilon(f)$, is defined as: 
\begin{align*}
    & \operatorname{rec}_\epsilon(f) \coloneqq \max\{\operatorname{rec}^z_\epsilon(f): z\in\mathcal Z\}\\
    & \operatorname{rec}^z_\epsilon(f) \coloneqq \max\{\operatorname{rec}^{z, \lambda}_\epsilon(f):\lambda \text{ is a distribution on }\\
    & \quad \quad \quad \quad \mathcal X\times \mathcal Y \cap f^{-1}  \text{ with } \lambda(f^{-1}(z))\geq 0.5\}\\
    & \operatorname{rec}^{z, \lambda}_\epsilon(f) \coloneqq \min\Bigg\{\frac{1}{\lambda(f^{-1}(z)\cap R)}: R\in\mathcal R   \\
    & \quad \quad \quad \quad \quad \text{ with } \epsilon\cdot \lambda(f^{-1}(z)\cap R) > \lambda(R - f^{-1}(z))\Bigg\}, 
    \end{align*} 
\end{defn}

The two definitions of the rectangle bound has been shown to be equivalent.

\begin{fact}[{\cite[Lemma~1]{JK10}}]\label{fact:rec_bound_defs}
Let $f: \mathcal X \times \mathcal Y \rightarrow \mathcal Z$ be a function and let $\epsilon > 0$. Then, 
\begin{enumerate}
    \item $\widetilde{\operatorname{rec}}^z_\epsilon(f) \leq \operatorname{rec}^z_{\epsilon/2}(f)$.
    \item $\widetilde{\operatorname{rec}}^z_\epsilon(f) \geq \frac{1}{2}\cdot \left(\frac{1}{2} - \epsilon\right)\cdot \operatorname{rect}^z_{2\epsilon}(f)$. 
\end{enumerate}
\end{fact}

In the one-way communication model, a \emph{one-way rectangle} $R^{A\rightarrow B}$ is a set $S \times \mathcal Y$, where $S\subseteq \mathcal X$. For a distribution $\mu$ over $\mathcal X\times \mathcal Y$, let $\mu_{R^{A\rightarrow B}}$ be the distribution that arises from $\mu$ conditioned on the event $R^{A\rightarrow B}$ and let $\mu(R^{A\rightarrow B})$ represent the probability, under $\mu$, of event $R^{A\rightarrow B}$~{\cite[Definition~4]{Jain2009}}. The \emph{one-way rectangle bound}, denoted as $\operatorname{rec}^{A\rightarrow B}_\epsilon(f)$, is defined as~{\cite[Definition~6]{Jain2009}}
\begin{align*}
   \operatorname{rec}^{A\rightarrow B}_\epsilon(f) \coloneqq \displaystyle\max_\mu \min_{R^{A\rightarrow B}} \log \frac{1}{\mu(R^{A\rightarrow B})}. 
\end{align*}



\begin{defn}[Smooth-rectangle bound {\cite[Definition~2]{JK10}}]\label{def:smooth_rec_LP}
Let $f: \mathcal X \times \mathcal Y \rightarrow \mathcal Z$ be a function. The $\epsilon$-smooth rectangle bound of $f$, denoted as $\widetilde{\operatorname{srec}}_\epsilon(f)$, is defined to be $\max\{\widetilde{\operatorname{srec}}^z_\epsilon(f):z\in\mathcal Z\}$, where $\widetilde{\operatorname{srec}}^z_\epsilon(f)$ is given by the optimal value of the linear program below: 

\begin{align*}
    & \underline{\text{Primal}} \\
    & \text{Min} \quad \sum_{R \in \mathcal R} w_R \\
    & \text{st} \quad \forall (x, y)\in f^{-1}(z): \sum_{R: (x, y)\in R} w_R\geq 1 - \epsilon, \\
    & \quad \forall (x, y) \in f^{-1}(z): \sum_{R: (x, y)\in R} w_R \leq 1,\\
    & \quad \forall (x, y)\in f^{-1} - f^{-1}(z): \sum_{R:(x, y)\in R} w_R\leq\epsilon,  \\
    & \quad \forall R: w_R\geq 0. 
    \\[1em]
    & \underline{\text{Dual}} \\
    & \text{Max} \sum_{(x, y)\in f^{-1}(z)} ((1 - \epsilon) \mu_{x, y} - \phi_{x, y})\\
    & \quad \quad \quad \quad
    - \sum_{(x, y)\in f^{-1} - f^{-1}(z)} \epsilon \cdot \mu_{x, y} \\
    & \text{st} \quad \forall R: \sum_{(x, y)\in f^{-1}(z)\cap R} (\mu_{x, y} - \phi_{x, y}) \\
    & \quad \quad \quad \quad - \sum_{(x, y)\in (R\cap f^{-1}) - f^{-1}(z)} \mu_{x, y}\leq 1,\\
    & \quad \quad \forall (x, y): \mu_{x,y}\geq 0, \phi_{x, y}\geq 0.
\end{align*}
\end{defn}

\begin{defn}[Smooth-rectangle bound: natural definition~{\cite[Definition~3]{JK10}}]\label{def:smooth_rec_natural}
Let $f: \mathcal X \times \mathcal Y \rightarrow \mathcal Z$ be a function. The $(\epsilon, \delta)$-smooth-rectangle bound of $f$, denoted by $\operatorname{srec}_{\epsilon, \delta}(f)$, is defined as (refer to the definition of $\operatorname{rec}_\epsilon^{z, \lambda}$ in Definition~\ref{def:rec_bound_conventional})
\begin{align*}
    & \operatorname{srec}_{\epsilon, \delta}(f) \coloneqq \max\{\operatorname{srec}^z_{\epsilon, \delta}(f): z\in\mathcal Z\}\\
    & \operatorname{srec}^z_{\epsilon, \delta}(f) \coloneqq \max \{\operatorname{srec}^{z, \lambda}_{\epsilon, \delta}(f): \lambda \text{ is a probability }\\
    & \quad \quad \quad \quad \quad \text{distribution on }(\mathcal X \times \mathcal Y) \cap f^{-1}\}\\
    & \operatorname{srec}^{z, \lambda}_{\epsilon, \lambda}(f)\coloneqq \max \Bigg\{\operatorname{rec}^{z, \lambda}_\epsilon (g): g: \mathcal X \times \mathcal Y \rightarrow \mathcal Z ; \\
    & \quad \quad \quad \Pr_{(x, y)\leftarrow \lambda} [f(x, y) \neq g(x, y)] \leq \delta; \lambda(g^{-1}(z))\geq 0.5\Bigg\}. 
\end{align*}
\end{defn}

Definitions~\ref{def:smooth_rec_LP}  and~\ref{def:smooth_rec_natural} are also equivalent. 
\begin{fact}[{\cite[Lemma~2]{JK10}}]\label{fact:srec_bound_defs}
Let $f: \mathcal X \times \mathcal Y \rightarrow \mathcal Z$ be a function and let $\epsilon > 0$. Then, for all $z\in \mathcal Z$, 
\begin{enumerate}
    \item $\widetilde{\operatorname{srec}}^z_\epsilon(f) \leq \operatorname{srec}^z_{\epsilon/2, (1 - \epsilon)/2}(f)$. 
    \item $\widetilde{\operatorname{srec}}^z_{\epsilon}(f)\ge \frac{1}{2} \cdot \left(\frac{1}{4} - \epsilon\right) \cdot \operatorname{srec}^z_{2\epsilon, \epsilon/2}(f)$. 
\end{enumerate}
\end{fact}

\begin{defn}[Partition bound~{\cite[Definition~1]{JK10}}]\label{def:partition_bound}
Let $f: \mathcal X \times \mathcal Y \rightarrow \mathcal Z$ be a function. The $\epsilon$-partition bound of $f$, denoted $\operatorname{prt}_\epsilon(f)$, is given by the optimal value of the following linear program:

$\displaystyle
\begin{aligned}
    & \underline{\text{Primal}} \\
    & \text{Min} \quad \sum_z\sum_R w_{z, R}\\
    & \text{st} \quad \forall (x, y) \in f^{-1}: \sum_{R:(x, y)\in R} w_{f(x, y), R} \geq 1 - \epsilon,  \\
    & \forall (x, y):\sum_{R: (x, y)\in R} \sum_z w_{z, R} = 1,\\
    & \forall z, \forall R: w_{z, R}\geq 0.
    \\[1em]
    & \underline{\text{Dual}}\\
    & \text{Max} \quad \sum_{(x, y)\in f^{-1}} (1 - \epsilon) \mu_{x, y} + \sum_{(x, y)} \phi_{x, y}\\
    & \text{st} \quad \forall z, \forall R: \sum_{(x, y)\in f^{-1}(z) \cap R} \mu_{x, y} 
    + \sum_{(x, y)\in R} \phi_{x, y} \leq 1,\\
    & \quad \quad \forall (x, y): \mu_{x,y}\geq 0, \phi_{x, y}\in\mathbb R.
\end{aligned}
$%
\end{defn}

\begin{defn}[Relaxed partition bound {\cite[Definition~3.2]{kerenidis2015lower}}]\label{def:realaxed_prt}
    Let $f: \mathcal I \rightarrow \mathcal Z$, where $\mathcal I \subseteq \mathcal X \times \mathcal Y$ The distributional relaxed partition bound, denoted as  $\bar{\operatorname{prt}}^\mu_\epsilon(f)$, is the value of the following linear program. (The value of $z$ ranges over $\mathcal Z$ and $R$ over all rectangles, including the empty rectangle.)
\begin{align*}
    & \min_{\eta, p, R, z\geq 0} \quad \frac{1}{\eta}\\
    & \text{st } \sum_{(x, y)\in\mathcal I} \mu_{x, y} \sum_{R: (x, y)\in R} p_{R, f(x, y)} \\
    & \quad \quad + \sum_{(x, y)\notin\mathcal I} \mu_{x, y} \sum_{z, R:(x, y)\in R} p_{R, z} \geq (1 - \epsilon)\eta, \\
    & \quad \forall (x, y) \in \mathcal X \times \mathcal Y, \sum_{z, R: (x, y)\in R} p_{R, z} \leq \eta,\\
    & \quad \sum_{R, z} p_{R, z} = 1. 
\end{align*}
The relaxed partition bound is given by $\bar{\operatorname{prt}}_\epsilon(f) = \max_\mu \bar{\operatorname{prt}}^\mu_\epsilon(f)$. 
\end{defn}

The relaxed partition bound subsumes the norm-based methods (for example, the $\gamma_2$ method) and the rectangle-based methods (for example, the rectangle/corruption bound, the smooth rectangle
bound and the discrepancy bound). 
\begin{fact}[{\cite[Lemma~3.3]{kerenidis2015lower}}]\label{fact:rprt}
For all functions $f:\mathcal X\times \mathcal Y \rightarrow \mathcal Z$, $\epsilon$ and $z\in\mathcal Z$, we have $\operatorname{srec}^z_\epsilon(f)\leq \bar {\operatorname{prt}}_\epsilon(f) \leq \operatorname{prt}_\epsilon(f)$. 
\end{fact}

 Furthermore, the same authors showed that the distributional
complexity is an upper bound on the logarithm of the relaxed partition bound. 
\begin{fact}[{\cite[Corollary~1.4]{kerenidis2015lower}}]\label{fact:Dmu-prt}
    For every $\epsilon, \mu$ and every $f:\mathcal X\times \mathcal Y \rightarrow \{0, 1\}$,  $D^\mu_\epsilon(f) \geq \log \bar{\operatorname{prt}}^\mu_\epsilon(f)$. 
\end{fact}

\begin{defn}[Smooth discrepancy~{\cite[Definition~4]{JK10}}]\label{def:sdisc_LP}
Let $f: \mathcal X \times \mathcal Y \rightarrow \{0, 1\}$ be a Boolean function. The smooth discrepancy of $f$, denoted by $\widetilde{\operatorname{sdisc}}_\epsilon(f)$, is given by the optimal value of the linear program below

\begin{align*}
    & \underline{\text{Primal}} \\
    & \text{Min} \quad \sum_{R\in\mathcal R} w_R + v_R  \\
    & \text{st} \quad \quad 
    \forall (x, y) \in f^{-1}(1): 1 + \epsilon \geq \sum_{R:(x, y)\in R} w_{R} - v_R \geq 1,  \\
    & \quad \quad \forall (x, y) \in f^{-1}(0): 1 + \epsilon \geq \sum_{R: (x, y)\in R} v_R - w_{R}\geq 1,   \\
    & \quad \quad \forall R: w_F, v_R \geq 0. \\
    \\[1em]
    & \underline{\text{Dual}}\\
    & \text{Max} \sum_{(x, y)\in f^{-1}} \mu_{x, y} - (1 + \epsilon) \phi_{x, y},\\
    & \text{st} \quad \forall R: \sum_{(x, y)\in f^{-1}(1) \cap R} (\mu_{x, y} - \phi_{x, y)} \\
    & \quad \quad \quad \quad - \sum_{(x, y)\in R \cap f^{-1}(0)} (\mu{x, y} - \phi_{x, y}) \leq 1, \\
    & \quad \quad \forall R: \sum_{(x, y) \in f^{-1}(0)\cap R} (\mu_{x, y} - \phi_{x, y}) \\
    & \quad \quad \quad \quad - \sum_{(x, y)\in R\cap f^{-1}(1)} (\mu_{x, y} - \phi_{x, y}) \leq 1,\\
    & \quad \quad \forall (x, y): \mu_{x, y} \geq 0; \phi_{x,y}\geq 0.
\end{align*}
\end{defn}

\begin{defn}[Smooth discrepancy: natural definition~{\cite[Definition~9]{JK10}}]\label{def:sdisc_natural}
Let $f: \mathcal X \times \mathcal Y \rightarrow \{0, 1\}$ be a Boolean function. The $\delta$-smooth discrepancy of $f$, denoted as $\operatorname{sdisc}_\delta(f)$, is defined as follows
\begin{align*}
    & \operatorname{sdisc}_\delta(f) \coloneqq \max\{\operatorname{sdisc}^\lambda_\delta(f): \lambda \text{ is a distribution on }\mathcal X \times \mathcal Y\\
    & \quad \quad \quad \quad \quad \cap f^{-1}\},\\ 
    & \operatorname{sdisc}^\lambda_\delta(f) \coloneqq: \max\Bigg\{\operatorname{disc}^\lambda(g):g: \mathcal X \times \mathcal Y \rightarrow \mathcal Z;  \\
    & \quad \quad \quad \quad \quad \quad \Pr_{(x, y)\leftarrow \lambda} [f(x, y) \neq g(x, y)] < \delta\Bigg\}. 
\end{align*}
\end{defn}

Definitions~\ref{def:sdisc_LP} and~\ref{def:sdisc_natural} are equivalent. 
\begin{fact}[{\cite[Lemma~3]{JK10}}]
Let $f: \mathcal X \times \mathcal Y \rightarrow \{0, 1\}$ be a function and let $\epsilon > 0$. Then, 
\begin{enumerate}
    \item $\operatorname{sdisc}_{1/2 - \epsilon/8}(f) \geq 
    \widetilde{\operatorname{sdisc}}_\epsilon(f)$,
    \item $\frac{1}{2}\cdot \operatorname{sdisc}_{\frac{1}{4 + 2\epsilon}}(f) \leq \widetilde{\operatorname{sdisc}}_\epsilon(f)$. 
\end{enumerate}
\end{fact}

\begin{defn}[One-way subdistribution bounds,~{\cite[Dfinitino~3.2]{Jain2008}}]
For a distribution $\mu$ over $\mathcal X \times \mathcal Y$, define
$\operatorname{sub}^{A\rightarrow B, \mu}_{\mathcal Y, \epsilon}(f) \coloneqq \min_{\lambda} S_{\infty}(\lambda \Vert \mu)$\footnote{For distributions $\mathcal D_1, \mathcal D_2$, with respect to set $\mathcal X$, the relative co-min-entropy of $\mathcal D_1$ with respect to $\mathcal D_2$ is defined as $S_\infty (\mathcal D_1 \Vert \mathcal D_2) = \inf \{c: \mathcal D_2 \geq \mathcal D_1/2^c\}$. When $\mathcal X$ is finite, we have $S_\infty (\mathcal D_1 \Vert \mathcal D_2) = \max_{x\in\mathcal X} \log \frac{\mathcal D_1(x)}{\mathcal D_2(x)}$~\cite{Jain2008}.}, where $\lambda$ is taken over all distributions that are one-message-like\footnote{We say that $\lambda$ is one-message-like for $\mu$ with respect to $\mathcal X$ if for all $(x, y) \in \mathcal X \times \mathcal Y$, whenever $\lambda_{\mathcal X}(x) > 0$, we have $\mu_{\mathcal X} (x) > 0$ and $\lambda_{\mathcal Y}(y|x)= \mu_{\mathcal Y}(y|x)$. Here, $P_{\mathcal X}(x) = \sum_{y \in \mathcal Y} P(x, y)$ and $P_{\mathcal Y} (y|x) = \frac{P(x, y)}{P_{\mathcal X}(x)}$ for some distribution $P$. \label{foot:one-message}} for $\mu$ (with respect to $\mathcal X$) and one-way $\epsilon$-monochromatic for $f$. The one-way subdistribution bound is defined as such $\operatorname{sub}^{A \rightarrow B}_{\mathcal Y, \epsilon}(f) = \max_\mu \operatorname{sub}^{A \rightarrow B, \mu}_{\mathcal Y, \epsilon}(f)$, where $\mu$ is taken over all distributions on $\mathcal X \times \mathcal Y$. When the maximization is restricted to product distributions µ, we refer to the quantity as the one-way product subdistribution bound $\operatorname{sub}^{A\rightarrow B,[]}_{\mathcal Y, \epsilon}(f)$. 
\end{defn}

In the definition above, we say that a distribution $\lambda$ is $\epsilon$-monochromatic for $f$ if there exists a function $g: \mathcal Y \rightarrow \mathcal Z$ such that $\Pr_{XY \sim \lambda}[(X, Y, g(Y))\in f]\geq 1 - \epsilon$~{\cite[Definition~3.1]{Jain2008}}. 

The subsequent results show how different lower bound measures relate to each other. Since we only consider total Boolean functions in this work, results that hold for partial functions are useful for us in drawing the equivalence relation between complexity measures. Since the set of total functions is a subset of partial functions, lower bound results for partial functions also hold true for total functions. 

\begin{fact}[{\cite[Theorem~1]{JK10}}]\label{fact:partition}
    Let $f:\mathcal X\times\mathcal Y\rightarrow \mathcal Z$ be a (partial) function. Then, 
    \begin{enumerate}
        \item $R^{\operatorname{pub}}_\epsilon(f) \geq \log \operatorname{prt}_\epsilon(f)$. 
        \item $\operatorname{prt}_\epsilon(f)\geq \operatorname{srec}_\epsilon(f)$. 
        \item $\operatorname{srec}_\epsilon(f) \geq \operatorname{rec}_\epsilon(f)$. 
    \end{enumerate}
\end{fact}

\begin{fact}[{\cite[Lemma~4]{JK10}}]\label{fact:recz-disc}
 Let $f:\mathcal X\times \mathcal Y \rightarrow \{0, 1\}$ be a function; let $z\in \{0, 1\}$ and let $\lambda$ be a distribution on $\mathcal X\times \mathcal Y \cap f^{-1}$. Let $\epsilon, \delta > 0$, then $\operatorname{rec}^z_\epsilon(f)\geq \left(\frac{1}{2} - \epsilon\right)\operatorname{disc}^\lambda(f) - \frac{1}{2}$. 
\end{fact}

\begin{fact}[{\cite[Equation~29]{Braverman15}}]\label{fact:disc-sdisc}
Let $f: \mathcal X \times \mathcal Y\rightarrow \{0, 1\}$. Then, $\left(\frac{1}{\operatorname{disc}(f)}\right)^{O(1)} \leq 2^{O(\operatorname{sdisc}_{1/5}(f))}$. 
\end{fact}

We shall see how these lower bound measures associates with different variants of public-coin randomized communication complexity. In the one-way communication model, Jain, Klauck and Nayak~\cite{Jain2008} showed that for any relation $f\subseteq \mathcal X \times \mathcal Y \times \mathcal Z$, the subdistribution bound characterizes the public-coin randomized one-way communication complexity. 

\begin{fact}[{\cite[Theorem~4.4]{Jain2008}}]\label{fact:sub1}
Let $f \subseteq \mathcal X \times \mathcal Y \times \mathcal Z$ be a relation and let $0 \leq \epsilon \leq 1/6$. There exist universal constants $c_1,c_2$ such that $\operatorname{sub}^{A \rightarrow B}_{\mathcal Y, \epsilon}(f) - 1 \leq c_1 \cdot R^{A\rightarrow B, \operatorname{pub}}_\epsilon(f)\leq c_2 \left(\operatorname{sub}^{A \rightarrow B}_{\mathcal Y, \epsilon}(f) + \log \frac{1}{\epsilon} + 2 \right)$. 
\end{fact}


The tight characterization of $R^{A \rightarrow B, \operatorname{pub}}_\epsilon(f)$ in terms of the one-way subdistribution bound by Jain, Klauck and Nayak~\cite{Jain2008} was subsequently strengthened by Jain~\cite{jain2010strong} via the robust conditional relative min-entropy and relative min-entropy bounds. 

\begin{defn}[Robust conditional relative min-entropy bound~{\cite[Definition~3.4]{jain2010strong}}]
Let $f\subseteq \mathcal X \times \mathcal Y \times \mathcal Z$ be a relation. The $\epsilon$-error $\delta$-robust conditional relative min-entropy bound of $f$ with respect to distribution $\mu$, denoted as $\operatorname{rcment}^{\mu}_{\epsilon, \delta}(f)$, is defined as 
\begin{align*}
    & \operatorname{rcment}^{\mu}_{\epsilon, \delta}(f) \coloneqq \min\{\operatorname{rcment}^{\mu}_{\delta}(\lambda) \vert \lambda \text{ is one-way for } \mu \text{ and } \\
    & \quad \quad \quad \quad\quad \quad \quad \operatorname{err}_f(\lambda) \leq \epsilon\}, 
\end{align*}
where $\operatorname{err}_f(\lambda) \coloneqq \min\left\{\Pr_{(x, y) \leftarrow \lambda} [(x, y, h(y))\notin f \vert h: \mathcal Y \rightarrow \mathcal Z]\right\}$. The $\epsilon$-error $\delta$-robust conditional relative min-entropy bound of $f$, denoted as $\operatorname{rcment}_{\epsilon, \delta}(f)$, is 
\begin{align*}
    & \operatorname{rcment}_{\epsilon, \delta}(f) \coloneqq \max\{\operatorname{rcment}^{\mu}_{\epsilon, \delta}(f) \vert \mu \text{ is a distribution over }\\
    & \quad \quad  \quad \quad \quad \quad \quad \mathcal X \times \mathcal Y\}. 
\end{align*}
\end{defn}

\begin{defn}[Relative min-entropy bound~{\cite[Definition~3.5]{jain2010strong}}]
Let $f\subseteq \mathcal X \times \mathcal Y \times \mathcal Z$ be a relation. The $\epsilon$-error relative min-entropy bound of $f$ with respect to distribution $\mu$ is 
\begin{align*}
    & \operatorname{ment}^{\mu}_{\epsilon}(f) \\
    & \coloneqq \min\{S_{\infty}(\lambda \Vert \mu) \vert \lambda \text{ is one-way for } \mu \text{ and } \operatorname{err}_f(\lambda) \leq \epsilon\}, 
\end{align*}
where $S_{\infty}(\lambda \Vert \mu) = \max_{x \in\mathcal X} \log \frac{\lambda(x)}{\mu(x)}$\footnote{We use $\lambda(x)$ and $\mu(x)$ to denote de probability of $x$ under $\lambda$ and $\mu$ respectively.}. The $\epsilon$-error relative min-entropy bound of $f$, denoted as $\operatorname{ment}_\epsilon(f)$, is defined as
\begin{align*}
    \operatorname{ment}_\epsilon(f) \coloneqq \max \{\operatorname{ment}^{\mu}_{\epsilon}(f) \vert \mu \text{ is a distribution over }\mathcal X \times \mathcal Y\}. 
\end{align*}
\end{defn}

\begin{fact}[{\cite[Theorem~4.3]{jain2010strong}}]\label{fact:ment}
Let $f \subseteq \mathcal X \times \mathcal Y \times \mathcal Z$ be a relation and let $\epsilon > 0$. Then, $\operatorname{ment}_{2\epsilon}(f) - 1 \leq R^{A \rightarrow B, \operatorname{pub}}_\epsilon(f) \leq \operatorname{rcment}_{\epsilon/5, \epsilon/5}(f) + O \left(\log \frac{1}{\epsilon}\right)$. Therefore, $R^{A \rightarrow B, \operatorname{pub}}_\epsilon(f) = \Theta(\operatorname{ment}_\epsilon(f)) = \Theta(\operatorname{rcment}_{\epsilon, \epsilon}(f))$. 
    
\end{fact}

We formally define the Merlin-Arthur complexity below. 
\begin{defn}[Merlin-Arthur complexity~{\cite[Definition~8]{Klauck2003}},~{\cite[Definition~1]{klauck2011arthur}}]
In a Merlin-Arthur ($MA$) protocol for a Boolean function $f:\{0, 1\}^n \times \{0, 1\}^n \rightarrow \{0, 1\}$,  Alice receives a message (also known as the proof) from Merlin. Then, Alice and Bob communicate using public key randomness until they compute an output (the proof cannot depend on the randomness). The cost of an $MA$ protocol is the sum of the length $a$ of the proof, and the length $c$ of the overall communication between Alice and Bob. 

We say that the protocol computes $f$ if for all inputs $x, y$ with $f(x, y) = 1$, there exists a proof such that $x, y$ is accepted with probability $1 - \epsilon$ and for all inputs $x, y$ with $f(x, y) = 0$ and all proofs, the probability that $x, y$ is accepted is at most $\epsilon$, for some $\epsilon < 1/2$. The Merlin-Arthur complexity of $f$, denoted $MA_\epsilon(f)$, is the smallest cost of an $MA$ protocol for $f$. 
\end{defn}
Klauck~\cite{Klauck2003} drew the relation between public coin interactive proofs, the rectangle bound and randomized communication complexity. 

\begin{fact}[{\cite[Corollary~2]{Klauck2003}}]\label{fact:MApub-Rpub}
Let $f:\{0, 1\}^n \times \{0, 1\}^n \rightarrow \{0, 1\}$. Then, 
\begin{enumerate}
    \item $R^{\operatorname{pub}}_{1/4}(f)\geq\Omega(\operatorname{rec}_{1/4}(f))$,
    \item $R^{\operatorname{pub}}_{1/4}(f)\geq \max\{MA_{1/4}(f), MA_{1/4}(\neg f)\}\geq \Omega(\sqrt{\operatorname{rec}_{1/4}(f)})$,
    \item $\operatorname{rec}_{1/4}(f)\geq \Omega(\max\{AM_{1/4}(f), AM_{1/4}(\neg f)\})$. 
\end{enumerate} 
\end{fact}

Recall the definition of online interactive proofs from Section~\ref{sec:useful_1}. Similarly to how communication complexity classes are defined in the work of Babai et al.~\cite{babai1986complexity}, Chakrabarti et al. showed that the following communication complexity classes are equivalent:  $\mathsf{OMA}^{[1]} = \mathsf{OIP}^{[1]} = \mathsf{OIP}_{+}^{[1]} = \mathsf{MA}^{[1, A]} = \mathsf{R}^{[1, A]}$\footnote{$\mathsf{MA}^{[1, A]}$  denotes the one-round (or one-way) MA complexity where Alice is the first to send a message, whereas $\mathsf{R}^{[1, A]}$ denotes the one-round (or one-way) public-coin randomized communication complexity with Alice being the first to send a message.}~{\cite[Theorem~5.1]{chakrabarti2015verifiable}}. We note here that by changing the criteria for efficiency from polylogarithmic to constant, the result of Chakrabarti et al. remains almost the same in our setting. 

\begin{fact}[{\cite[Theorem~5.1]{chakrabarti2015verifiable}}]\label{OIP}
For all $f$, and $\epsilon > 0$, we have  $OMA_\epsilon^{A \rightarrow B}(f) = \Theta(OIP_\epsilon^{A \rightarrow B}(f)) = \Theta(OIP_{+, \epsilon}^{A \rightarrow B}(f)) = \Theta(MA_\epsilon^{A \rightarrow B}(f))$. Moreover, $ MA_\epsilon^{A \rightarrow B}(f) \leq R^{A \rightarrow B, \operatorname{pub}}_\epsilon(f) \leq (MA_\epsilon^{A \rightarrow B}(f))^2$. 
\end{fact}

Introduced by G\"{o}\"{o}s and Watson~\cite{goos2016communication}, the small bounded-error probabilities communication complexity, denoted as $SBP(f)$, is given by $SBP(f) \coloneqq \min_{\alpha(n)>0} R^{\operatorname{pub}}_{\alpha, \alpha/2}(f) + \log (1/\alpha)$, where $R^{\operatorname{pub}}_{\alpha, \beta}(f)$ is the public-coin randomized communication complexity of a protocol that computes $f$ with probability of acceptance of at least $\alpha(n)$ on all 1-inputs and at most $\beta(n)$ on all 0-inputs of length $n$. The authors showed that this complexity measure is exactly characterized by the rectangle bound.
\begin{fact}[{\cite[Theorem~1]{goos2016communication}}]\label{fact:SBP}
Let $f: \{0, 1\}^n \times \{0, 1\}^n \rightarrow \{0, 1\}$ and $0 < \epsilon\leq 1/8$. Then, $SBP(f) = \Theta(\operatorname{rec}_\epsilon(f))$\footnote{Although it is written in~\cite{goos2016communication} that $SBP(f) = \Theta(\operatorname{rec}_{1/8}(f))$, one can replace the constant factor of 1/8 by any positive constant at most 1/8 while affecting the rectangle bound by only a constant factor~\cite{Klauck2003}.} for all $f$.
\end{fact}

Other lower bound techniques such as the relative discrepancy bound, the positive partition bound, the weak partition bound, the public-coin partition bound, the adaptive relative discrepancy and the weak regularity are also shown to be related to each other as well as to the public-coin randomized communication co,complexity. We note that although the partition bound was mentioned earlier in Definition~\ref{def:partition_bound}, we shall restate it below as a fixed-distribution version, equivalent to the partition bound in Definition~\ref{def:partition_bound} in the worst case distribution. 

\begin{defn}[Partition bound (fixed-distribution variant)~{\cite[Definition~5]{fontes2016relative}}]\label{def:partition_bound_2}
Let $f: \mathcal X \times \mathcal Y \rightarrow \{0, 1\}$ and let $\epsilon > 0$. For a distribution, $\mu$, the partition bound of $f$ with respect to $\mu$, denoted as $\operatorname{prt}^\mu_\epsilon(f)$ ,is the optimal value of the following problem: 
\begin{align*}
    & \max_{\kappa, \phi_{x, y}} \quad \phi - \epsilon\kappa\\
    & \text{st} \quad \quad \quad \quad \phi(R) - \kappa\mu(R \cap f^{-1}(z)) \leq 1 & \forall R, z\\
    & \quad \quad \quad \quad \kappa \geq 0. 
\end{align*}
Furthermore, $\operatorname{prt}_\epsilon(f) = \max_\mu \operatorname{prt}^\mu_\epsilon(f)$. 
\end{defn}

\begin{defn}[Relative discrepancy bound~{\cite[Definition~3]{fontes2016relative}},~\cite{ganor2016exponential}]\label{def:rdisc}
Let $\mu$ be a distribution over $\mathcal X \times \mathcal  Y $ and let $f: \operatorname{supp}(\mu) \rightarrow \{0, 1\}$ be a function\footnote{As in~\cite{fontes2016relative}, $\operatorname{supp}(\mu) = \operatorname{supp}(f)$ is assumed.}. The distributional relative discrepancy bound of $f$, denoted as $\operatorname{rdisc}^{\mu}_\epsilon(f)$, is the optimal value of the following problem:

\begin{align*}
    & \sup_{\alpha, \delta, \phi_{xy}}  \quad \frac{1}{\delta} \left(\frac{1}{2} - \alpha - \epsilon \right) \\
    & \text{st} \quad \left(\frac{1}{2} - \alpha\right) \cdot \phi(R) \leq \mu(R \cap f^{-1}(z)) \quad \forall R, z \text{ s.t. } \phi(R) \geq \delta\\
    & \quad \quad \sum_{xy} \phi_{xy} = 1\\
    & \quad \quad 0\leq \alpha < \frac{1}{2}, 0 < \delta < 1, \phi_{xy} \geq 0 \quad \forall (x, y). 
\end{align*}
For the non-distributional case, the relative discrepancy bound $\operatorname{rdisc}_\epsilon(f) = \max_\mu \operatorname{rdisc}_\epsilon^{\mu}(f)$, where the maximum is taken over all distribution $\mu$ over $\mathcal X \times \mathcal Y$, implicitly adding nonnegativity and normalization constraints on $\mu$. 
\end{defn}

\begin{defn}[Positive partition bound~{\cite[Definition~6]{fontes2016relative}}]
Let $f: \mathcal X \times \mathcal Y \rightarrow \{0, 1\}$ be a (possibly partial) Boolean function and let $\epsilon > 0$. For a distribution $\mu$, the positive partition bound of $f$, denoted as $\operatorname{prt}^{+, \mu}_\epsilon(f)$, is the optimal value of the following maximization problem:
\begin{align*}
    & \max_{\kappa, \phi_{xy}} \quad \quad  \phi - \epsilon\kappa\\
    & \text{st} \quad \quad \quad \quad  \phi(R) - \kappa \mu(R \cap f^{-1}(z)) \leq 1 & \forall R, z\\
    & \quad \quad \quad \quad  \kappa \geq 0, \quad \phi_{xy} \geq 0 & \forall (x, y). 
\end{align*}
The positive partition bound is defined as $\operatorname{prt}^+_\epsilon(f) = \max_{\mu} \operatorname{prt}^{+, \mu}_\epsilon(f)$. 
\end{defn}

\begin{defn}[Weak partition bound~{\cite[Definition~7]{fontes2016relative}}]
Let $f: \mathcal X \times \mathcal Y \rightarrow \{0, 1\}$ denote a (possibly partial) Boolean function and let $\epsilon > 0$. For a distribution $\mu$, the weak partition bound, denoted as $\operatorname{wprt}^{\mu}_\epsilon(f)$, is the optimal value of the problem below:
\begin{align*}
    & \max_{\kappa, \phi_{xy}} \quad \quad \phi - \epsilon\kappa\\
    & \text{st}  \quad \quad \quad \quad \phi(R) - \kappa \mu(R \cap f^{-1}(z)) \leq 1 & \forall R, z\\
    & \quad \quad \quad \quad  \kappa \geq 0, \quad \phi_{xy}\geq 0, \quad \kappa\mu_{xy} - \phi_{xy} \geq 0 & \forall (x, y). 
\end{align*}
The weak partition bound is defined as $\operatorname{wprt}_\epsilon(f) = \max_\mu \operatorname{wprt}^{\mu}_\epsilon(f)$. 
\end{defn}

\begin{defn}[Public-coin partition bound~{\cite[Definition~2]{jain2014quadratically}}]
Let $\subseteq \mathcal X \times \mathcal Y \times \mathcal Z$ be a relation and let $\epsilon > 0$. The $\epsilon$-public-coin partition bound pf $f$, denoted as $\operatorname{pprt}_\epsilon(f)$, is the optimal value of the linear program below. 

\begin{align*}
    & \underline{\text{Primal}} \\
    & \text{Min} \quad \sum_z\sum_R w_{z, R}  \\
    & \text{st} \quad \forall (x, y): \sum_{z: (z, y, z)\in f} \sum_{R: (x, y)\in R} w_{z, R} \geq 1 - \epsilon,  \\
    & \quad \quad \forall (x, y): \sum_{R: (x, y)\in R} \sum_{z} w_{z, R} = 1 \\
    & \quad \quad \forall (z, R): w_{z, R} = \sum_{P:(z, R)\in P} a_P, \\
    & \quad \quad\sum_P a_P = 1,\\
    & \quad \quad \forall (z, R): w_{z, R} \geq 0; \quad \forall P: a_P \geq 0.\\
    \\[1em]
    & \underline{\text{Dual}}\\
    & \text{Max} \quad (1 - \epsilon)\sum_{(x, y)} \mu_{x, y} + \sum_{(x, y)} \phi_{x, y} + \lambda\\
    & \text{st} \quad \forall (z, R): \sum_{(x, y)\in R: (x, y, z)\in f} \mu_{x, y} + \sum_{(x, y)\in R} \phi_{x, y} + v_{z, R }\leq 1,\\
     & \quad \quad\forall P: \sum_{(z, R) \in P} v_{z, R} \geq \lambda,\\
     & \quad \quad \forall (x, y): \mu_{x, y} \geq 0, \phi_{x, y}\in \mathbb R;  \quad \forall (z, R): v_{z, R} \in \mathbb R\\
    & \quad \quad \lambda\in \mathbb R.
\end{align*}

Here, $R$ represents a rectangle in $\mathcal X \times \mathcal Y$ and $P$ is a partition along with outputs $z \in \mathcal Z$, i.e. $P = \{(z_1, R_1), \cdots, (z_m, R_m)\}$ such that $\{R_1, \cdots , R_m\}$ form a partition of $\mathcal X \times \mathcal Y$ into rectangles and $z_i \in\mathcal Z$ for all $i\in[m]$. 
\end{defn}

\begin{fact}[Adaptive relative discrepancy~{\cite[Definition~9]{fontes2016relative}},~\cite{ganor2016exponential}]
Let $f: \mathcal X \times \mathcal Y \rightarrow \{0, 1\}$ be a (possibly partial) Boolean function and let $\epsilon > 0$. For a distribution $\mu$, the adaptive relative discrepancy of $f$ with respect to $\mu$, denoted as $\operatorname{ardisc}^\mu_\epsilon(f)$, is the optimal value of the optimization problem below:
\begin{align*}
    & \sup_{\alpha, \delta, \phi^P_{x, y}}  \frac{1}{\delta}\left(\frac{1}{2} - \alpha - \epsilon \right) \\
    & \text{st} \quad \left(\frac{1}{2} - \alpha \right) \phi^P(R) \leq \mu(R \cap f^{-1}(z))\\
    & \quad \quad \quad \quad \quad \quad \quad \quad \quad \quad \forall P \forall (z, R) \in P \text{ st } \phi^P(R) \geq \delta \\
    &  \quad \quad \phi^P = 1 \quad  \quad  \quad  \quad  \quad  \quad  \quad \quad  \quad  \quad  \quad   \forall P\\
    & \quad \quad  0 \leq \alpha < \frac{1}{2}, 0 < \delta < 1, \phi^P_{x,y}\geq 0 \quad \quad   \forall P, \forall (x, y)
\end{align*}
Furthermore, $\operatorname{ardisc}_\epsilon(f) = \max_\mu \operatorname{ardisc}^\mu_\epsilon(f)$. 
\end{fact}

\begin{defn}[Weak regularity~{\cite[Definition~2.14]{nolin2020communication}}~\cite{chattopadhyay2017lower}]
Let $f:\mathcal X \times \mathcal Y \rightarrow Z$ be a total function and $\mu$ be a distribution over $\mathcal X \times \mathcal Y$. The weak regularity of $f$ with respect to $\mu$, denoted by $\operatorname{wreg^{\mu}_\epsilon(f)}$, is the optimal value of the following optimization problem: 
\begin{align*}
    & \min \quad \delta \\
    & \text{st} \quad \quad \mu(R \cap f^{-1}(z)) \geq \frac{1}{|\mathcal Z|}(\mu(R) - \delta) & \forall R, z.
\end{align*}
Alternatively, 
\begin{align*}
    \operatorname{wreg}^{\mu}_\epsilon(f) = \max_{R, z} \mu(R) - |\mathcal Z| \mu(R \cap f^{-1}(z)). 
\end{align*}
We say that $f$ is $\delta$-weakly regular with respect to $\mu$ for any $\delta \geq \operatorname{wreg}^\mu(f)$. 
\end{defn}

\begin{fact}[{\cite[Corollary~2]{fontes2016relative}}]\label{facgt:Dmu_ardisc}
For any $\mu$, $f: \operatorname{supp(\mu)} \rightarrow \{0, 1\}$ and $\epsilon \in (0, 1/8)$, we have $\log \left(\operatorname{ardisc}^\mu_\epsilon(f) \right) \leq D^{\mu}_\epsilon(f) \leq \left(\log \operatorname{ardisc}^\mu_{\epsilon/8}(f) + 2\log \frac{1}{\epsilon} + 6\right)^2$. 
\end{fact}

\begin{fact}[{\cite[Proposition~1]{fontes2016relative}}]\label{fact:prt_rs}
For all $f, \mu, \epsilon$, we have $\operatorname{wprt}^\mu_\epsilon(f) \leq \operatorname{prt}^{+, \mu}_\epsilon(f) \leq \operatorname{prt}^\mu_\epsilon(f)$ and $\operatorname{wprt}^\mu_\epsilon(f) \leq \bar{\operatorname{prt}}^\mu_\epsilon(f) \leq \operatorname{prt}^\mu_\epsilon(f)$. 
\end{fact}

\begin{fact}[{\cite[Lemma~1]{jain2014quadratically}}]\label{fact:pprt-Rpub}
Let $f\subseteq \mathcal X \times \mathcal Y \times \mathcal Z$ be a relation and let $\epsilon > 0$. Then, $\log \left(\operatorname{pprt}_\epsilon(f)\right) \leq R^{\operatorname{pub}}_\epsilon(f)$. 
\end{fact}

\begin{fact}[{\cite[Theorem~1]{jain2014quadratically}}]\label{fact:Rpub-pprt}
    Let $f\subseteq \mathcal X \times \mathcal Y \times \mathcal Z$ be a relation and let $\epsilon > 0$. Then $R^{\operatorname{pub}}_{2\epsilon}(f) \leq \left(\log \operatorname{pprt}_\epsilon(f) + \log\frac{1}{\epsilon} + 1\right)^2$. 
\end{fact}

\begin{fact}[{\cite[Theorem~5]{fontes2016relative}}]\label{fact:prt+-rdisc}
Let $\mu$ be a distribution on $\mathcal X \times \mathcal Y$ and let $f$ be a Boolean function on the support of $\mu$ such that  either $\operatorname{rdisc}^{\mu}_\epsilon(f) \geq 1$ or $\operatorname{prt}^{+, \mu}_{4\epsilon}(f) > 2$ . Then, for any $\epsilon\in (0, 1/4)$, we have $\frac{\epsilon}{2} \operatorname{prt}^{+, \mu}_{4\epsilon}(f) \leq \operatorname{rdisc}^{\mu}_\epsilon(f)\leq \operatorname{prt}^{+, \mu}_\epsilon(f)$. 
\end{fact}

\begin{fact}[{\cite[Theorem~8]{fontes2016relative}}]\label{def:pprt-ardisc}
For any distribution $\mu$, for any function $f: \operatorname{supp}(f) \rightarrow \{0, 1\}$ and $\epsilon \in (0 ,1/4)$ such that either $\operatorname{ardisc}^\mu_\epsilon(f) \geq 1$ or $\operatorname{pprt}^\mu_{4\epsilon}(f) > 2$, we have $\frac{\epsilon}{2}\operatorname{pprt}^\mu_{4\epsilon} \leq \operatorname{ardisc}^\mu_\epsilon(f)\leq \operatorname{pprt}^\mu_\epsilon(f)$. 
\end{fact}

\begin{fact}[{\cite[Proposition~2.16]{nolin2020communication}}]\label{fact:discmu-wregmu}
Let $f:\mathcal X \times \mathcal Y \rightarrow \{0, 1\}$. Then, $\operatorname{wereg}^\mu(f) = \operatorname{disc}^\mu(f)$. 
\end{fact}

\begin{fact}[{\cite[Proposition~2.18]{nolin2020communication}}]\label{fact:wreg-wprt}\label{fact:wreg_wprt}
Let $f: \mathcal X \times \mathcal Y \rightarrow \mathcal Z$ be a total function. Then, $\operatorname{wprt}^{\mu}_\epsilon(f) \geq \frac{1 - \epsilon\cdot \frac{\mathcal Z|}{|\mathcal Z| - 1}}{\operatorname{wreg}^\mu(f)}$. 
\end{fact}

In the model of zero-communication protocols and inspired by the physics setup of Bell experiments, Laplante, Lerays and Roland~\cite{laplante2012classical} introduced a new lower bound on quantum communication complexity that subsumes the factorization norm~{\cite[Theorem~1]{linial2007lower}} when the players share a quantum sate. 
With slight modifications from~{\cite[Definition~4]{JK21}}, we define the quantum partition bound below. 


\begin{defn}[Quantum partition bound~{\cite[Definition~4]{JK21}}]
Let $f: \mathcal X \times \mathcal Y \rightarrow \{0, 1\}$ and let $\epsilon > 0$. Let $\perp$ be a special character. The quantum partition bound of $f$ with error $\epsilon$, denoted by $\operatorname{eff}^*_\epsilon(f)$, is given by the optimal value of the following non-linear program: 
\begin{align*}
    & \max \quad \quad \frac{1}{\eta} \\
    &  \text{st} \quad \sum_{(x, y): f(x, y)=1} q(f(x, y) | x, y)\geq (1 - \epsilon)\eta, & \forall x, y \in \mathcal X \times \mathcal Y\\
    & \quad \quad  \sum_{f(x, y)\in \{0, 1\}} q(f(x, y)| x, y) = \eta & \forall x, y\in \mathcal X \times \mathcal Y\\
    & \quad \quad \quad q(f(x,y)| x, y) \in \mathcal Q(\{0, 1\} \cup \perp, \mathcal X, \mathcal Y)
\end{align*}
\end{defn}


Based on the fact that the communication complexity of predicates generalizes the communication complexity of (total or partial) functions, the next result, originally stated for predicates, holds true. 
\begin{lem}\label{fact:eff*-gamma2alpha}
Let $f: \mathcal X \times \mathcal Y \rightarrow \{-1, 1\}$ and let $\epsilon > 0$. When $\alpha = \frac{1 + 2\epsilon}{1 - 2\epsilon}$, we have $\operatorname{eff}^*_\epsilon(f)\geq \gamma^\alpha_2(f)$.  
\end{lem}
\begin{proof}
For a total function $f: \mathcal X \times \mathcal Y \rightarrow \{-1, 1\}$, let $\vee_f$ denote the predicate on $(\{-1, 1\})^2 \times (\mathcal X \times \mathcal Y)$ given by $\vee(ab, xy) = 1 \Leftrightarrow a\cdot b = f(x, y)$. Note that for a function $f$, one can express $\gamma^\alpha_2(f)$ as  $\max_\mu \gamma^{\alpha, \mu}_2(f)$, where $\gamma^{\alpha, \mu}_2(f)$ is a distributional version of $\gamma^\alpha_2(f)$. Then, we have 
\begin{align*}
    & \operatorname{eff}^*_\epsilon(f) \\
    & \geq \max_\mu \operatorname{eff}^{*, \mu}_\epsilon(\vee_f) \geq (1 - 2\epsilon) \max_\mu \gamma^{\alpha, \mu}_2(f)  = (1 - 2\epsilon) \gamma^\alpha_2(f),
\end{align*} 
where the second inequality is due to~{\cite[Theorem~3]{JK21}}. 
\end{proof}


Given a function $f:\mathcal X\times \mathcal Y\rightarrow \{0, 1\}$, and error parameter $\epsilon$, the \emph{distributional information complexity} of $f$ with error $\epsilon$, denoted as $IC^\mu_\epsilon(f)$, is defined to be the infimum of the information cost (see Section~\ref{sec:preliminaries}) over all (randomized) protocols $\pi$ that achieve an error of at most $\epsilon$ with respect to the distribution $\mu$~{\cite[Definition~3.2]{braverman2012interactive}}. Mathematically speaking, 
\begin{align}\label{def:3.2}
IC^{\mu}_\epsilon(f) \coloneqq \displaystyle \inf_{\pi: \Pr_{(x, y)\sim \mu}[\pi(x, y) \neq f(x, y)]\leq \epsilon} IC^{\mu}(\pi). 
\end{align} 
For a function $f$ with error $\epsilon$, its \emph{max-distributional information complexity} is given by~{\cite[Definition~3.3]{braverman2012interactive} 
\begin{align}\label{def:3.3}
    IC_{D, \epsilon}(f) \coloneqq \max_{\mu \text{ a distribution on } \mathcal X \times \mathcal Y} IC^{\mu}_\epsilon(f),
\end{align}
whereas its \emph{information complexity} is defined as~{\cite[Definition~3.4]{braverman2012interactive}}]
\begin{align}\label{def:3.4}
    IC_\epsilon(f) \coloneqq \inf_{\substack{\pi \text{ is a protocol with } \\ \Pr[\pi(x, y) \neq f(x, y)]\leq \epsilon \text{ for all } (x, y)}} \max_{\mu} IC^\mu (\pi). 
\end{align}

It is easy to see that the information complexity of a function lower bounds its max-distributional information complexity. 

\begin{fact}[\cite{braverman2012interactive}]\label{fact:IC-ICD}
Let $f: \mathcal X \times \mathcal Y \rightarrow \{0, 1\}$ be a Boolean function and let $\epsilon > 0$. Then, $IC_\epsilon(f) \geq IC_{D, \epsilon}(f)$. 
\end{fact}
\begin{proof}
First, observe that for any distribution $\mu$, we have $IC^\mu(\pi) \geq IC^\mu_\epsilon(f)$. Then, taking the maximum over $\mu$ on both sides gives $\max_\mu IC^\mu(\pi) \geq \max_{\mu} IC^\mu_\epsilon(f)$. By Equations (\ref{def:3.2}) and (\ref{def:3.4}),we have
\begin{align*}
    IC_\epsilon(f) = \inf_\pi \max_\mu IC^\mu(\pi) = \inf_\pi IC^{\mu_1}(\pi) = IC^{\mu_1}_\epsilon(f),
\end{align*}
where $\mu_1$ is the distribution that maximizes $IC^{\mu}(\pi)$. From the observation made earlier, we can write 
\begin{align*}
    IC^{\mu_1}_\epsilon(f) \geq IC^{\mu_2}_\epsilon(f), 
\end{align*}
where $\mu_2$ is the distribution that maximizes $IC^{\mu}_\epsilon(f)$. By Equation (\ref{def:3.3}), we get
\begin{align*}
    C_{\mu_2}(f, \epsilon)  = \max_{\mu}  C_{\mu}(f, \epsilon) = IC_D(f, \epsilon). 
\end{align*}
\end{proof}

Together with the result below, max-distributional information complexity characterizes the information complexity.

\begin{fact}[{\cite[Theorem~3.5]{braverman2012interactive}}]\label{fact:3.5}
Let $f:\mathcal X\times \mathcal Y \rightarrow \{0, 1\}$ be any function, and let $\epsilon \geq 0$ be an error parameter. For all $0 < \alpha < 1$, we have $IC_{\epsilon/\alpha}\left(f\right)\leq \frac{IC_{D, \epsilon}(f)}{1 - \alpha}$. 
\end{fact}

Kerenidis et al. proved that the information complexity of a function $f$ is bounded from below by the relaxed partition bound~\cite{kerenidis2015lower}. 
\begin{fact}[{\cite[Theorem~1]{kerenidis2015lower}}]\label{fact:IC-prt}
    There exists a positive constant $C$ such that for all functions $f:\{0, 1\}^n \times \{0, 1\}^n \rightarrow \{0, 1\}$. all $\epsilon, \delta\in (0, \frac{1}{2}]$ and all distributions $\mu$, we have $IC^{\mu}_\epsilon(f) \geq \frac{\delta^2}{C}\cdot \left(\log \bar{\operatorname{prt}}^{\mu}_{\epsilon + 3\delta}(f) - \log 2 \right) - \delta$. 
\end{fact}


On the other hand, Braverman and Weinstein showed that information complexity can also be lower bounded by the discrepancy bound~\cite{braverman2016discrepancy}. A direct product theorem implicit in the work of ~\cite{bar2004information} gives an upper bound on information complexity in terms of the distributional communication complexity. 
\begin{fact}[{\cite[Theorem~1.1]{braverman2016discrepancy}}]\label{fact:discmu-ICmu}
Let $f: \mathcal X \times \mathcal Y \rightarrow \{0, 1\}$ be a Boolean function and let $\mu$ be any distribution on $\mathcal X \times \mathcal Y$. Then, $IC^{\mu}_{1/10}(f) \geq \Omega\left(\log \left(\frac{1}{\operatorname{disc}^{\mu}(f)}\right)\right)$. 
\end{fact}

\begin{fact}[{\cite[Theorem~2.4]{BarakBCR10}},~\cite{bar2004information}]\label{fact:ICpi_Dmu}
For every $f, \epsilon, \mu$, there exists a protocol $\pi$ that computes $f$ on inputs drawn from $\mu$ with probability of error at most $\epsilon$ and communication at most $D^{\mu^k}_\rho(f^k)$ such that $IC^{\mu}(\pi) \leq \frac{2D^{\mu^k}_\epsilon(f^k)}{k}$. 
\end{fact}

Intuitively,  the communication cost $CC(\pi)$ of a protocol $\pi$ is an upper bound on its information cost over any distribution $\mu$, since there cannot be more information leaked than the length of the messages exchanged. In other words, $IC^{\mu}(\pi) \leq CC(\pi)$. Another way of seeing this is by the observation that the information cost of a protocol $\pi$ is always bounded by its length $|\pi|$, and therefore communication complexity is always an upper bound on information complexity. 
\begin{fact}\label{fact:IC-CC}
For any function $f$ and error parameter $\epsilon > 0$, $IC_\epsilon(f) \leq R^{\operatorname{pub}}_\epsilon(f)$. 
\end{fact}
\begin{proof}
The public-coin randomized communication complexity of a function $f$ with error $\epsilon$  is defined as \begin{align}\label{eqn:Rpub}
R^{\operatorname{pub}}_\epsilon(f) = \displaystyle\min_\pi \max_{(x, y)\sim \mu} CC(\pi(x, y)).
\end{align}
According to Equation~(\ref{def:3.3}), its corresponding information complexity is 
\begin{align}\label{eqn:IC}
IC_\epsilon(f) = \displaystyle\inf_\pi\max_\mu IC^{\mu}(\pi)).
\end{align} 
Comparing Equations~(\ref{eqn:Rpub}) and~(\ref{eqn:IC}), we get $IC_\epsilon(f) \leq R^{\operatorname{pub}}_\epsilon(f)$ using the fact that information cost of a protocol lower bounds its communication cost.  
\end{proof}

Chakrabarti et al.~\cite{chakrabarti2001informational} were the first to define the external information cost explicitly for the purpose of studying communication complexity. The external information cost is the amount of information an observer necessarily has to learn when the parties perform the computation, and hence is always smaller than the communication complexity. The natural analogues of Equations (\ref{def:3.2}), (\ref{def:3.3}) and (\ref{def:3.4}) for external information complexity are as follows: given a function $f:\mathcal X\times \mathcal Y\rightarrow \{0, 1\}$, and error parameter $\epsilon$, the \emph{distributional external information complexity} is given by~{\cite[Definition~3.13]{braverman2012interactive}}
\begin{align}\label{def:3.13}
    IC^{\operatorname{ext}, \mu}_\epsilon(f) \coloneqq \displaystyle \inf_{\pi: \Pr_{(x, y)\sim \mu}[\pi(x, y)
    \neq f(x, y)]\leq \epsilon} IC^{\operatorname{ext}, \mu}(\pi), 
\end{align}  
while its \emph{max-distributional external information complexity} is defined as~{\cite[Definition~3.14]{braverman2012interactive}}
\begin{align}\label{def:3.14}
    IC^{\operatorname{ext}}_{D, \epsilon} (f) \coloneqq \max_{\mu \text{ a distribution on } \mathcal X \times \mathcal Y} IC^{\operatorname{ext}, \mu}_\epsilon(f). 
\end{align}
Finally, the \emph{external information complexity} of $f$ with error $\epsilon$ is~{\cite[Definition~3.15]{braverman2012interactive}}
\begin{align}\label{def:3.15}
    IC^{\operatorname{etc}}_\epsilon(f) \coloneqq \inf_{\substack{\pi \pi \text{ is a protocol with } \\ \Pr[\pi(x, y) \neq f(x, y)]\leq \epsilon \text{ for all } (x, y)}} \max_{\mu} IC^{\operatorname{ext}, \mu} (\pi). 
\end{align}

Braverman~\cite{braverman2012interactive} showed that the (internal) information complexity is always upper bounded by the external information complexity.  
\begin{fact}[{\cite[Proposition~3.12]{braverman2012interactive}}]\label{fact:3.12}
For every protocol $\pi$ and distribution $\mu$, $IC^{\operatorname{ext}, \mu}(\pi) \geq IC^{\mu}(\pi)$. 
\end{fact}

Fact~\ref{fact:3.12} implies the following relationship between internal and external information complexity. 

\begin{fact}[\cite{braverman2012interactive}]\label{fact:ext-int}
Let $f$ be a function. For all distributions $\mu$ ,  $IC^{\operatorname{ext}, \mu}_\epsilon(f) \geq IC^{\mu}_\epsilon(f), \quad  IC^{\operatorname{ext}}_{D, \epsilon}(f) \geq IC_{D, \epsilon}(f), \quad IC^{\operatorname{ext}}_\epsilon(f) \geq IC_\epsilon (f)$.
\end{fact}

The external information analogues of Facts~\ref{fact:IC-ICD} and~\ref{fact:3.5} are as below.

\begin{fact}[\cite{braverman2012interactive}]\label{fact:ICext-ICDext}
Let $f:\mathcal X \times \mathcal Y \rightarrow \{0, 1\}$ and let $\epsilon > 0$. Then, $IC^{\operatorname{ext}}_\epsilon(f) \geq IC^{\operatorname{ext}}_{D, \epsilon}(f)$. 
\end{fact}

\begin{fact}[{\cite[Theorem~3.16]{braverman2012interactive}}]\label{fact:3.16}
Let $f:\mathcal X\times \mathcal Y \rightarrow \{0, 1\}$ be any function, and let $\epsilon \geq 0$ be an error parameter. For each value of the parameter $0 < \alpha < 1$, we have $IC^{\operatorname{ext}}_{\epsilon/\alpha}\left(f\right)\leq \frac{IC^{\operatorname{ext}}_{D, \epsilon}(f)}{1 - \alpha}$. 
\end{fact}

Analogously to how the internal information cost of a protocol is upper bounded by its communication cost over any distribution on the inputs, the external information cost of a protocol is also upper bounded by its information cost, i.e.,  $IC^{\operatorname{ext}}_\mu(\pi) \leq CC(\pi)$. It follows that external information complexity is always upper bounded by communication complexity~\cite{BarakBCR10}. The proof is similar to that of Fact~\ref{fact:IC-CC}, where one replaces the definition of internal information complexity with its external analogue. 

\begin{fact}\label{fact:ICext-CC}
For any function $f$ and error parameter $\epsilon > 0$, $IC^{\operatorname{ext}}_\epsilon(f) \leq R^{\operatorname{pub}}_\epsilon(f)$. 
\end{fact}


The amortized communication complexity of a function $f$ with respect to distribution $\mu$, denoted as $AC^{\mu}_\epsilon(f)$, is defined as $AC^{\mu}_\epsilon(f) \coloneqq \lim_{k \rightarrow \infty} \frac{D^{\mu^k}_\epsilon(f^k)}{k}$ when the limit exists. In a recent paper, Braverman and Roo showed that information complexity exactly equals to amortized communication complexity~\cite{braverman2014information}.  
\begin{fact}[{\cite[Theorem~6.3]{braverman2014information}}]\label{fact:amortized}
For any function $f$, distribution $\mu$ and error $\epsilon$, $AC^{\mu}_\epsilon(f) = IC^{\mu}_\epsilon(f)$. 
\end{fact}

The partition bound $\operatorname{prt}_\epsilon(f)$ was shown to have an information-theoretic interpretation in terms of R\'{e}nyi entropy of order $\infty$~\cite{prabhakaran2015r}. Let $f$ be a relation, $\mu$ be a distribution over $\mathcal X \times \mathcal Y$ and $\epsilon > 0$. The R\'{e}nyi information cost of $f$ over $\mu$ with error $\epsilon$, denoted as $IC^{\infty, \mu}_\epsilon(f)$, is the minimum external information cost of a protocol $\pi$ that computes $f$ over $\mu$ with error at most $\epsilon$~{\cite[Definition~2.32]{nolin2020communication}}, i.e. 
\begin{align*}
    IC^{\infty, \mu}_\epsilon(f) = \min_{\pi \text{ computing } f \text{ with error } \epsilon} I_{\infty}(X, Y: \pi, R).   
\end{align*}
Here, $I_{\infty}(U: V) = \log \left(\sum_{v \in \mathcal V} \sum_{u\in \mathcal U} p_{V|U} (v|u)\right)$ is the R\'{e}nyi mutual information, where $U, V$ are random variables over unitaries $\mathcal U, \mathcal V$ and $p_{V|U} (v|u) = \frac{\Pr[U = u \wedge V = v]}{\Pr[U = u]}$~{\cite[Definition~2.31]{nolin2020communication}}.  

The pseudotranscript  complexity introduced by~\cite{prabhakaran2015r} is a complexity measure based on  combinatorial properties of a function’s truth table. Consider a function $f : \mathcal X \times \mathcal Y \rightarrow \mathcal Z$. Define a random variable $Q$ over a space $\mathcal Q$ to be a \emph{pseudotranscript} for $f$ if there exist two functions $g: \mathcal Q \times \mathcal X \rightarrow \mathbb R^+$ and $h: \mathcal Q \times \mathcal Y \rightarrow \mathbb R^+$ such that $\Pr[Q = q|X = x,Y = y] = g(q,x)h(q,y)$, for all $q \in \mathcal Q, x \in \mathcal X, y \in \mathcal Y$. Specifically, the transcript of a deterministic or private coin protocol is a pseudotranscript, and the transcript of a public coin protocol concatenated with the public coins is also a pseudotranscript. We say that a pseudotranscript computes a relation $f$ if there exists a mapping $\mathcal M: \mathcal Q \rightarrow \mathcal Z$ such that $\Pr[\mathcal M(\mathcal Q) \in f(X, Y)] \geq 1 - \epsilon$. The R\'{e}nyi information cost of a pseudotranscript is defined as a relaxation of $IC^{\infty, \mu}_\epsilon(f)$ by replacing the communication protocol with a pseudotranscript in its definition. 
\begin{defn}[R\'{e}nyi pseudoinformation cost of a relation~{\cite[Definition~2.36]{nolin2020communication}}]
Let $f$ be a relation, $\mu$ be a distribution and $\epsilon > 0$. Let inputs $X, Y$ be random variables and let $Q$ be a pseudotranscript relative to $X, Y$. The R\'{e}nyi pseudoinformation cost of $f$, denoted by $pIC^\infty_\epsilon(f, \mu)$, is given by 
\begin{align*}
    pIC^{\infty, \mu}_\epsilon(f) = \min_{Q \text{ computing } f \text {with error } \epsilon} IC^{\infty, \mu}_\epsilon(Q). 
\end{align*}
\end{defn}

Prabhakaran and Prabhakaran showed the equivalence between $pIC^\infty(f)$ and the partition bound~\cite{prabhakaran2015r}. 
\begin{fact}[{\cite[Theorem~2.37]{nolin2020communication}},~{\cite[Theorem~3]{prabhakaran2015r}}]\label{fact:pICmu-prtmu}
Let $f: \mathcal X \times \mathcal Y \rightarrow 2^{\mathcal Z}$ be any relation, $\mu$ be a distribution and $\epsilon\in [0, 1/2)$. Then, $pIC^{\infty, \mu}_\epsilon(f) = \log \operatorname{prt}^\mu_\epsilon(f)$. 
\end{fact}

\begin{fact}[{\cite[Theorem~2]{prabhakaran2015r}}]\label{fact:pIC-prt}
Let $f: \mathcal X \times \mathcal Y \rightarrow 2^{\mathcal Z}$ be any relation and let $\epsilon\in [0, 1/2)$. Then, $pIC^{\infty}_\epsilon(f) = \log \operatorname{prt}_\epsilon(f)$. 
\end{fact}

Touchette~\cite{touchette2015quantum} introduced quantum information complexity as an extension of (internal) information complexity to the quantum setting. This is inspired by quantum state-redistribution protocols~\cite{devetak2008exact, yard2009optimal}. Quantum information complexity has been shown to satisfy many of the natural properties possessed by information complexity. In particular, it is equal to amortized quantum communication complexity. We first review the definitions of a few measures relating to quantum information complexity before stating the results. 

We follow the quantum communication model defined in the work of Anshu et al.~\cite{anshu2017exponential}  , which is close to the one defined by Cleve and Buhrman~\cite{cleve1997substituting}, where players share prior entanglement and are allowed to send quantum messages to each other. In particular, a $k$-round protocol $\pi$ for a classical problem from input registers $A _{in} = X, B_{in} = Y$ to output registers $A_{out}, B_{out}$ is defined by a sequence of isometries $U_1, \cdots, U_{k+1}$ with a pure state $\psi$ shared between Alice and Bob. After the application of $U_i$, Alice holds register $A_i$, Bob holds register $B_i$ and the communication register is $C_i$. Readers are referred to~\cite{anshu2017exponential} for more details on the model. 


Let $\pi$ be a protocol and let $\mu$ be a distribution on $\mathcal X \times \mathcal Y$. The \emph{quantum communication cost} and \emph{quantum information complexity} of $\pi$ on input $\mu$, denoted by $QCC^\mu(\pi)$ and $QIC^\mu(\pi)$ respectively, are given by 
\begin{align*}
    & QCC^{\mu}(\pi) = \sum_i |C_i|\\
    & QIC^{\mu}(\pi) = \sum_{i \geq 1, \text{ odd}} I(C_i: R_X, R_Y | B_i) \\
    & \quad \quad \quad \quad \quad + \sum_{i \geq 1, \text{ even}} I(C_i:: R_X, R_Y | A_i),  
\end{align*}
where $R_X, R_Y$ are the canonical purification of input $X, Y$, $I(X:Z|Y) \coloneq \mathbb E_{y \leftarrow Y} [I(X:Z|Y = y)] = S(X|Y) + S(Z|Y) - S(XZ|Y)$ is the condition mutual information~\cite{anshu2017exponential}.  For any function $f$, any distribution $\mu$ and any error $\epsilon > 0$, 
\begin{align*}
    & QCC^{\mu}_\epsilon(f) = \min_\pi QCC^{\mu}(\pi)\\
    & QIC^{\mu}_\epsilon(f) = \min_\pi  QIC^{\mu}(\pi). 
\end{align*}
The \emph{max-distributional quantum information complexity} of $f$ error $\epsilon$ is~{\cite[Definition~4.1]{Braverman15}}
\begin{align}\label{def:QICD}
    QIC_{D, \epsilon}(f) = \max_{\mu \text{ is a distribution on } \mathcal X \times \mathcal Y} QIC^{\mu}_\epsilon(f). 
\end{align}

The quantum analogue of Fact~\ref{fact:IC-CC} is as below. 
\begin{fact}\label{fact:QIC-QCC}
For any function $f$, $QIC_\epsilon(f) \leq Q^*_\epsilon(f)$. 
\end{fact}

Let $f$ be a function, $\mu$ be a distribution on the inputs $\mathcal X \times \mathcal Y$ and $\epsilon > 0$. The \emph{amortized quantum communication complexity} of a function $f$ under the input distribution $\mu$ and with error $\epsilon > 0$, denoted by $AQCC^{\mu}_e\space(f)$ is given by
\begin{align*}
AQCC^{\mu}_\epsilon(f) \coloneqq \lim_{k \rightarrow \infty} QCC^{\mu^k}_\epsilon(f^k).  
\end{align*}

\begin{fact}[\cite{Braverman15}]\label{fact:QIC-QICD}
For any $f$ and $\epsilon > 0$, we have $QIC_{D, \epsilon}(f) \leq QIC_\epsilon(f)$. 
\end{fact}

\begin{fact}[{\cite[Theorem~4.13]{Braverman15}}]\label{fact:QIC-QICD-alpha}
Let $f$  be a relation and $\alpha, \epsilon\in(0,1)$. Then, $QIC_{\epsilon/\alpha}\left(f\right) \leq \frac{QIC_{D, \epsilon}(f)}{1- \alpha}$. 
\end{fact}

\begin{fact}[{\cite[Theorem~2]{touchette2015quantum}}]\label{fact:AQCC-QCC}
For any function $f$, distribution $\mu$ and error $\epsilon > 0$, we have $AQCC^{\mu}_\epsilon(f) = QIC^{\mu}_\epsilon(f)$. 
\end{fact}

\begin{fact}[{\cite[Corollary~5.8]{braverman2018near}}]\label{fact:QCC-QIC}
For all Boolean functions $f$, $QCC_{1/3}(f) \leq 2^{O(QIC_{1/3}(f)+1)}$. 
\end{fact}

\begin{fact}[\cite{anshu2017exponential}]\label{fact:AQCC-AC}
For any function $f$,  distribution $\mu$ and error $\epsilon > 0$, we have $AQCC^{\mu}_\epsilon(f) \leq AC^{\mu}_\epsilon(f)$. 
\end{fact}

\begin{fact}[{\cite[Theorem~5.7]{Braverman15}}]\label{fact:QICD-sdisc}
There exists an absolute constant $\delta > 0$ such that for any Boolean function $f: \mathcal X\times \mathcal Y \rightarrow \{0, 1\}$, we have $QIC_{D, \delta}(f) \geq \Omega(\operatorname{sdisc}_{1/5}(f) - O(1))$. 
\end{fact}

\subsection{Class 4}
This class consists of only two complexity measures. Namely, the unbounded error probabilistic communication complexity and the sign-rank.  This constant equivalence relationship is shown in Figure~\ref{fig:C4}. 
\begin{figure}[H]
\centering
\begin{tikzpicture}[->,>=Stealth,auto,node distance=2.5cm, thick]
    \node (UPP) {$UPP(f)$};
    \node (signrank) at (4, 0) {$\operatorname{signrank}(f)$};

    \draw [->] (UPP) to node[midway, above] {Fact~\ref{UPP_signrank}}  (signrank);
    \draw [->] (signrank) to (UPP);
\end{tikzpicture}
\caption{Relationship between complexity measures in Class 4.}
    \label{fig:C4}
\end{figure}

\begin{theorem}\label{thm:4}
Class 4 includes the following complexity measures: 
\begin{multicols}{2}
    \begin{itemize}
        \item $UPP(f)$
        \item $\operatorname{signrank}(f)$
    \end{itemize}
\end{multicols}
\end{theorem}
\begin{proof}
    We have $UPP(f)\leftrightarrow \operatorname{signrank}(f)$ by Fact~\ref{UPP_signrank}. 
\end{proof}

Before stating the fact used in the proof of Theorem~\ref{thm:4}, we give definitions of the sign-rank and the unbounded error probabilistic communication complexity. 

Sign-rank is one of the important analytic notions in communication complexity. The sign-rank of a matrix $M$ is the minimal rank of a real matrix whose entries have the same sign pattern as $M$. 
\begin{defn}[sign-rank~{\cite[Definition~1]{alon2016sign}}]\label{def:signrank}
For a real matrix $M$ with nonzero entries, let $\operatorname{sign}(M)$ denote its sign matrix such that $(\operatorname{sign}(M))_{i, j} = \operatorname{sign}(M_{i, j})$ for all $i, j$. The sign-rank of a matrix $S$, denoted by $\operatorname{signrank}(S)$, is defined as $\operatorname{signrank}(S) \coloneqq \min \{\operatorname{rank}(M): \operatorname{sign}(M)  = S\}$.
\end{defn}

The unbounded error probabilistic communication complexity was introduced by Paturi and Simon~\cite{Paturi1984}. Given a function $f:\{0, 1\}^n\times\{0, 1\}^n\rightarrow \{0, 1\}$, an unbounded error protocol for $f$ is a protocol $\pi$ such that $\Pr\left[\pi(x, y) = f(x, y)\right]>1/2$, where the probability is over the coin tosses of the players. The worst-case cost of $\pi$ is the maximum number of bits transmitted  during an execution of $\pi$, taken over all choices of $x$ and $y$. The unbounded error communication complexity of $f$, denoted as $UPP(f)$, is the minimum cost of an unbounded error protocol for $f$. In their work, Paturi and Simon showed the equivalence relation between $UPP(f)$ and the sign-rank of a matrix.

\begin{fact}[{\cite[Theorem~3]{Paturi1984}}]\label{UPP_signrank}
Let $f: \mathcal X \times \mathcal U \rightarrow \{0, 1\}$. Then, we have $\lceil \log \operatorname{signrank}(f)\rceil\leq UPP(f)\leq \lceil \log \operatorname{signrank}(f)\rceil + 1$
\end{fact}

 \subsection{Class 5}
This class consists of measures from learning complexity and complexity measures under the product distribution. A product distribution $\mu$ on $\mathcal X\times \mathcal Y$ is a distribution that can be expressed as $\mu(x, y) =\mu_{\mathcal X}(\mathcal X)\times \mu_{\mathcal Y}(\mathcal Y)$, where $\mu_{\mathcal X}$ and $\mu_{\mathcal Y}$ are distributions over $\mathcal X$ and $\mathcal Y$, respectively. The elements in this class consist of complexity measures under product distribution. We use $[]$ at the superscript to denote the corresponding complexity measure under product distributions. For instance, $\mathcal{C}^{[]}(f)$ denotes the complexity measure $\mathcal{C}$ under product distributions. The constant equivalence relationship between these complexity measures are presented in  Figure~\ref{fig:C5}. 

\begin{figure*}[t]
\centering
\begin{tikzpicture}[->,>=Stealth,auto,node distance=2.5cm, thick]
     \node (DABprod) {$D^{A\rightarrow B, []}_\epsilon(f)$};
     \node (RABprod) [right of=DABprod, xshift=1cm] {$R^{A\rightarrow B, []}_\epsilon(f)$};
     \node (recABprod) [right of=RABprod, xshift=3cm] {$\operatorname{rec}^{A \rightarrow B, []}_\epsilon(f)$};
     \node (QABprod) [right of=recABprod, xshift=2cm]  {$Q^{A\rightarrow B, []}_\epsilon(f)$};
     \node (VC) [below of=QABprod] {$VC(f)$};
     \node (sq) [below of=VC, yshift=-10cm] {$sq(f)$};
     \node (discprod) [left of=sq, xshift=-2cm] {$\operatorname{disc}^{[]}(f)$};
     \node (ICprod) [left of=discprod, xshift=-3cm] {$IC^{[]}_\epsilon(f)$};
     \node (ACprod) [above of=ICprod, yshift=2cm] {$AC^{[]}_\epsilon(f)$};
     \node (ICprodpi) [left of=ICprod, xshift=-1cm] {$IC^{[]}(\pi)$};
     \node (Dprod) [above of=ICprodpi, yshift=7cm] {$D^{[]}_\epsilon(f)$}; 
     \node (ICextprodpi) [above right of=ICprodpi, yshift=0.5cm, xshift=-0.5cm] {$IC^{\operatorname{ext},[]}(\pi)$};
     \node (ICextprod) [above right of=ICprod, yshift=0.5cm]
     {$IC^{\operatorname{ext},[]}_\epsilon(f)$};
     \node (subABprod) [below right of=RABprod, xshift=1cm, yshift=-0.5cm]  {$\operatorname{sub}^{A \rightarrow B, []}_{\mathcal Y, \epsilon}(f)$};
     \node (wregprod) [above of=discprod] {$\operatorname{wreg}^{[]}(f)$};
     \node (wprtprod) [above of=wregprod] {$\operatorname{wprt}^{[]}_\epsilon(f)$};
     \node (prt+prod) [above left of=wprtprod] {$\operatorname{prt}^{+,[]}_\epsilon(f)$};
     \node (rdiscprod) [left of=prt+prod,  xshift=-1cm] {$\operatorname{rdisc}^{[]}_\epsilon(f)$};
     \node (barprtprod) [above right of=wprtprod] {$\bar{\operatorname{prt}}^{[]}_\epsilon(f)$};
     \node (prtprod) [above of=wprtprod, yshift=2cm] {$\operatorname{prt}^{[]}_\epsilon(f)$};
     \node (pICprod) [above of=prtprod] {$pIC^{\infty, []}_\epsilon(f)$};
     \node (ardiscprod) [below right of=Dprod, yshift=-0.5cm] {$\operatorname{ardisc}^{[]}_\epsilon(f)$};
     \draw [->] (DABprod) to node[midway, above] {Fact~\ref{fact:RAB-DAB}} (RABprod);
     \draw [->] (RABprod) to node[midway, above] {Fact~\ref{fact:RAB-rec}} (recABprod);
      \draw  [->] (recABprod) to node[midway, above] {Fact~\ref{fact:rec-QAB}} (QABprod);
      \draw [->] (QABprod) to node[midway, left] {Facts~\ref{fact:QAB-RAB},~\ref{fact:vc-RAB}} (VC);
      \draw [->] (VC)to node[midway, left] {Fact~\ref{fact:vc-sq}} (sq);
      \draw [->] (sq) to node[midway, above] {Fact~\ref{fact:sq-disc}} (discprod);
      \draw [->] (Dprod) to node[near end, right] {Fact~\ref{fact:Dmu-DmuAB_prod}} (DABprod);
      \draw [->] (ICprodpi) to node[midway, right] {Fact~\ref{fact:ICpi_Dmu_prod}} (Dprod);
      \draw [->] (ICprod) to node[midway, above] {Eq.(\ref{def:3.2})} (ICprodpi);
      \draw [->] (discprod) to node[midway, above] {Fact~\ref{fact:disc_IC_prod}} (ICprod);
      \draw [->] (ICprodpi) to node[midway, right ] {Fact~\ref{fact:pi_int_ext_prod}} (ICextprodpi);
      \draw [->] (ICextprodpi)  to node[near end] {} (ICprodpi);
      \draw [->] (ICextprod) to node[near start, right] {Fact~\ref{fact:IC_int_ext} } (ICprod);
      \draw [->] (ACprod) to node[near start, right] {Fact~\ref{fact:amortized_prod}} (ICprod);
      \draw [->] (ICprod) to node[midway, left] {} (ACprod);
      \draw [->] (ICprod) to node[near end, left] {} (ICextprod);
      \draw [->] (RABprod) to node[midway, left] {Fact~\ref{fact:RAB-subAB_prod}} (subABprod);
      \draw [->] (subABprod) to node[midway, left] {} (RABprod);
      \draw [->] (discprod) to node[midway, right] {Fact~\ref{fact:discmu-wregmu_prod}} (wregprod);
      \draw [->] (wregprod) to node[midway, right] {Fact~\ref{fact:wreg_wprt_prod}} (wprtprod);
      \draw [->] (wprtprod) to node[midway, left] {Fact~\ref{fact:prt_rs_prod}} (prt+prod);
      \draw [->]  (wprtprod) to node[midway, right] {Fact~\ref{fact:prt_rs_prod}} (barprtprod);
      \draw [->] (prt+prod) to node[midway, left] {Fact~\ref{fact:prt_rs_prod}}  (prtprod);
      \draw [->] (barprtprod) to node[midway, right] {Fact~\ref{fact:prt_rs_prod}}  (prtprod);
      \draw [->] (prtprod)  to node[midway, above] {Fact~\ref{fact:Dmu-prt_prod}} (Dprod);
      \draw [->] (prtprod) to node[midway, right] {Fact~\ref{fact:pICmu-prtmu_prod}} (pICprod);
      \draw [->] (pICprod) to node[midway, right] {} (prtprod);
      \draw [->] (Dprod) to node[midway, right] {Fact~\ref{facgt:Dmu_ardisc-_prod}} (ardiscprod);
      \draw [->] (ardiscprod) to node[midway, left] {} (Dprod);
      \draw [->] (rdiscprod) to node[midway, above] {Fact~\ref{fact:prt+-rdisc_prod}} (prt+prod);
      \draw [->] (prt+prod) to node[midway, above] {} (rdiscprod);
\end{tikzpicture}
\caption{Relationship between complexity measure in Class 5.}
\label{fig:C5}
\end{figure*}

\begin{theorem}\label{thm:5}
Class 5 includes the following complexity measures: 
\begin{multicols}{3}
    \begin{itemize}
        \item $D^{[]}_\epsilon(f)$
        \item $D^{A \rightarrow B, []}_\epsilon(f)$
        \item $R^{A \rightarrow B, []}_\epsilon(f)$
        \item $Q^{A \rightarrow B, []}_\epsilon(f)$
        \item $\operatorname{rec}^{A \rightarrow B, []}_\epsilon(f)$
        \item $\operatorname{sub}^{A \rightarrow B, []}_{\mathcal Y}(f, \epsilon)$
        \item $VC(f)$
        \item $sq(f)$
        \item $\operatorname{disc}^{[]}(f)$
        \item $IC^{[]}(\pi)$
        \item $IC^{[]}_\epsilon(f)$
        \item  $IC^{\operatorname{ext},[]}(\pi)$
        \item  $IC^{\operatorname{ext},[]}_\epsilon(f)$
        \item $pIC^{\infty,[]}_\epsilon(f)$
        \item $AC^{[]}_\epsilon(f)$
        \item $\operatorname{rdisc}^{[]}_\epsilon(f)$
        \item $\operatorname{ardisc}^{[]}_\epsilon(f)$
        \item $\operatorname{prt}^{[]}_\epsilon(f)$
        \item $\operatorname{prt}^{+, []}_\epsilon(f)$
        \item $\bar{\operatorname{prt}}^{[]}_\epsilon(f)$
        \item $\operatorname{wprt}^{[]}_\epsilon(f)$
        \item $\operatorname{wreg}_\epsilon(f)$
    \end{itemize}
\end{multicols}
\end{theorem}

\begin{proof}
First, we have $D^{A \rightarrow B, []}_\epsilon(f) \rightarrow R^{A \rightarrow B, []}_\epsilon(f)$ by Fact~\ref{fact:RAB-DAB}; $R^{A \rightarrow B, []}_\epsilon(f) \rightarrow \operatorname{rec}^{A \rightarrow B, []}_\epsilon(f)$ by Fact~\ref{fact:RAB-rec};  $\operatorname{rec}^{A \rightarrow B, []}_\epsilon(f) \rightarrow Q^{A \rightarrow B, []}_\epsilon(f)$ by Fact~\ref{fact:rec-QAB}; $Q^{A \rightarrow B, []}_\epsilon(f) \rightarrow VC(f)$ by Facts~\ref{fact:QAB-RAB} and~\ref{fact:vc-RAB}; $VC(f) \rightarrow sq(f)$ by Fact~\ref{fact:vc-sq}; $sq(f) \rightarrow \operatorname{disc}^{[]}(f)$ by Fact~\ref{fact:sq-disc}; $\operatorname{disc}^{[]}(f) \rightarrow IC^{[]}_\epsilon(f)$ by Fact~\ref{fact:disc_IC_prod}; $IC^{[]}_\epsilon(f) \rightarrow IC^{[]}(\pi)$ by Equation~\ref{def:3.2}; $IC^{[]}(\pi) \rightarrow D^{[]}_\epsilon(f)$ by Fact~\ref{fact:ICpi_Dmu_prod}. 

From $\operatorname{disc}^{[]}(f)$
, we have $\operatorname{disc}^{[]}(f) \rightarrow \operatorname{wreg}^{[]}(f)$ by Fact~\ref{fact:discmu-wregmu_prod}; 
$\operatorname{wreg}^{[]}(f) \rightarrow \operatorname{wprt}^{[]}_\epsilon(f)$ by Fact~\ref{fact:wreg_wprt_prod};
$\operatorname{wprt}^{[]}_\epsilon(f) \rightarrow \operatorname{prt}^{+, []}_\epsilon(f), \bar{\operatorname{prt}}^{[]}_\epsilon(f)$ by Fact~\ref{fact:prt_rs_prod};
$\operatorname{prt}^{+, []}_\epsilon(f), \bar{\operatorname{prt}}^{[]}_\epsilon(f) \rightarrow \operatorname{prt}^{[]}_\epsilon(f)$ by Fact~\ref{fact:prt_rs_prod};
$\operatorname{prt}^{[]}_\epsilon(f) \rightarrow D^{[]}_\epsilon(f)$ by Fact~\ref{fact:Dmu-prt_prod}. 

In addition, $R^{A \rightarrow B, []}_\epsilon(f) \leftrightarrow \operatorname{sub}^{A \rightarrow B, []}_{\mathcal Y, \epsilon}(f)$ by Fact~\ref{fact:RAB-subAB_prod}; 
$D^{[]}_\epsilon(f) \leftrightarrow \operatorname{ardisc}^{[]}_\epsilon(f)$ by Fact~\ref{facgt:Dmu_ardisc-_prod};
$\operatorname{prt}^{[]}_\epsilon(f) \leftrightarrow pIC^{\infty, []}_\epsilon(f)$ by Fact~\ref{fact:pICmu-prtmu_prod};
$\operatorname{prt}^{+, []}_\epsilon(f) \leftrightarrow \operatorname{rdisc}^{[]}_\epsilon(f)$ by Fact~\ref{fact:prt+-rdisc_prod};
$IC^{[]}(\pi) \leftrightarrow IC^{\operatorname{ext},[]}(\pi)$ by Fact~\ref{fact:pi_int_ext_prod};
$IC^{[]}_\epsilon(f) \leftrightarrow IC^{\operatorname{ext},[]}_\epsilon(f)$ by Fact~\ref{fact:IC_int_ext};
$IC^{[]}_\epsilon(f) \leftrightarrow AC^{[]}_\epsilon(f)$ by Fact~\ref{fact:amortized_prod}. 
\end{proof}

Now, we state facts used in the proof of Theorem~\ref{thm:5}. We begin by listing results from Section~\ref{sec:useful_lemma_class_3} that hold for product distribution. 

\begin{fact}[{\cite[Proposition~2.1] {Sherstov2010}, \cite[Proposition~3.28]{KNisan96}}]\label{fact:disc-Dmu_prod}
For every function $f:\mathcal X \times \mathcal Y \rightarrow \{0, 1\}$ and every $\epsilon \geq 0$, $D^{[]}_{1/2 - \epsilon}(f) \geq \log \left(\frac{2\epsilon}{\operatorname{disc}^{[]}(f)}\right)$. 
\end{fact}

\begin{fact}[{\cite[Corollary~1.4]{kerenidis2015lower}}]\label{fact:Dmu-prt_prod}
For every $\epsilon > 0$ and every $f:\mathcal X\times \mathcal Y \rightarrow \{0, 1\}$,  $D^{[]}_\epsilon(f) \geq \log \bar{\operatorname{prt}}^{[]}_\epsilon(f)$. 
\end{fact}

\begin{fact}[{\cite[Corollary~2]{fontes2016relative}}]\label{facgt:Dmu_ardisc-_prod}
For any $f: \operatorname{supp(\mu)} \rightarrow \{0, 1\}$ and $\epsilon \in (0, 1/8)$, we have $\log \left(\operatorname{ardisc}^{[]}_\epsilon(f) \right) \leq D^{[]}_\epsilon(f) \leq \left(\log \operatorname{ardisc}^{[]}_{\epsilon/8}(f) + 2\log \frac{1}{\epsilon} + 6\right)^2$. 
\end{fact}

\begin{fact}[{\cite[Proposition~1]{fontes2016relative}}]\label{fact:prt_rs_prod}
For all $f, \epsilon$, we have $\operatorname{wprt}^{[]}_\epsilon(f) \leq \operatorname{prt}^{+, []}_\epsilon(f) \leq \operatorname{prt}^{[]}_\epsilon(f)$ and $\operatorname{wprt}^{[]}_\epsilon(f) \leq \bar{\operatorname{prt}}^{[]}_\epsilon(f) \leq \operatorname{prt}^{[]}_\epsilon(f)$. 
\end{fact}

\begin{fact}[{\cite[Theorem~5]{fontes2016relative}}]\label{fact:prt+-rdisc_prod}
Let $[]$ be a product distribution on $\mathcal X \times \mathcal Y$ and let $f$ be a Boolean on the support of $\mu$ such that  either $\operatorname{rdisc}^{[]}_\epsilon(f) \geq 1$ or $\operatorname{prt}^{+, []}_{4\epsilon}(f) > 2$ . Then, for any $\epsilon\in (0, 1/4)$, we have $\frac{\epsilon}{2} \operatorname{prt}^{+, []}_{4\epsilon}(f) \leq \operatorname{rdisc}^{[]}_\epsilon(f)\leq \operatorname{prt}^{+, []}_\epsilon(f)$. 
\end{fact}

\begin{fact}[{\cite[Theorem~6]{fontes2016relative}}]\label{def:pprt-ardisc_prod}
For any distribution $\mu$, for any function $f: \operatorname{supp}(f) \rightarrow \{0, 1\}$ and $\epsilon \in (0 ,1/4)$ such that either $\operatorname{ardisc}^{[]}_\epsilon(f) \geq 1$ or $\operatorname{pprt}^{[]}_{4\epsilon}(f) \geq 2$, we have $\frac{\epsilon}{2}\operatorname{pprt}^{[]}_{4\epsilon} \leq \operatorname{ardisc}^{[]}_\epsilon(f)\leq \operatorname{pprt}^{[]}_\epsilon(f)$. 
\end{fact}

\begin{fact}[{\cite[Proposition~2.16]{nolin2020communication}}]\label{fact:discmu-wregmu_prod}
Let $f:\mathcal X \times \mathcal Y \rightarrow \{0, 1\}$. Then, $\operatorname{wereg}^{[]}(f) = \operatorname{disc}^{[]}(f)$. 
\end{fact}

\begin{fact}[{\cite[Proposition~2.18]{nolin2020communication}}]\label{fact:wreg_wprt_prod}
Let $f: \mathcal X \times \mathcal Y \rightarrow \mathcal Z$ be a total function. Then, $\operatorname{wprt}^{[]}_\epsilon(f) \geq \frac{1 - \epsilon\cdot \frac{\mathcal Z|}{|\mathcal Z| - 1}}{\operatorname{wreg}^{[]}(f)}$. 
\end{fact}

\begin{defn}[Two-way product subdistribution bounds~{\cite[Definition~3.3]{jain2008direct}}]
For a distribution$\mu$µ over $\mathcal X\times \mathcal Y$, define $\operatorname{sub}^{\mu}_\epsilon(f) \coloneqq \min_{\lambda} S^{\infty}(\lambda \Vert \mu)$, where $\lambda$ is taken over all distributions which are both SM-like (simultaneous-message-like)\footnote{We say that $\lambda$ is SM-like for $\mu$ if there exists a distribution $\chi$ on $\mathcal X \times \mathcal Y$ such that $\chi$ is one-message-like (recall definition from Footnote~\ref{foot:one-message})for $\mu$ with respect to $\mathcal X$ and $\lambda$ is one-message-like for $\chi$ with respect to  $\mathcal Y$.} for $\mu$ and $\epsilon$-monochromatic for $f$. When the minimization is restricted to distributions that are SM-like
and $(\epsilon, z)$-monochromatic for some fixed value $z \in \mathcal Z$, we denote the resulting measure as $\operatorname{sub}^{\mu, z}_\epsilon(f)$.
The two-way product subdistribution bound, denoted as $\operatorname{sub}^{[]}_\epsilon(f) \coloneqq \max_{\mu} \operatorname{sub}^{\mu}_\epsilon(f)$, where $\mu$ ranges over all product distributions on $\mathcal X \times \mathcal Y$.
\end{defn}

In the definition above, we say that a distribution $\lambda$ is $(\epsilon, z)$-monochromatic for $f$ if $\Pr_{X,Y\sim \lambda}[(X, Y, z) \in f]\geq 1 - \epsilon$. The distribution $\lambda$ is $\epsilon$-monochromatic for $f$ if it is $(\epsilon, z)$-monochromatic for $f$ for some $z\in \mathcal Z$. 

\begin{fact}[{\cite[Lemma~3.1]{Jain2008}}]
Let $f \subseteq \mathcal X \times \mathcal Y \times \mathcal Z$ be a relation and 
let $\delta \in (0, 1)$. Then, $\operatorname{rec}^{[]}_\epsilon(f) \geq \operatorname{sub}^{[]}_\epsilon(f) \geq \operatorname{rec}^{[]}_{\epsilon/\delta^2}(f) - \log \frac{1}{(1 - \delta)^2}$. 
\end{fact}

\begin{fact}[{\cite[Theorem~1.1]{braverman2016discrepancy}}]\label{fact:disc_IC_prod}
Let $f: \mathcal X \times \mathcal Y \rightarrow \{0, 1\}$ be a Boolean function. Then, $IC^{[]}_{1/10}(f) \geq \log \left(\frac{1}{\operatorname{disc}^{[]}(f)}\right)$. 
\end{fact}

\begin{fact}[{\cite[Theorem~2.4]{BarakBCR10}}]\label{fact:ICpi_Dmu_prod}
For every $f, \epsilon$, there exists a protocol $\pi$ that computes $f$ on inputs drawn from a product distribution $\mu$ with probability of error at most $\epsilon$ and communication at most $D^{\mu^k}_\rho(f^k)$ such that $IC^{[]}(\pi) \leq \frac{2D^{[]^k}_\epsilon(f^k)}{k}$. 
\end{fact}

\begin{fact}[{\cite[Theorem~1.1]{kerenidis2015lower}}]\label{fact:IC-prt_prod}
    There exists a positive constant $C$ such that for all functions $f:\{0, 1\}^n \times \{0, 1\}^n \rightarrow \{0, 1\}$. all $\epsilon, \delta\in (0, \frac{1}{2}]$, we have $IC^{[]}_\epsilon(f) \geq \frac{\delta^2}{C}\cdot \left(\log \bar{\operatorname{prt}}^{[]}_{\epsilon + 3\delta}(f) - \log 2 \right) - \delta$. 
\end{fact}

\begin{fact}[{\cite[Theorem~6.3]{braverman2014information}}]\label{fact:amortized_prod}
For any function $f$, distribution $\mu$ and error $\epsilon$, $AC^{[]}_\epsilon(f) = IC^{[]}_\epsilon(f)$. 
\end{fact}

\begin{fact}[{\cite[Theorem~2]{touchette2015quantum}}]\label{fact:AQCC-QCC_prod}
For any function $f$, distribution $\mu$ and error $\epsilon > 0$, we have $AQCC^{[]}_\epsilon(f) = QIC^{[]}_\epsilon(f)$. 
\end{fact}

\begin{fact}[\cite{anshu2017exponential}]\label{fact:AQCC-AC_prod}
For any function $f$ and error $\epsilon > 0$, we have $AQCC^{[]}_\epsilon(f) \leq AC^{[]}_\epsilon(f)$. 
\end{fact}

\begin{fact}[{\cite[Theorem~2.37]{nolin2020communication}}]\label{fact:pICmu-prtmu_prod}
Let $f: \mathcal X \times \mathcal Y \rightarrow 2^{\mathcal Z}$ be any relation and let $\epsilon\in [0, 1/2)$. Then, $pIC^{\infty, []}_\epsilon(f) = \log \operatorname{prt}^{[]}_\epsilon(f)$. 
\end{fact}

The fact that quantum communication complexity is cheaper than randomized communication complexity extends to the one-way communication model under product distributions. Moreover, by definition of distributional communication complexity, the deterministic version of Fact~\ref{fact:Rpub-RpubAB} (see also~{\cite[Exercise~4.2.1]{KNisan96}}) naturally extends to distributional communication complexity under product distribution. In addition, the one-way rectangle bound over product distributions serves as an upper bound on $R^{A\to B, []}_{2\epsilon}(f)$ and a lower bound on $Q^{A\to B, []}_{\epsilon}(f)$ under certain conditions. 

\begin{fact}\label{fact:QAB-RAB}
Let $f: \mathcal X \times \mathcal Y \rightarrow \{0, 1\}$. Then, $Q^{A\rightarrow B, []}_\epsilon(f)\leq R^{A\rightarrow B, []}_\epsilon(f)$. 
\end{fact}
\begin{proof}
In order to simulate a private coin randomized protocol with a quantum protocol, Alice and Bob each prepare the state $\displaystyle\sum_i \sqrt{p(i)}\ket{i}$ and measure it locally. The rest of the protocol proceeds by replacing classical transformations with their quantum counterparts. 
\end{proof}

\begin{fact}\label{fact:Dmu-DmuAB_prod}
Let $f: \mathcal X \times \mathcal Y \rightarrow \{0, 1\}$. We have $D^{A \rightarrow B, []}_\epsilon(f) \geq \log D^{[]}_\epsilon(f)$. 
\end{fact}

\begin{fact}[{\cite[Theorem~4]{Jain2009}}]\label{fact:rec-QAB}
Let $f:\mathcal X \times \mathcal Y \rightarrow \mathcal Z$ be a total function and let $\epsilon\in(0, 1/2)$. Also, let $\operatorname{rec}^{A\to B, []}_\epsilon(f) > 2\log(1/\epsilon)$. Then, $Q^{A\to B, []}_{\epsilon^{3/8}}(f) \geq \Omega(\operatorname{rec}^{A\to B, []}_\epsilon(f))$. 
\end{fact}

\begin{fact}[~\cite{Jain2009, Jain2008}]\label{fact:RAB-rec}
Let $\epsilon\in(0, 1/4)$, $R^{A\to B, []}_{2\epsilon}(f) = O(\operatorname{rec}_\epsilon^{A\to B, []}(f))$.
\end{fact}


Now, we introduce more complexity measure related to learning theory. In learning theory, a \emph{concept class} $\mathtt C$ is any set of functions $f:\mathcal X\times \mathcal Y \rightarrow \{-1, 1\}$ for a finte set $\mathcal X \times \mathcal Y$. Ths \emph{statistical query (SQ) dimension} under a given distribution $\mu$, denoted as $\operatorname{sqdim}_\mu(\mathtt C)$, is defined to be the largest $d\in\mathbb Z_+$ for which there are functions $f_1, f_2, \cdots, f_d\in\mathtt C$ and a probability distribution $\mu$ on $\mathcal X\times\mathcal Y$ such that $\left\vert \mathbb E_{(x, y)\sim \mu} \left[f_i(x, y)f_j(x, y)\right]\right\vert\leq\frac{1}{d}$, for all $i\neq j$. We denote $\operatorname{sqdim}(\mathtt C)\coloneqq \max_{\mu} \operatorname{sqdim}_{\mu}(\mathtt C)$~\cite{Sherstov2008}. 

We identify $\mathtt C$ with the sign matrix $M$, whose rows are indexed by the functions of $\mathcal C$ and columns indexed by the inputs $(x, y)\in \mathcal X\times \mathcal Y$, and entries given by $M_{f, (x, y)} = f(x, y)$. For a matrix $M\in\{-1, 1\}^{n\times n}$, we define its SQ dimension, $\operatorname{sqdim}(M)$, to be the SQ dimension of the rows of $M$ viewed as functions $\{1, 2,\cdots, n\}\rightarrow \{-1, 1\}$. The \emph{statistical query complexity} $\operatorname{sq}(M)$ of $M$ as the SQ dimension of its rows, i.e. $\operatorname{sq}(M) \coloneqq \operatorname{sqdim}(\{f_1,...,f_n\})$, where $f_1, \cdots, f_n$ are the rows of $M$.

Sherstov~\cite{Sherstov2008} showed that statical query complexity characterizes the discrepancy. 
\begin{fact}[{\cite[Theorem~7.1]{Sherstov2008}}]\label{fact:sq-disc}
Let $M\in\{-1, 1\}^{m \times n}$. Then, $\sqrt{\frac{1}{2}\operatorname{sq}(M)} < \frac{1}{\operatorname{disc}^{[]} (M)} < 8 \operatorname{sq}(M)^2$. 
\end{fact}

A combinatorial quantity that captures the learning complexity of a concept class is the \emph{Vapnik-Chervonenkis} (VC) dimension. Let $f: \mathcal X \times \mathcal Y\rightarrow \mathcal Z$ be a function whose communication complexity we are interested in. For any $x\in X$, let $f_x: \mathcal Y\rightarrow \mathcal Z$ be such that for every $y\in \mathcal Y$, $f_x(y) \coloneqq f(x, y)$. Define the set $f_{\mathcal X} = \{f_x\vert x\in \mathcal X\}$. Similarly, for $y\in \mathcal Y$, let $f_y(x)\coloneqq f(x, y)$ for all $x\in \mathcal X$ and $f_Y = \{f_y\vert y\in \mathcal Y\}$. In other words, the function class  $f_{\mathcal X}$ (resp. $f_{\mathcal Y}$), is defined by the set of rows (resp. columns) of the communication matrix associated with $f$. Conversely, given a concept/function class $f_{\mathcal X}$, an element $f_x\in f_{\mathcal X}$ is a feature vector of dimension $\vert \mathcal Y\vert$. This corresponds to the geometric representation of boolean matrices as points and half spaces in the context of learning theory.  The function $f(x, y)$ evaluates to 1 if and only if $f_x$ belongs to the half-space $\{h: \langle h, y\rangle >0\}$ and evaluates to 0 otherwise. This extends to the case of sign vectors and matrices~\cite{kremer1999randomized}. The VC dimension  of a matrix $M\in\{-1, 1\}^{n\times n}$, denoted as $VC(M)$, is the largest $d\in\mathbb Z_+$ such that $M$ features a $2^d\times d$ submatrix whose rows are the distinct elements of $\{-1, 1\}^d$. Results from~\cite{Sherstov2010} and~\cite{reyzin2020statistical} show that statistical query complexity can be upper and lower bounded by VC dimension. 

\begin{fact}[{\cite[Lemma~4.1] {Sherstov2010}}]\label{fact:sq-vc}
Let $\mathtt C$ be a concept class. Then, $\operatorname{sq}(\mathtt C)\leq 2^{O(VC(\mathtt C))}$. 
\end{fact}

\begin{fact}[{\cite[Observation~20] {reyzin2020statistical}}]\label{fact:vc-sq}
    For a concept class $\mathcal C$, $\operatorname{sq}(\mathtt C) = \Omega(VC(\mathtt C))$. 
\end{fact}

Yao's principle~\cite{Yao77} holds for one-way communication under product distribution. However, Sherstov~\cite{Sherstov2010} showed that the max relationship between $R^{ []}_\epsilon(f)$ and $D^{[]}_\epsilon(f)$ does not hold for the two-way model. In particular, they give a function $f:\{0, 1\}^n\times \{0, 1\}^n\rightarrow \{0, 1\}$ for which $\max_{[]} D^{[]}_\epsilon(f) = \Theta(1)$ but $R^{pub}_\epsilon (f) = \Theta(n)$. 

\begin{fact}[{\cite{yao1983lower},\cite[Theorem~2.2]{Kremer1999}}]\label{fact:RAB-DAB}
For every function $f: \mathcal X \times \mathcal Y \rightarrow \{0, 1\}$ and for every $0 < \epsilon < 1$, we have $R^{A\rightarrow B, []}_\epsilon(f) = \displaystyle\max_{[]} D^{A\rightarrow B, []}_\epsilon(f)$. 
\end{fact}

The two results show the equivalence between $VC(f)$, $R^{A\to B, []}_{\epsilon}(f)$, $\operatorname{sub}^{A\rightarrow B, []}_{\mathcal Y, \epsilon}(f)$ and $\operatorname{rec}^{A \rightarrow B, []}_\epsilon(f)$. 

\begin{fact}[{\cite[Theorem~3.2]{Kremer1999}}]\label{fact:vc-RAB}
For every function $f: \mathcal X \times \mathcal Y \rightarrow \{0, 1\}$ and for every $\epsilon \leq 1/8$ ,$R^{A\to B, []}_{1/3}(f)=\Theta(VC(f))$. 
\end{fact}

\begin{fact}[{\cite[Theorem~4.5]{Jain2008}}]\label{fact:RAB-subAB_prod}
Let $f \subseteq \mathcal X \times \mathcal Y \times \mathcal Z$ be a relation and let $0 \leq \epsilon \leq 1/6$. There are universal constants $c_1, c_2$ such that $\operatorname{sub}^{A \rightarrow B, []}_{\mathcal Y, \epsilon} (f), -1 \leq c_1 \cdot R^{A \rightarrow B, []}_\epsilon(f) \leq c_2 \cdot \left[\operatorname{sub}^{A\rightarrow B, []}_{\mathcal Y, \epsilon}(f) + \log \frac{1}{\epsilon} + 2\right]$. 
\end{fact}


Moreover, the distributional complexity under product distributions is upper-bounded by VC dimension in both the one-way and two-way communication model. 

\begin{fact}[{\cite[Theorem~3.2] {Kremer1999}},~{\cite[Theorem~3.1]{Sherstov2010}}]
Let $M$ be a sign matrix and let $0 < \epsilon \leq 1/3$. Then, $D^{A\rightarrow B, []}_\epsilon(M)\leq O\left(\frac{1}{\epsilon}VC(M)\log\frac{1}{\epsilon}\right)$. 
\end{fact}


The result below show that one-way quantum communication complexity under the product distribution can be lower bounded by the statistical query complexity and the rectangle bound. 

\begin{fact}[{~\cite[Fact~2.8]{klauck2001lower}}]\label{fact:sq-QAB}
    Let $\mathcal C\subseteq \{c: \{0, 1\}^n\rightarrow \{0, 1\}\}$. For every $\gamma>0$, we have $\Omega(\log(\gamma\cdot \sqrt{\operatorname{sq}(\mathtt C)}))\leq Q^{A\rightarrow B, []}_{1/2+\gamma}(\mathtt C)$. 
\end{fact}



Last but not least, we state results related to information complexity. 

\begin{fact}[{\cite[Proposition~3.12]{Braverman12}}]\label{fact:pi_int_ext_prod}
For any protocol $\pi$, $IC^{[]}(\pi) = IC^{\operatorname{ext}, []}(\pi)$. 
\end{fact}

\begin{fact}[\cite{Braverman12}]\label{fact:IC_int_ext}
Let $f: \mathcal X \times\mathcal Y \rightarrow \{0, 1\}$, then we have $IC^{\operatorname{ext}, []}_\epsilon(f) = IC^{[]}_\epsilon(f)$,   $IC^{\operatorname{ext}}_{D, \epsilon}(f) = IC_{D, \epsilon}(f)$, $IC^{\operatorname{ext}}_\epsilon(f) = IC_\epsilon(f)$. 
\end{fact}

\section{Separation}\label{sec:separation}
In this subsection, we justify the separation between different classes. This is done by either mentioning known results, or finding a function $f$ such that for a real function $g: \mathbb R \rightarrow \mathbb R$, complexity measures $\mathcal C$ from Class $i$ and $\mathcal C'$ from Class $i+1$, we have $\mathcal C(f) \geq g(n)$ and $\mathcal C'(f) = O(1)$, where $i\in \{1, 2, 3, 4\}$. 

\paragraph{Separation between Class 1 and Class 2}
Consider the Equality problem  defined as $EQ:\{0,1\}^n \times \{0,1\}^n \to \{0, 1\}$  such that for all $(x,y)\in\{0,1\}^n \times \{0,1\}^n$, 
\begin{align*}
    EQ(x,y) = \begin{cases} 1, \text{ if } x=y\\ 0, \text{ otherwise.}\end{cases}
\end{align*}
The Equality problem provides a polynomial separation between Class 1 and Class 2. More specifically, $D(EQ) = n+1$ by~{\cite[Example~1.21]{KNisan96}}. On the other hand, observe that $\gamma_2(EQ) = 1$ since the communication matrix $M_{EQ}$ of $EQ$ is the identity matrix $I$ and $\gamma_2(I)$=1.   

\begin{remark}
We note that separation in the other direction is not possible as deterministic communication complexity is always at least randomized communication complexity, even in the public coin and exact error case. 
\end{remark}

\paragraph{Separation between Class 2 and Class 3}
Cheung et al..~\cite{cheung2023separation} proved exponential separation between Class 2 and Class 3. However, to the best of our knowledge, an unbounded separation is not known to exist. 
\begin{fact}[Corollary 6 in\cite{cheung2023separation}]
    There is a Boolean function $f: \{0, 1\}^n\times \{0, 1\}^n\rightarrow \{0, 1\}$ with communication matrix $M_f$, such that $\gamma_2(M_f)\geq 2^{n/43}$ and $\gamma_2^\alpha(M_f)\leq O(\operatorname{poly}(n))$. 
\end{fact} 

 In the lemmas below, we prove the unbounded separation between Class 2 and Class 3. Define the Hamming Distance problem as $HD_1: \{0,1\}^n \times \{0,1\}^n \to \{0, 1\}$  such that for all $(x,y)\in\{0,1\}^n \times \{0,1\}^n$, 
$$HD_1(x,y) = \begin{cases} 1, \text{ if } d(x,y)= 1\\ 0, \text{ otherwise,}\end{cases}$$
where $d(x,y)$ denotes the Hamming distance between $x$ and $y$. We show the following lemma. 

\begin{lem}
$R^{\operatorname{pub}}_\epsilon(HD_1) = O(1)$. 
\end{lem}
\begin{proof} 
Assume that the following Equality Testing protocol costs $O(1)$ and has negligible error $\epsilon>0$: 

$$\operatorname{EQTEST}_\epsilon(x,y) = \begin{cases} 1 \text{, with probability  }1 - \epsilon \text{, if } x=y\\0 \text{, with probability  }1-\epsilon\text{, if} x\neq y.\end{cases}$$

The randomized public coin protocol with constant error for the $HD_1$ protocol goes as follows: 
\begin{enumerate}
    \item Perform $\operatorname{EQTEST}_\epsilon$ on $(x,y)$. 
    \item If $\operatorname{EQTEST}_\epsilon(x,y)=1$, output 0 for $HD_1$ and terminate the protocol. Otherwise, remove half of the entries of $x,y$ using a public random coin. Denote the shortened strings as $x', y'$ respectively. Perform $\operatorname{EQTEST}_\epsilon$ on $(x^\prime, y^\prime)$ and accept as per the Equality Test.
    \item Repeat Step 2 for $100$ times and accept if $\operatorname{EQTEST}$ outputs 0 for at least 33 times.
\end{enumerate}

Now, we analyze the correctness of the protocol. There are two cases: 
\begin{enumerate}[(i)]
    \item If $d(x,y)=0$, then $\operatorname{EQTEST}_\epsilon$ returns 1 with probability at least $1- \epsilon$. The protocol is terminated and we get $HD_1(x, y) = 0$  with probability at least $1 - \epsilon$. 
    \item If $d(x,y)\neq 0,$, 
    \begin{enumerate}[(a)]
        \item If  $d(x, y) = 1$, then $\operatorname{EQTEST}(x, y)=1$ with probability $\epsilon$. Moreover, $d(x^\prime, y^\prime)$ is either 0 or 1 with probability 1/2 each.  In the former case, $\operatorname{EQTEST}_\epsilon(x', y')=1$ with probability $1 - \epsilon$, while in the latter  case, with probability $\epsilon$. So the total probability that $\operatorname{EQTEST}$ outputs 1 on $x,y$ is at least
        \begin{align*}
            \epsilon + \frac{1}{2}\cdot \left(1 - \epsilon\right) + \frac{1}{2}\cdot \epsilon = \frac{1}{2} + \epsilon.
        \end{align*}
    \item If $d(x,y)=2$, $\operatorname{EQTEST}(x, y)=1$ with probability $\epsilon$. Furthermore, $d(x^\prime, y^\prime) = 0$ with probability 1/4 and $1\leq d(x^\prime, y^\prime) \leq 2$ with probability 3/4. In this case, the total probability that $\operatorname{EQTEST}$ outputs 1 on $x, y$ is at most
\begin{align*}
    \epsilon + \frac{1}{4}\cdot (1 - \epsilon) + \frac{3}{4}\cdot \epsilon  \approx \frac{1}{4} + \epsilon. 
\end{align*}

By a similar argument, if $d(x,y)=k$ for some $k>2$, $\operatorname{EQTEST}(x, y)=1$ with probability $\epsilon$. Also, $x'=y'$ with probability $1/2^k$ and hence the total probability that $\operatorname{EQTEST}$ outputs 1 on $x, y$ is approximately $\frac{1}{2^k}+\epsilon$. Hence, the error of the above protocols depends on the Hamming distance of $x,y$.

    \end{enumerate}
\end{enumerate}

This gives a standard randomised public coin protocol with constant communication.

\end{proof}

We are now interested in the lower bound on $\gamma_2(HD_1)$. Notice that $HD_1$ corresponds to a special case of the threshold function defined in the work of Hambardzymyan et al.~\cite{hambardzumyan2023dimension}. Namely, the threshold function is given by $\bar{\operatorname{thr}}_k: \{0, 1\}^n \rightarrow \{0, 1\}$ such that 
\begin{align*}
    \bar{\operatorname{thr}}_k(x) = 1 \quad \Leftrightarrow \quad \sum_{i=1}^n x_i < k . 
\end{align*}
$HD_1$ corresponds to the case when $k = 2$. The authors showed the following result. 
\begin{fact}[{\cite[Lemma~2.15]{hambardzumyan2023dimension}}]\label{fact:shown}
    $\gamma_2(HD_1) = \Theta(\sqrt n)$. 
\end{fact}
Nevertheless, we give a simpler proof for their result, avoids the heavy machinery of Fourier analysis. We consider the Boolean variant of the communication matrix and show the following lemma. 
\begin{lem}\label{lem:lb_HD_1}
    $\gamma_2(HD_1) \geq \Omega(\sqrt{n})$.  
\end{lem}
\begin{proof}
    Note that we have by an alternative definition of $\gamma_2$ that 
and the above observation that 
\begin{align*}
    \gamma_2(M_n) = \max_{u, v:\Vert u\Vert_2 = \Vert v\Vert_2 = 1} \Vert uv^T \circ M_n\Vert_{tr}. 
\end{align*}
By picking $u = v = (1/\sqrt {2^n}, \cdots, 1/\sqrt{2^n})^T\in\mathbb R^{2^n}$, we get by the structure of $M_n$, that 
\begin{align*}
        & \gamma_2(M_n) \\
        & \geq \left\Vert \frac{1}{2^n}J \circ M_n\right\Vert_{tr} \geq  \left\Vert \frac{1}{2^n}J \circ M_n\right\Vert_{F} = \sqrt{2^n\cdot n\cdot \frac{1}{2^n}} = \sqrt n, 
    \end{align*}
    where $\Vert\cdot\Vert_{tr}$ and $\Vert\cdot \Vert_F$ denotes the nuclear and Frobenius norm respectively, and $J$ is the all ones matrix. This completes the proof. 
\end{proof} 

For completeness, we give an upper bound on $\gamma_2(HD_1)$. 
\begin{lem}\label{lem:ub_HD_1}
    $\gamma_2(HD_1) \leq O(\sqrt{n})$.  
\end{lem}
\begin{proof}
    The communication matrix of $HD_1$ for strings length $n$ is a $2^n\times 2^n$ Boolean matrix $M_n$. When $n=1$, 
$$M_1 = \begin{bmatrix}
0 & 1\\
1 & 0\\
\end{bmatrix}.$$ For the case of $n=2$, 
$$M_2 = \begin{bmatrix}
0 & 1 & 1 & 0\\
1 & 0 & 0 & 1\\
1 & 0 & 0 & 1\\
0 & 1 & 1 & 0\\
\end{bmatrix}.$$ Hence, $M_n$ can be described recursively as follows: 
\[
\renewcommand\arraystretch{1.3}
M_n=\mleft[
\begin{array}{c |c}
  M_{n-1} & I \\
  \hline
  I & M_{n-1}\\
\end{array}
\mright] 
\]
   
By its trivial decomposition, we can write $M_n = I M_n $. Since every row/column of $M_n$ contains $n$ ones, we have 
\begin{align*}
    \gamma_2(M_n) = \gamma_2(I M_n) \leq \lVert I \rVert_{2\to\infty} \lVert M_n\rVert_{1\to 2}=1\cdot\sqrt{n},
\end{align*}
where $\Vert M_n\Vert_{p\rightarrow q}\coloneqq \sup_{v\neq 0}\frac{\Vert M_n v\Vert_q}{\Vert v\Vert_p}$. 
\end{proof}

Combining Lemmas~\ref{lem:lb_HD_1} and~\ref{lem:ub_HD_1} give $\gamma_2(HD_1) = \Theta(\sqrt n)$.  

\begin{remark}
Separation in the other direction is not possible since $\gamma_2(f) \geq \gamma_2^\alpha(f)$ for all $f$ and $\alpha > 1$ by Fact~\ref{fact:gamma2inf-gamma2alpha}. 
\end{remark}

\paragraph{Separation between Class 3 and Class 4}

Let $C\approx 2^{n/3}$. Alice gets the vector $x\in [-C, C]^3$ and Bob gets the vector $y\in [-C, C]^3$. Define 
\begin{align*}
    f(x, y) = \operatorname{sign}\langle x, y\rangle, 
\end{align*}
where $\operatorname{sign}:\mathbb R\rightarrow \{-1, 1\}$ is the sign function, mapping positive inputs to 1 and negative inputs to -1; and $\langle \cdot,\cdot\rangle$, denotes the inner product between two vectors. Hatami Hosseini and Lovett~\cite{hatami2020sign} showed that the above function provides separation between Class 3 and Class 4. 
\begin{fact}[{\cite[Theorem~4]{hatami2020sign}}]
    The function $f$ defined above satisfies  $\operatorname{signrank}(f) = 3$ and $\operatorname{disc}(f) = 2^{-\Omega(n)}$. In particular, $R^{\operatorname{pub}}_\epsilon(f) = \Omega(n)$. 
\end{fact}

\begin{remark}
Separation in the opposite direction is not possible since the unbounded-error communication complexity of a function $f: \mathcal X \times \mathcal Y \rightarrow \{0, 1\}$ is never much more than its complexity in the other standard models. For instance, $UPP(f) \leq O(R^{\operatorname{pub}}_{1/3}(f )+ \log \log (|\mathcal X| + |\mathcal Y|))$ and $UPP(f) \leq O(Q^*_{1/3}(f )+ \log \log (|\mathcal X| + |\mathcal Y|))$~{\cite[Section~2.1]{sherstov2011unbounded}}. 
\end{remark}

\paragraph{Separation between Class 4 and Class 5}
We first define the notion of a maximum class. Let $C \subseteq \{-1, 1\}^N$ be a class with VC dimension $d$. Then, $C$ is called a \emph{maximum class} if it meets the Sauer-Shelah's bound~\cite{sauer1972density} with equality, i.e. $\vert C\vert = \displaystyle\sum_{i=0}^d \binom{N}{i}$. 

Alon et al.~\cite{alon2016sign} showed a polynomial separation between $\operatorname{signrank}(f)$ and $VC(f)$~\cite{alon2016sign} via a maximum class of low VC dimension and high signrank. Specifically, let $P$ be the set of points in a projective plane of order $n$ and let $L$ be the set of lines in it. Let $N = N_{n, 2} = \vert P\vert = \vert L\vert$, where $N_{n, 2}\coloneqq n^2 + n  +1 = \frac{n^3 - 1}{n-1}$. For every $\ell\in L$, fix some linear order on the points in $\ell$. A set $T\subset P$ is called an \emph{interval} if $T\subseteq\ell$ for some line $\ell\in L$, and $T$ forms an interval with respect to the order fixed on $\ell$.  

\begin{fact}[Theorem 8 in \cite{alon2016sign}]
    The class $Q$ of all intervals is a maximum class of VC dimension 2. Moreover, there exists a choice of linear orders for the lines in $L$ such that the resulting $Q$ has signrank $\Omega(n^{1/2}/\log n)$. 
\end{fact}


\begin{remark}
Separation between Class 3 and 4 in the opposite direction is not possible since $\operatorname{signrank}(f) \geq VC(f)$~\cite{alon2016sign}.   
\end{remark}

\section{Conclusion and open problems}\label{sec:conclusion}
We outline a hierarchy for constant communication complexity. by categorizing communication complexity measures into five classes. In each class, one complexity measure is related by a function of another complexity measure, independent of the input size. We show separation between consecutive classes. Such separations are only meaningful in one direction. 


Our classification is mainly based on known results in the literature. Although we have made every attempt to comprehensively include all relevant complexity measures, some measures may have been inadvertently omitted. There is also a small number of measures that we are unable to classify. For example, it is natural for one to include $QMA(f)$ in Class 3 since both $MA(f)$ and $Q^*_\epsilon(f)$ are also in that class. In particular, having $MA(f) \leq R^{\operatorname{pub}}_\epsilon(f)$, $QMA(f) \leq Q^*_\epsilon(f)$ should hold analogously. In his paper, Klauck proved that $QMA(f)$ is lower bounded by the one-sided smooth discrepancy, $\operatorname{sdisc}^1_{\delta}(f)$~{\cite[Theorem~2]{klauck2011arthur}}. To the best of our knowledge, no lower bound on $\operatorname{sdisc}^1_{\delta}(f)$ is known. If one could prove a lower bound on  $\operatorname{sdisc}^1_{\delta}(f)$ in terms of any complexity measure in Class 3 and independent of the input size, then both $QMA(f)$ and $\operatorname{sdisc}^1_{\delta}(f)$ can be added to Class 3. 

\begin{openproblem}
Let $f: \mathcal X \times \mathcal Y \rightarrow \{0, 1\}$ be a Boolean function. Which of the complexity measures in Class 3 lower bounds $\operatorname{sdisc}^1_{\delta}(f)$, without dependency on the input size? 
\end{openproblem}

\begin{conjecture}
$QMA(f)$ and $\operatorname{sdisc}^1_{\delta}(f)$ belong to Class 3. 
\end{conjecture}

On the other hand, it is surprising to see that $AM$ and $MA$ complexities belong to different classes despite both models being almost similar. This is due to the fact that $AM_\epsilon(f) \leq O(SBP(f) + \log n)$~\cite{goos2016communication}. 

Other unclassified measures include the (two-way) Las Vegas communication complexity, $LV(f)$.
Jain and Klauck introduced the Las Vegas partition bound and proved that it is at most exponential in $LV(f)$~{\cite[Lemma~7]{JK10}}. {However, it is not clear if $LV(f)$ is lower bounded by $R^{\operatorname{pub}}_\epsilon(f)$ in Class 3 or upper bounded by $D(f)$ in Class 1.

With regards to information complexity, Braverman showed that $IC_0(f) = IC_{D, 0}(f)$~{\cite[Theorem~3.6]{braverman2012interactive}} and $IC^{\operatorname{ext}}_0(f) = IC^{\operatorname{ext}}_{D,0}(f)$~{\cite[Theorem,~3.17]{braverman2012interactive}}. Since internal information complexity is upper bounded by external information complexity, which is always upper bounded by communication complexity, this implies that $IC_0(f)$ and $IC_{D,0}(f)$ are respectively upper bounded by $IC^{\operatorname{ext}}_0(f)$ and $IC^{\operatorname{ext}}_{D,0}(f)$, which are both upper bounded by $R^{\operatorname{pub}}_0(f)$. As there are limited results on exact information complexity, it is not clear whether  $IC_0(f)$ can be lower bounded by any of the elements in Class 1, without dependency on the input size.

\begin{openproblem}
Let $f: \mathcal X \times \mathcal Y \rightarrow \{0, 1\}$ be a Boolean function. Which of the complexity measures in Class 1 lower bounds $IC_0(f)$, without dependency on the input size? 
\end{openproblem}

\begin{conjecture}\label{conj:zero_error_IC}
$IC_0(f), IC_{D,0}(f), IC^{\operatorname{ext}}_0(f)$ and $IC^{\operatorname{ext}}_{D,0}(f)$ belong to Class 1. 
\end{conjecture}

Furthermore, the equivalence between the internal zero-error information complexity of a function $f$ and its amortized communication complexity, does not hold in the case of amortized zero-error communication complexity. In fact, it was conjectured that the (average case) amortized zero-error communication complexity of a function is exactly characterized by its external information complexity~\cite{braverman2013information, braverman2012coding}. If this conjecture is proven to be true, then the average case amortized zero-error communication complexity belongs to the same class as the complexity measures mentioned in Conjecture~\ref{conj:zero_error_IC}. 

Touchette showed that $AQCC^{\mu}_\epsilon(f) = QIC^{\mu}_\epsilon(f)$ for any function $f$ and error $\epsilon  > 0$~{\cite[Theorem~2]{touchette2015quantum}}. While one can draw the relation $Q^{*, \mu}_\epsilon(f) \rightarrow QIC^{\mu}_\epsilon(f) \rightarrow AQCC^{\mu}_\epsilon(f) \rightarrow AC^{\mu}_\epsilon(f) \rightarrow IC^{\mu}_\epsilon(f)$ in Class 3, the lower bound on $Q^{*, \mu}_\epsilon(f)$ is not known. In fact, this was posed as an open problem in the work of Jain and Zhang~\cite{Jain2009}. 

\begin{openproblem}
Let $f: \mathcal X \times \mathcal Y \rightarrow \{0, 1\}$ be a Boolean function. Which of the complexity measures in Class 3 lower bounds $Q^{*, \mu}_\epsilon(f)$, without dependency on the input size?    
\end{openproblem}

\begin{conjecture}
$Q^{*, \mu}_\epsilon(f), QIC^{\mu}_\epsilon(f)$ and $AQCC^{\mu}_\epsilon(f)$ belong to Class 3.
\end{conjecture}

Under product distribution over inputs, we have a similar open problem and conjecture.  
\begin{conjecture}
Let $f: \mathcal X \times \mathcal Y \rightarrow \{0, 1\}$ be a Boolean function. Which of the complexity measures in Class 5 lower bounds $Q^{*, []}_\epsilon(f)$, without dependency on the input size?     
\end{conjecture}

\begin{conjecture}
$Q^{*, []}_\epsilon(f), QIC^{[]}_\epsilon(f)$ and $AQCC^{[]}_\epsilon(f)$ belong to Class 5.
\end{conjecture}

\section{Acknowledgments}
DL specially thanks Alexander Belov, Srijita Kundu, Jo\~{a}o F. Doriguello, Dmitry Gavinsky and Uma Girish for helpful discussions.  
AA and DL acknowledge funding from the Latvian Quantum Initiative under EU Recovery and Resilience Facility under project no.\ 2.3.1.1.i.0/1/22/I/CFLA/001. DL and HK thanks funding from the Singapore Ministry of Education (partly through the Academic Research
Fund Tier 3 MOE2012-T3-1-009) and by the Singapore National Research Foundation. DL and HK is also  supported by Majulab UMI 3654.





\bibliographystyle{elsarticle-num-names}
\bibliography{references}{}

\end{document}